%% file: ncmdp27Arx.tex
\newif\ifdouble
\newif\ifArX

\doublefalse
\ArXtrue
\ifdouble
\documentclass[journal]{IEEEtran}
\else
\documentclass[12pt, draftclsnofoot, onecolumn]{IEEEtran}
\fi

\usepackage{setspace}
\usepackage{comment}
\usepackage{amsmath, latexsym, amssymb, amsthm, amsfonts}
\usepackage[mathscr]{eucal}
\usepackage{graphicx}
\usepackage[usenames,dvipsnames]{color}
\usepackage{epsfig}
\usepackage{amsfonts}
\usepackage{psfrag}

\usepackage{wrapfig}

\usepackage[update,prepend]{epstopdf}

\usepackage{color}

\DeclareMathOperator*{\argmax}{arg\,max}

\begin{document}

\newtheorem{lemma}{Lemma}
\newtheorem{claim}{Claim}
\newtheorem{theorem}{Theorem}
\newtheorem{corollary}{Corollary}
\newtheorem{definition}{Definition}
\newtheorem{example}{Example}
\newtheorem{proposition}{Proposition}
\newtheorem{condition}{Condition}
\newtheorem{assumption}{Assumption}
\newtheorem{conjecture}{Conjecture}
\newtheorem{problem}{Problem}
\newtheorem{remark}{Remark}

\def\thelemma{\arabic{section}.\arabic{lemma}}
\def\thetheorem{\arabic{section}.\arabic{theorem}}
\def\thecorollary{\arabic{section}.\arabic{corollary}}
\def\thedefinition{\arabic{section}.\arabic{definition}}
\def\theexample{\arabic{section}.\arabic{example}}
\def\theproposition{\arabic{section}.\arabic{proposition}}
\def\thecondition{\arabic{section}.\arabic{condition}}
\def\theassumption{\arabic{section}.\arabic{assumption}}
\def\theconjecture{\arabic{section}.\arabic{conjecture}}
\def\theproblem{\arabic{section}.\arabic{problem}}
\def\theremark{\arabic{section}.\arabic{remark}}

\newcommand{\manualnames}[1]{
\def\thelemma{#1.\arabic{lemma}}
\def\thetheorem{#1.\arabic{theorem}}
\def\thecorollary{#1.\arabic{corollary}}
\def\thedefinition{#1.\arabic{definition}}
\def\theexample{#1.\arabic{example}}
\def\theproposition{#1.\arabic{proposition}}
\def\theassumption{#1.\arabic{assumption}}
\def\theremark{#1.\arabic{remark}}
}

\newcommand{\beginsec}{
\setcounter{lemma}{0}
\setcounter{theorem}{0}
\setcounter{corollary}{0}
\setcounter{definition}{0}
\setcounter{example}{0}
\setcounter{proposition}{0}
\setcounter{condition}{0}
\setcounter{assumption}{0}
\setcounter{conjecture}{0}
\setcounter{problem}{0}
\setcounter{remark}{0}
}
\newcommand{\la}{\lambda}
\newcommand{\eps}{\varepsilon}
\newcommand{\ph}{\varphi}
\newcommand{\vr}{\varrho}
\newcommand{\al}{\alpha}
\newcommand{\bet}{\beta}
\newcommand{\gam}{\gamma}
\newcommand{\kap}{\kappa}
\newcommand{\s}{\sigma}
\newcommand{\sig}{\sigma}
\newcommand{\om}{\omega}
\newcommand{\Gam}{\mathnormal{\Gamma}}
\newcommand{\off}[1]{}
\newcommand{\Del}{\mathnormal{\Delta}}
\newcommand{\Th}{\mathnormal{\Theta}}
\newcommand{\La}{\mathnormal{\Lambda}}
\newcommand{\X}{\mathnormal{\Xi}}
\newcommand{\PI}{\mathnormal{\Pi}}
\newcommand{\Sig}{\mathnormal{\Sigma}}
\newcommand{\Ups}{\mathnormal{\Upsilon}}
\newcommand{\Ph}{\mathnormal{\Phi}}
\newcommand{\Ps}{\mathnormal{\Psi}}
\newcommand{\Om}{\mathnormal{\Omega}}

\newcommand{\D}{{\mathbb D}}
\newcommand{\M}{{\mathbb M}}
\newcommand{\N}{{\mathbb N}}
\newcommand{\Q}{{\mathbb Q}}
\newcommand{\R}{{\mathbb R}}
\newcommand{\U}{{\mathbb U}}
\newcommand{\Z}{{\mathbb Z}}
\newcommand{\T}{{\mathbb T}}

\newcommand{\EE}{{\mathbb E}}
\newcommand{\FF}{{\mathbb F}}
\newcommand{\PP}{{\mathbb P}}
\newcommand{\ONE}{\boldsymbol{1}}

\newcommand{\calA}{{\cal A}}
\newcommand{\calB}{{\cal B}}
\newcommand{\calC}{{\cal C}}
\newcommand{\calD}{{\cal D}}
\newcommand{\calE}{{\cal E}}
\newcommand{\calF}{{\cal F}}
\newcommand{\calG}{{\cal G}}
\newcommand{\calH}{{\cal H}}
\newcommand{\calI}{{\cal I}}
\newcommand{\calJ}{{\cal J}}
\newcommand{\calL}{{\cal L}}
\newcommand{\calM}{{\cal M}}
\newcommand{\calN}{{\cal N}}
\newcommand{\calP}{{\cal P}}
\newcommand{\calR}{{\cal R}}
\newcommand{\calS}{{\cal S}}
\newcommand{\calT}{{\cal T}}
\newcommand{\calU}{{\cal U}}
\newcommand{\calV}{{\cal V}}
\newcommand{\calX}{{\cal X}}
\newcommand{\calY}{{\cal Y}}

\newcommand{\bI}{{\mathbf I}}
\newcommand{\bJ}{{\mathbf J}}
\newcommand{\bK}{{\mathbf K}}
\newcommand{\bT}{{\mathbf T}}

\newcommand{\scrA}{\mathscr{A}}
\newcommand{\scrM}{\mathscr{M}}
\newcommand{\scrS}{\mathscr{S}}
\newcommand{\scrU}{\mathscr{U}}
\newcommand{\scrI}{\mathscr{I}}
\newcommand{\scrP}{\mathscr{P}}
\newcommand{\scrR}{\mathscr{R}}

\newcommand{\frA}{\mathfrak{A}}
\newcommand{\frM}{\mathfrak{M}}
\newcommand{\frS}{\mathfrak{S}}

\renewcommand{\proof}{\noindent{\bf Proof.\ }}

\newcommand{\lan}{\langle}
\newcommand{\ran}{\rangle}
\newcommand{\uu}{\underline}
\newcommand{\oo}{\overline}
\newcommand{\skp}{\vspace{\baselineskip}}
\newcommand{\supp}{{\rm supp}}
\newcommand{\diag}{{\rm diag}}
\newcommand{\trace}{{\rm trace}}
\newcommand{\w}{\wedge}
\newcommand{\lt}{\left}
\newcommand{\rt}{\right}
\newcommand{\pl}{\partial}
\newcommand{\abs}[1]{\lvert#1\rvert}
\newcommand{\norm}[1]{\lVert#1\rVert}
\newcommand{\mean}[1]{\langle#1\rangle}
\newcommand{\To}{\Rightarrow}
\newcommand{\til}{\widetilde}
\newcommand{\wh}{\widehat}
\newcommand{\dist}{{\rm dist}}
\newcommand{\grad}{\nabla}
\newcommand{\iy}{\infty}
\newcommand{\AddedTh}[1]{\textbf{\textcolor{Black}{#1}}}
\newcommand{\shifrin}[1]{{#1}}
\newcommand{\FMO}[1]{\textbf{\textcolor{red}{***FIXME OMER #1 ***}}}

\newcommand{\be}{\begin{equation}}
\newcommand{\ee}{\end{equation}}

\newcommand{\tab}{\hspace*{0.3in}}
\newcommand{\Tab}{\hspace*{1.0in}}
\newcommand{\no}{\nonumber}
\newcommand{\noi}{\noindent}
\newcommand{\txt}{\textrm}
\newcommand{\ds}{\displaystyle}
\newcommand{\RR}{\mathbb{R}}
\newcommand{\vf}{\varphi}
\newcommand{\del}{\frac{\partial}{\partial t}}
%\definecolor{co}{rgb}{0.8,0,0.8}
%\definecolor{gr}{gray}{0.5}
\newcommand{\gr}{\color{gr}}
\newcommand{\vp}{\varepsilon}
\newcommand{\E}{{\mathbb E}}
\newcommand{\MS}[1]{{{#1}}}
\newcommand{\Omer}[1]{{\textcolor{red}{#1}}}
\newcommand{\Asaf}[1]{{\textcolor{red}{#1}}}
\newcommand{\Olga}[1]{{\textcolor{red}{#1}}}

\title{Coded Retransmission in Wireless Networks Via Abstract MDPs: Theory and Algorithms\thanks{Parts of this work will appear at the IEEE International Symposium on Information Theory, ISIT 2015, Hong Kong.}}
 % \#1570010543}

\author{\IEEEauthorblockN{{\small{Mark Shifrin, Asaf Cohen, Olga Weisman, Omer Gurewitz}}}\\
\IEEEauthorblockA{{\small{Department of Communication Systems Engineering\\
Ben-Gurion University of the Negev, Beer-Sheva, 84105, Israel\\
Email: {shifrin@tx.technion.ac.il,coasaf@bgu.ac.il,weismano@post.bgu.ac.il,gurewitz@bgu.ac.il }}}}}

\maketitle
%--------------- table of contents ----------------------------------------
%
\begin{abstract}
%\boldmath
Consider a transmission scheme with a single transmitter and multiple receivers over a faulty broadcast channel. For each receiver, the transmitter has a unique infinite stream of packets, and its goal is to deliver them at the highest throughput possible. While such \emph{multiple-unicast} models are unsolved in general, several \emph{network coding based schemes} were suggested. In such schemes, the transmitter can either send an uncoded packet, or a coded packet which is a function of a few packets. \MS{The packets sent can be received by the designated receiver} (with some probability) or heard and stored by other receivers. Two functional modes are considered; the first presumes that the storage time is unlimited, while in the second it is limited by a given Time to Expire (TTE) parameter.

We model the transmission process as an infinite-horizon Markov Decision Process (MDP). Since the large state space renders exact solutions computationally impractical, we introduce \emph{policy restricted} and \emph{induced} MDPs with significantly reduced state space, and prove that with proper reward function they have equal optimal value function (hence equal optimal throughput). We then derive a reinforcement learning algorithm, which learns the optimal policy for the induced MDP. This optimal strategy of the induced MDP, once applied to the policy restricted one, significantly improves over uncoded schemes. Next, we enhance the algorithm by means of analysis of the structural properties of the resulting reward functional. We demonstrate that our method scales well in the number of users, and automatically adapts to the packet loss rates, unknown in advance. In addition, the performance is compared to the recent bound by Wang, which assumes much stronger coding (e.g., intra-session and buffering of coded packets), yet is shown to be comparable.
\end{abstract}

\input{introduction3}

\input{RelatedWork3}

%\input{ICcompare}
\input{Model1_O}

%\input{Induced_MDP4}
\input{Induced_MDP5_O}
%\input{State_Aggregation3}
\input{State_Aggregation1_O}

\input{V_Study3}
\input{Simulation2}

%\input{discussion}
\input{Appndx4}
%\input{Conclusions}
%\bibliographystyle{alpha}
%\subsection*{Acknowledgments}
\bibliographystyle{IEEEtran}
\bibliography{ncmdp91,ncmdp_review}

%\bibliography{ncmdp_review}
%\bibliography{ncmdp-I}
% that's all folks
\end{document}

%% file: introduction3.tex
%%%%%%%%%%%%%%%%%%%%%%%%%%%%%%%%%%%%%%%%%%%%%%%%%%%%%%%%%%%
\section{Introduction}
%%%%%%%%%%%%%%%%%%%%%%%%%%%%%%%%%%%%%%%%%%%%%%%%%%%%%%%%%%%

Typical wireless access architectures constitute a gateway, or an Access Point (AP), to which all nearby clients are connected by means of a wireless medium. Among the prominent examples for such architecture is the prevailing IEEE 802.11 or LTE infrastructure mode setting. The downlink traffic implied in such topology comprises an AP sending (usually independent) traffic streams to the corresponding users. Furthermore, common wireless standards incorporate reliability mechanisms in order to overcome the inherently poor qualities of the radio channel. For example, IEEE 802.11, like many other network protocols, attains reliability through retransmission.

Network coding \cite{ahlswede2000network} refers to the transmission of predefined functions (usually a linear combination) of packets in order to achieve higher throughput, error correction and better security. Wireless communication, and in particular the transmission over the wireless channel which is broadcast in nature hence can potentially be heard by non-addressees of the dedicated stream is a natural platform for network coding. Nonetheless, in order for such a mechanism to be effective, the overhearing users need to store the relevant parts of the traffic streams even when they are not the intended addressee.
%\MS{%Among other known tools and protocols which aim to strengthen the reliability of communication we mention NORM which supports reliable multicast and unicast via a NACK-based ARQ protocol using forward error correction (FEC)-based transmissions.
%In contrary to this protocol, we address multiple unicast communication, based on \textit{smart retransmissions}.
%\off{Note that FEC-based techniques, e.g. Reed-Solomon error correction code can be effectively combined with our methods. Namely, each packet can be encoded by FEC and next transmitted or retransmitted by means of NC.}}

In this work, we address the aforementioned scenario of a single AP sending unicast streams to $K$ corresponding listeners. We assume that all streams are fully backlogged, i.e., there is a packet pending for each receiver at all times (infinite horizon). We also assume a typical stop-and-wait ARQ (automatic repeat-request) mechanism, similar to the one adopted by IEEE 802.11 standard. In such schemes, a sender sends one frame at a time, where each frame is sent repeatedly until the sender receives an acknowledgment (ACK) frame from the receiver. That is, the next packet to some user will be transmitted only after the previous packet to that user was received correctly. We adopt the decoding and data storage pattern known in literature as instantly decodable network coding~\cite{sorour2012densifying}, specifically, each user stores packets even if not destined to it, yet \textit{only uncoded packets are stored at the receivers while coded combinations are discarded}. We assume that the data stored at the listeners is known to the AP at all times; this can be achieved by each receiver piggybacking a list of its current stored packets not destined to it, on the user's upstream traffic (each DATA or ACK sent by the user to the AP).

Using network coding at the AP, the challenge in each downstream transmission to is determine whether to send an ordinary unicast packet to one of the intended receivers, or to send a linear combinations of packets. Note that even under this seemingly moderate setup, in which users store only uncoded packets, and the AP has at most a single packet pending per user at a time, since each user can potentially store a packet for each other user (i.e., $2^{K-1}$ possibilities per user, \MS{where $K$ stands for the number of users}), the number of different options for stored packets before each transmission-opportunity (termed the \textit{state space}) is enormous ($2^{(K-1)\cdot K}$). Consequently, no efficient solution optimal in the general case exists~\cite{dougherty2006nonreversibility}.

In this paper, we design a computationally feasible, scalable and robust methodology which effectively addresses the aforementioned problem. Furthermore, in addition to the generic problem described above, we also consider a more complicated setup \textit{in which the storage time of packets at the receivers is limited} by a Time to Expire (TTE) constraint, i.e., a packet that its storage time has expired, is invalidated and discarded. We present a theoretical framework and a model-based learning implementation which allow us to acquire the on-line transmission and retransmission policy under such channel conditions. In particular, we address three specific challenges. First, the fundamental challenge of network coding - deciding what is the most effective linear combination of the data to be transmitted. This problem becomes further complicated, once TTE constraints are introduced. Second, in contrast to most known works, our model presumes \textit{infinite data streams} for all listeners, rather than limiting the amount of data to a fixed block. Finally, the encoding decisions are made in an environment without prior knowledge of the packet loss probability. As we elaborate in the related work section, previous works in the area mainly considered various optimization problems for multicast transmissions and/or finite horizons (finite block length). However, this is the first work to address all these challenges in a unified framework.

Our main contributions are as follows: we model the transmission process by a Markov decision processes (MDP). Since the original state space is intractable, we utilize state aggregation. State aggregation (sometimes referred to as state abstraction) is a technique to partition the state space such that all states belonging to the same partition subset are aggregated into one meta-state, such that the same policy applies to all states in the meta-state. %The objective is that the aggregate MDP will retain the same policy values and optimal policies as the original MDP.
In contrast to a complex exhaustive search to find the optimal aggregation, we \textit{force} a state aggregation, based on proved coding concepts. %sacrificing optimality for computability. Within this framework, we introduce definitions of
We further introduce a policy restricted MDP and an induced MDP which undergoes a dramatic state space reduction, and show that in case one chooses the appropriate reward function for the induced MDP, the overall reward of both processes will be equal. Specifically, instead of keeping track of all possible packets (coded and uncoded), we only keep track of two state variables: (i) The size of the maximal group of users in which each member of the group has a packet destined to each other user in the group but its own (i.e., maximal clique; accordingly, in the sequel we will refer to any set of users each having packets of all other users as a clique, and the maximal such set the maximal clique).
Note that for each clique, a single coded packet which linearly combines all the packets destined to the users in the clique can be sent, and each user receiving the coded packet can decode its own packet. %Furthermore, note that transmitting a coded packet to the maximal clique maximizes the instantaneous transmission rate, yet does not necessarily maximize the long term overall transmission rate.
(ii) The number of users whose packets are not stored by any other user.  %(usually users that their pending packets have never been sent before, typically because they have received their last intended transmission successfully).
Note that this abstraction allows us to significantly reduce the state space from $O(2^{K^2})$ to $O(K^2)$.
Consequently, we also restrict the action space, such that the only allowed actions are transmitting a packet to one of the users currently not having its packet backlogged at any other user, or transmitting a coded packet to the maximal clique. Hence, we name the MDP which only allows restricted actions based on the aggregation a \emph{policy restricted} MDP, and the MDP which sees only aggregated states an \emph{induced} MDP.

Given the transition probabilities, %In order to find the transition probabilities between the meta-states we use a model based learning algorithm. In particular, we utilize reinforcement learning with an $\epsilon-greedy$ policy. %***FIXME why $\epsilon-greedy$?***.
the optimal policy can be read off the Bellman equation for the induced MDP,
which has a relatively small state space and thus can be efficiently solved.
However, since the transition probabilities are hard to calculate, we \textit{learn} them using a model-based learning algorithm.
Namely, we derive a novel on-line explore and exploit learning algorithm, which iterates between the learning phase and the Bellman equation solution phase in our problem. Hence, we achieve the \emph{optimal policy}, which, in turn, results in the optimal throughput (under the constraints imposed by the aggregation and state reduction). Note that this approach is independent of the channel conditions, and works equally effectively for any packet loss, including when the packet loss is not stable and fluctuates around some value. We also study the \emph{structural properties} of the value function, and use these properties to both gain deep understanding on the behavior of optimal policies and accelerate the reinforcement learning (RL) procedure. Specifically, we prove that under mild conditions, there exists a "threshold type policy", namely as a function of the maximal clique size, there is only one transition from one optimal action to the other, and once sending a clique is optimal, it continues to be optimal for the larger cliques. %states(i.e. having larger maximal clique size).
We show that our algorithm is both computationally tractable and scalable. At the same time, its performance is comparable to the upper bounds in~\cite{wang2012capacity}, which are given for a \emph{much stronger coding scheme}, including intra-session coding, much larger state space and buffers, and no TTE.%Moreover, we compare the value functions found by our aggregation algorithm, to that found by exhaustive Q-learning, applied to the simple cases, and conclude that its convergence is within close range of the optimum.

%Next, we extend our model to cases in which the data has aging constraints and in particular each packet becomes obsolete within a fixed time period. Accordingly, each stored packet has a TTE counter, which is decremented every time slot. Once the counter expires, the packet is dropped.
We incorporate the TTE constraint within the aforementioned MDP model and propose two types of state aggregations. We compare our algorithms with known algorithms in the literature via extensive simulations.
%Finally, we state that other types of aggregations are possible, while using the same concept as we present in this paper.

%% file: RelatedWork3.tex
\subsection{Related work}
\noindent\textbf{Network coding.}
While the problem of NC has been widely treated in the multicast setting, multiple unicast still provides a rich ground for ongoing research. Coded retransmissions were considered in~\cite{sorour2011adaptive}, where, after sending a finite set of packets to all users and receiving acknowledgements, coded retransmissions are calculated and sent in order to complete the missing packets. Hence, this is a \emph{finite horizon} problem, where a block is sent only when the previous one is completely decoded.~\cite{sorour2012densifying} continued the above work, seeking to maximize the coding opportunities. Similar to our problem, in ~\cite{sorour2012densifying} users cannot store coded packets. However, ~\cite{sorour2012densifying} fits a \emph{multicast scenario} rather than multiple unicast. Moreover, the graph required to identify cliques in ~\cite{sorour2012densifying} grows with the stream size, while it is fixed in our scheme. Finite streams and clique structures were also addressed in~\cite{nguyen2009wireless}. Additional strategies for finite streams can be found in \cite{keller2008online},~\cite{xiao2008wireless} and~\cite{costa2008informed}.

In~\cite{eryilmaz2006delay}, the objective was to \emph{minimize the delay} using random linear NC. Random NC was also applied for mesh networks in~\cite{lin2010slideor}. The finite horizon work~\cite{parastoo2010optimal} minimized the delay by linear programming. Network coding for multi-hop wireless network was addressed in~\cite{rayanchu2008loss}. To the best of our knowledge, no previous work analytically treated the setting \MS{where the storage time of the side information was limited by some parameter (TTE)}. Practical insights on storage time constraints and imperfect acknowledge delivery are given in~\cite{katti2006xors}. We also mention the MDP based approach for perfect feedback~\cite{nguyen2007multimedia} and partially observable MDP for uncertain feedback~\cite{nguyen2009network}. Both works, however, are for finite horizon and do not include state aggregation. Thus, the problem of scalability of the solutions with the size of the stream is raised.

Recently, the seminal work in~\cite{wang2012capacity} gave codes and bounds for the erasure broadcast channel. The coding strategy therein was proved optimal for up to 3 users, and bounds were given for general $K$ (two users were considered earlier in~\cite{georgiadis2009broadcast}). The coding scheme therein assumed more than one packet per user can be coded and overheard (intra-session coding), while we only allow transmitting the first packet per session. Furthermore, the model in~\cite{wang2012capacity} allows storing coded packets, at the price of larger buffers and state space, while our model assumes instantly decodable codes. Nevertheless, we use the theoretical upper bound in~\cite{wang2012capacity} to evaluate the performance of the schemes suggested herein, and find them comparable despite the much simpler coding in this work. Note also that calculating the regions in~\cite{wang2012capacity} is exponentially complex in $K$, while the algorithms suggested herein scale well with the number of users.

To conclude, none of the aforementioned works addressed the problem of multiple unicast with infinite horizon addressed in this paper. Reference~\cite{cohen2013coded} attempted to provide heuristic algorithms for a small number of users, yet the algorithms therein show inferior performance compared to the learning-based solutions suggested in this work. In addition,~\cite{cohen2013coded} did not consider the channel condition, while our approach is adjustable to the packet loss uncertainty.

Random linear network coding (RLNC), (e.g.,~\cite{dougherty2005insufficiency}) is used only across flows (only inter-flow coding), then, regardless of the filed size used, such a coding scheme will effectively require \emph{all receivers  to decode all the data}, which is highly inefficient. Increasing the field size will only increase the probability that a sent packet is independent of the previously sent ones, but would still require each receiver to wait for a full rank on all the data in the system. Moreover, RLNC requires receivers to cache coded packets as well. Indeed, it is well known in the coding literature that RLNC is optimal for multicast (all receivers requiring all the information), yet highly inefficient for multiple unicast, which is the problem at hand.

Finally, note that the Wang's bound discussed and depicted in section~\ref{sec:Impl}, allows for the most general coding schemes, including larger window size, buffering of coded packets, intra-flow coding and high field sizes. Thus, our results are compared to the most general (and computationally expensive) coding scheme, and show good performance.

\input{ICcompare2}

\noindent\textbf{State aggregation.}
As a road-map paper for the state aggregation methods see \cite{li2006towards}. This work defined $5$ abstraction methods, where the most relevant to our setting is $\pi^*$-abstraction. %This is the most coarse abstraction.
We partially adopt their definitions of aggregated and detailed (ground) states and the corresponding abstraction function.
%$\pi^*$-abstraction is the most coarse abstraction
%As stated in~\cite{jong2005state}, and exemplified in~\cite{gordon1996chattering},
$\pi^*$-abstraction can be suboptimal compared to the original MDP~\cite{jong2005state}.
However, our approach is different from~\cite{li2006towards}, since we do not attempt to perform a search to find the aggregation which would preserve optimality, but rather, based on key principles in coding and re-transmission, define a robust MDP abstraction, in order to acquire the smallest states space and action space.
\MS{An adaptive aggregation for the average reward MDP was presented in~\cite{ortner2013adaptive}. %The author defines $\epsilon$-aggregation, where states with rewards and transition probabilities which lie within certain distance given by predefined constant $\epsilon$ are candidates for aggregation. Using this notation the bounded parameter MDP (BPMDP) is constructed. We do not use such a criterion; however such constant clearly exists for aggregations proposed here as well.
In this work, the aggregation is generic and partition into aggregated states is being updated in the process of the algorithm run. %Specifically, their algorithm strives to achieve the \textit{neat aggregation}, in order to preserve the optimality of the original MDP.
However, it is not clear how to predict the number of states in such an aggregation once the algorithm achieved the desired optimality bound. Our aggregation is fixed and predefined in order to \textit{specifically suit for the given communication problem}. %Moreover, the aggregation process in~\cite{ortner2013adaptive} involves a computationally complex procedure of finding \textit{all} cliques in the states graph. The computation complexity of this part in our method depends on the nature of the aggregation type, e.g. finding a maximum clique (see simple aggregation in section~\ref{sec:st-agg}).  The calculation complexity of the optimal policy for the BCMDP is done by extended value iteration, as in~\cite{auer2009near}, and depends on the number of states of the \textit{original} MDP. In contrary, the value iteration we use only depends on the state space of the induced MDP, which is significantly smaller.
Hence, both the aggregation and the state-space size we employ are predefined and result in a much simpler RL algorithm, at expense of optimality guarantees. }
%That is, our primary goal is reduction of the state space in order to gain computability, while preserving optimality as much as possible.}
\MS{Another survey work on abstraction, in the context of reinforcement learning is~\cite{ponsen2010abstraction}. State aggregation for continuous MDP is brought in~\cite{ferns2011bisimulation}.
%For methods to minimize the MDP state-space see, e.g.,~\cite{dean1997model},\cite{jong2005state},\cite{givan2003equivalence}.
The authors in~\cite{kearns2002sparse} proposed a near-optimal reinforcement learning algorithm aiming to asymptotically achieve the optimality of the original MDP. However, running time demands needed to achieve the desired optimality gap are not feasible for our purpose.}
 %They compared two ARQ schemes, in the first one ACKS are only sent upon decoding a packet, while in the second ACKS are sent upon merely receiving a packet. %They show that the second scheme performs better.

%% file: ICcompare2.tex
\noindent\textbf{Index Coding and ARQ.}\label{sec:ICc}
The relation between NC and Index Coding (IC)~\cite{birk2006coding} was formulated in~\cite{el2010index}. The most general formulation of the IC problem constitutes a setting of K nodes, each having a set of packets as side information and expecting an optionally distinct set of packets. At the beginning of the communication, all the data is at the base station, and the goal is to find a transmission strategy to satisfy all demands. Therefore, this is, in essence, a finite horizon problem. Of course, similar to previous works, IC, in general, allows for complex coding over all packets in the block and storing of coded packets at the receivers before decoding. In addition, reference~\cite{daidata} treated IC with side information which includes \emph{coded} packets as well.
\MS{Note that we do not use the classical formulation of these problems since we do not address decoding of finite blocks but view the infinite horizon view of the problem.}

Minimization of the overall transmission time was addressed in~\cite{lucani2009network}. The policy described in~\cite{lucani2009network}, if considered on a per-node basis, results in a greedy algorithm, maximizing the information gained from a single transmission. In the MDP-based approach herein, however, the transmission policy accounts for the \emph{ability to transit to more rewarding states in the future}, hence generalizes the greedy approach. 

Index coding in a scenario where each packet should be transmitted to all was compared to an ARQ scheme in~\cite{ghaderi2008reliability}. It was shown that as the number of users K grows, the number of transmissions with NC is constant, while it is logarithmic in K in the case of ARQ. ARQ schemes were also analyzed in~\cite{sundararajan2008arq} and implemented in~\cite{sundararajan2009feedback}, where the authors considered a broadcast network and the queue size at the sender side as the primary performance metric. As for unicast scenarios, the finite horizon scheme~\cite{el2010index} optimized the number of decoding operations, rather than the number of transmissions.

It is important to note that there are a few critical differences between the state of the art in index coding and the coding scheme suggested in the paper. First, index coding considers only finite horizon scenarios, i.e., each receiver is interested in a fixed, finite list of packets, and one has to devise, before communication starts, the best coding scheme in terms of minimizing the number of packets required to satisfy all demands. In our problem, users have \emph{infinite streams}, the state of the system (in terms of the side information available) changes after each transmission, and one have to make coding decisions \emph{after each transmission}.  Second, the state of the art index codes are not instantly decodable, namely, receivers might need to wait for the end of the block to decode their data. The scheme herein is instantly decodable.  Finally, index coding allows the receivers' demands to partially overlap, hence is more general in this sense. Yet, it is well known to be a hard problem (e.g. ~\cite{chaudhry2011complementary}), with no efficient solutions in the general case. Thus, it is beneficial to consider different settings, in which high gains can be efficiently achieved.

%% file: Model1_O.tex
%%%%%%%%%%%%%%%%%%%%%%%%%%%%%%%%%%%%%%%%%%%%%%%%%%%%%%%%%%%
\section{Model description}
\label{sec:model}% Induced MDP for network coding}
%%%%%%%%%%%%%%%%%%%%%%%%%%%%%%%%%%%%%%%%%%%%%%%%%%%%%%%%%%%
We consider a downlink wireless model, with one transmitter (access-point) and $K$ receivers. At the sender, we assume an infinite stream of packets for each user (i.e., unicast traffic). We assume a Stop-and-Wait based protocol, accordingly, even though the sender has an infinite set of packets per receiver, we assume only one such packet is active at a given time per receiver, i.e., the sender does not transmit new packets for a receiver until the active one is received correctly and acknowledged. Note that this mechanism conforms to the widely deployed IEEE 802.11 protocol suite. Our channel model assumes the packet sent at each slot is received at receiver $k$ with probability $p_k$, independently of the other receivers and of the previously received packets (memoryless independent users). The packet loss probabilities are assumed to be fixed in time. We assume that uncoded packets correctly received by a receiver which is not the intended one, are cached. Note that, on top of the coding scheme we suggest, of-the-shelve error correction codes can be utilized in order to improve $p_k$ at the expense of overhead.

We assume that packets overheard by undesignated users can be stored for future use. Yet, we assume that only uncoded packets can be stored at the receivers while coded or corrupted packets are discarded. We distinguish between two cases, unlimited storage time and limited storage time. We first treat the case where the stored packets are never outdated (i.e. storage time is unlimited). Denote by ${\bf M}$ the space of $K\times K$ binary matrices, where each $s\in \bf M$ represents a possible state. In particular, each line $i\in\{1,\cdots,K\}$ constitutes a vector of indicators such that $s_{ij}=1$ if and only if user $j$ has a packet designated for user $i$. We assume the AP always aware of the data kept by the receivers using status updates sent by each receiver. We assume that when a receiver overhears or decodes a packet destined to another, it is able to store it. The state of the system is updated after every transmission slot. At transmission slot $t$ the state is represented by $s(t)\in M$. In the case that user $k$ successfully decodes its packet, $s_{k,i}=0,\forall i$ is set. Setting the entire row $k$ to zero is motivated by the simple reasoning that users that stored the packet prior to the successful transmission can now discard it. The sender can now send the next packet for that user. In the case that  the destination fails to receive its packet, we set $s_{k,k'} = 1$ if the packet is heard by user $k'$ and $s_{k,k'} = 0$, $k \ne k'$, otherwise.

Next, we consider the limited storage time for which the time a packet can be stored at each receiver's buffer; we denote the number of time slots a packet can be stored by Time to Expire (TTE). Accordingly, a packet overheard by a non-intended receiver and which is stored for more than its maximal validation time is invalidated and discarded. For simplicity, we assume a system of identical users, i.e., all packets have a similar TTE \textit{limit} which we denote by $T$, i.e., the maximal time a packet can be stored is $T$ time slots. Respectively, each transmitted packet has a TTE associated with it. This value is updated every time slot, until the packet is correctly decoded or outdated and dropped. We denote the TTE of a stored packet, at some given time slot, as $\tau\in\{1,\cdots,T\}$ and by $\tau=0$ the case that no valid packet is stored. Every time slot, for every packet stored by a user, $\tau$ is decremented by $1$. Once $\tau$ becomes $0$, the corresponding packet is outdated and dropped.

We denote by ${\bf M^{TTE}}$ the space of $K\times K$  matrices, where each $s\in{\bf M^{TTE}}$ represents a matrix of TTE values associated with undecoded packets held by the receivers. In particular, each line $i\in\{1,\cdots,K\}$ constitutes a vector of TTE parameters, such that $s_{ij}=\tau$, if and only if user $j$ has a packet destined to user $i$, and there are $\tau$ time slots left till the packet expires. Similarly to the scenario without TTE constraint, we assume that the AP is always aware of what data is kept by which receivers. Whenever the intended receiver fails to receive its packet, the AP sets $s_{k,k'} = T$ if the packet is either heard by user $k'$, or user $k'$ already has this packet stored, and sets $s_{k,k'} = 0, \; k \ne k'$, otherwise. Hence, all users that overheard some packet have an equal value stored for its current TTE. This value is stored at the AP and is used for the transmission decisions.

Each packet is represented as $m$ symbols over the field ${\mathbb F}_{2^k} $. Thus, its payload consists of $mk$ bits. Now, each time a packet is sent, the sender has a few options as to which type of packet to send. These "options" constitute its action space. Specifically, it can either choose a single packet from the stream intended to a specific user, and send that packet to that  user (termed uncoded packet), or, alternatively, it can code together a few packets. In this work, we used the standard linear network coding~\cite{}, however, since nodes do not store coded packets, and we require instant decodability, coding is done over the binary field. Thus, at every transmission slot, the AP encodes
%Each time a packet is sent, the sender can either choose a packet to send to a single user (termed uncoded packet), or code together a few packets. Specifically, at every transmission slot the AP encodes
\be\label{z}
z=\alpha_1d_{1}\oplus \alpha_2d_{2}\oplus,\cdots,\oplus\alpha_kd_{k}
\ee
\noindent and sends this packet, where for each $k$, $\alpha_k\in\{0,1\}$, $d_i$ denotes the packet currently expected by user $i$ and $\oplus$ denotes bitwise XOR.
Namely, the AP decides on coefficients $\alpha_k\in\{0,1\}$, where $\alpha_k=1$ means a packet for user $k$ participates in the current coded transmission slot. Otherwise, $\alpha_k=0$. Note that choosing $\alpha_k=1$ for only one user is equivalent to transmitting an uncoded dedicated packet to user $k$. Hence, the action space is of size $2^k-1$, and it includes all possibilities of uncoded and coded packets (excluding the zero packet).
Recall that as previously explained, only such uncoded packets can be stored by undesignated receivers. Note that packets to be combined (coded) are assumed to have the same size (if not, the shorter ones are padded with trailing $0$s).
%Where $d_i$ is the packet currently expected by user $i$ and $\oplus$ stands for the bitwise XOR operator. Namely, AP decides on coefficients $\alpha_k\in\{0,1\}$, where $\alpha_k=1$ means a packet for user $k$ participates in the current transmission slot. Otherwise, $\alpha_k=0$. Note, that choosing $\alpha_k=1$ for only one user is equivalent to transmitting an uncoded dedicated packet to user $k$. Recall that as previously explained, only uncoded packets can be stored by undesignated receivers.

The setting described above can be seen as a framework including a state-space, an action-space which comprises the possible packet combinations the AP can send at any given time slot (denoting the action at transmission slot $t$ by $a(t)$) and the transition probabilities. Due to the Markov property, we deduce that the problem can be formulated as an MDP, with the objective to maximize the transmission throughput.
Hence, we define an appropriate stochastic reward $r(s(t+1),a(t),s(t))$, associated with transitioning from state $s(t)$ to state $s(t+1)$ after taking the action $a(t)$, such that positive reward is accumulated for each successfully decoded packet. For example, if a coded packet of $n$ packets is sent, and $m\leq n$ of them are successfully decoded by their intended receivers, we have $r(s(t+1),a(t),s(t))=m$. Failing to decode gives no reward. Storing a packet at the receiver which is not the addressee gives no reward. However, note that it may increase the \emph{potential} number of packets decoded in the future (that is, transition to a state with a higher potential value).

We assume that the same transmission effort is required by the AP whether it transmits an uncoded packet, a coded one or does not transmit at all, i.e., fixed transmission costs are assumed. Consequently, abstention from sending a packet at any transmission slot is the worst option possible. Hence, at each time slot exactly one packet is sent. The objective is to find a policy which maximizes the attained throughput, which is measured in \emph{$\big(\frac{\text{packets decoded}}{\text{time-slots}}\big)$}.

In the next section, we bring the technical definition of the MDP and state aggregation, in order to utilize it for the described model. For general definitions and theory of MDP the reader is referred to~\cite{bertsekas2005dynamicI}.

%% file: Induced_MDP5_O.tex
%%%%%%%%%%%%%%%%%%%%%%%%%%%%%%%%%%%%%%%%%%%%%%
\section{MDP with restricted action space and induced MDP}
\label{sec:induced}
 %%%%%%%%%%%%%%%%%%%%%%%%%%%%%%%%%%%%%%%%%%%%%%%%
In this section, we introduce the general notation which lays the ground for the state aggregation. We follow the concepts of abstract MDPs in~\cite{li2006towards}, yet adjust our notation and forthcoming analysis to fit our model and results throughout the rest of the paper.
% \MS{Note that the definitions and results presented in this section hold for any general MDP, yet they are specifically tailored to suit our problem.}

As previously mentioned the problem can be formulated as a finite MDP.
\off{Let us denote the ground MDP by $\scrM_0$, characterized by the five tuple $\scrM_0= \langle \scrS,\scrR,\scrP,\scrA,\gamma \rangle $, where $\scrS$ is the finite state-space, we term every state $s \in \scrS$ as a {\it detailed} state, since it includes a detailed account of which packets were received by which users; $\scrR$ is a bounded reward function; $\scrP$ are transition probabilities and $\scrA$ is a finite set of actions called the action space. We consider both long run average cost and discounted cost with $0 \leq \gamma < 1$ is a discount factor on the summed sequence of rewards. 	 	} 	
%\FMO{Mark, if we are following Littman's notations, why did you have to change the order of parameters. Note that the order makes a lot of sense and it is used by many other studies (including Wikipedia). In particular:
Let us denote the ground MDP by $\scrM_0$, characterized by the five tuple $\langle \scrS,\scrA,\scrP,\scrR,\gamma \rangle $, where $\scrS$ is the finite state-space, in which we term every state $s\in\scrS$ as a {\it detailed} state, since it includes a detailed account of system; $\scrA$ is a finite set of actions called the action space, $\scrP$ are transition probabilities with $p(s'|s,a)$ denoting the probability to proceed to state $s'$, being in state $s$ and acting with action $a$, $\scrR$ is a bounded reward function with $r(s',a,s)$ denoting the expected immediate reward gained by taking action $a$ in state $s$ and proceeding to state $s'$. We consider both long run average cost and discounted cost with $0 \leq \gamma < 1$ being a discount factor.  %***Based on time, I think we should change it throughout the paper}.
A {\it policy} is a mapping from states to actions ($\scrS \mapsto \scrA$). In this paper we will focus only on policies that do not depend on the time (stationary policies). We denote the set of all admissible policies by $\scrU$.
We denote by $p(s'|s,a)$ the probability to proceed to state $s'$, being in state $s$ and acting with action $a$, and by $r(s',a,s)$ stochastic reward function attained from such instance.
%\FMO{
The action in some state $s$ is denoted by $a(s)$. We further denote by $r(s,a)=\sum_{s'}r(s',a,s)p(s'|s,a)$ the average reward of being in state $s$ and taking action $a$. %*** $a_s$ is abuse of notations as we also use the subscript to denote the time ($a_t$)!. The notation $r_0(s,a)$ is not used here hence should be removed and placed just before or after it is used (if needed)}
As previously mentioned we consider two performance criteria: discounted infinite horizon cost and long run average cost. Specifically, the discounted infinite horizon cost associated with a given policy $\pi$ and initial state $s_0$ is given by
\[
J^{\pi}(s_0)=\E\big[\sum_{t=0}^{\infty}{\gamma^t}r(s'_{t+1},a^{\pi}_{t},s_{t})|s_0\big]
\]
where $s_t$ and $a^{\pi}_{t}$ denote the state visited at time slot $t$ and action taken on time slot $t$ based on state $s_t$ and according to policy $\pi_{AC}$. The long run average cost associated with policy $\pi_a$ is
\begin{equation}\label{ind1}
J^{\pi_{AC}}=\lim_{N\to\infty}\frac{1}{N}\E\big[\sum_{t=0}^{N}r(s'_{t+1},a^{\pi}_{t},s_{t})\big].
\end{equation}

Note that the initial state has no impact on the long run average cost (Eq.~\eqref{ind1}) as its effect is dissolved over time (\cite{bertsekas2005dynamicI}). In this section, we only refer to the discounted case. We examine the average case in Section~\ref{sec:vstudy} and in the appendices. The value function for the discounted case is given by {\small{$V(s_0)=\sup_{\pi\in\scrU}J^{\pi}(s_0)$}}.

We now define the restricted and induced MDPs, which allow us to work with much simpler MDPs in our communication problem, yet retain the notion of network coding hence the near-optimal performance.

The policy restricted MDP is stimulated by the state aggregation we suggest. State aggregation exploits properties present in the state space of the basic MDP (the detailed states) for aggregation of multiple detailed states into one aggregated state obtaining an MDP with smaller state space. In particular, a partition $\bar \scrS = \{ \bar s_1, . . . , \bar s_n \}$ of the detail state space may serve as an aggregated state space if each detailed state is mapped to one and only one aggregated state ($\bigcup_{i=1}^{n} \bar s_i =\scrS  \; ; \; \bar s_i \bigcap \bar s_j = \emptyset$).
\off{Furthermore, under state aggregation, the same action should be applied to all detailed states which are mapped to the same aggregated state for any admissible policy.
}  % end off
 We now formally define the Policy Restricted MDP.

\begin{definition}
\label{def:Policy_Restricted_MDP}
A {\it policy restricted} MDP denoted by ${\scrM_1}=\scrP(\scrM_0,\phi,\bar\scrA)$, is defined by

\noindent(I) A mapping  $\phi$ acting on $\scrS$, such that $\phi:\scrS\mapsto\bar\scrS$, where $\bar \scrS=\bigcup_{i}\bar s_i$ for disjoint $\bar s_i$,
\newline(II) A restricted action space $\bar\scrA\in\scrA$, and
\newline(III) A restricted set of policies $\bar\scrU\in\scrU$, such that for all $\bar\pi\in\bar\scrU$, it holds $\bar\pi(s)\in\bar\scrA$, $\forall s\in\scrS$ and if $\phi(s_1)=\phi(s_2)$ then $\bar a^{\bar\pi}(s_1)=\bar a^{\bar\pi} (s_2)$, where $\bar a^{\bar\pi}(s_1)=\bar\pi(s_1)$, and $\bar a^{\bar\pi}(s_2)=\bar\pi(s_2)$.  %\MS{such that $\phi(s)=\bar s$, where $s\in\scrS$, }%:S\mapsto\bar S$,
\end{definition}

In other words, we define a mapping rule $\phi(s)$ which associates each detailed state with an aggregated state, partitioning the state space ($\scrS$) into the aggregated state space ($\bar\scrS$). In correspondence to the aggregated state space, only policies that enforce the same action for all states belonging to the same aggregated state are admissible, i.e., the same action should be taken for all $s_i \in \bar s_i $. \MS{We will use the notation $s\in\bar s$ if it holds $\phi(s)=\bar s$, and $\bar\pi(\bar s)$ as the  equivalent to $\bar\pi(\phi(s))$. }

\off{consider adding a comment regarding Ortner: Note that in contrast to other studies which constraint the aggregated states to have identical or similar cost and transition probabilities (e.g., \cite{}), in this study the cost and the transition probabilities between states which are aggregated to a single meta state can be quite different.
}  % end off

Note that the policy restricted MDP is still based on the detailed state-space and thus is difficult to calculate. Accordingly, we define the induced MDP to which the detailed states are transparent. The induced MDP is formed by the atomic states, induced by the aforementioned aggregated states, hence, relies on significantly smaller state space, and has similar action rules. By means of the aggregated state space and the corresponding policy restriction space, one can define transition probabilities as follows:
Given an admissible policy $\bar \pi\in\bar\scrU$, the transition probabilities between the aggregated states which we denote by $p(\bar s'|\bar s,\bar a)$, are:

\begin{equation}\label{eq:ind1}
\begin{aligned}
p(\bar s'|\bar s,\bar a)= & \sum_{s'\in\bar s'} \sum_{s''}p(s'|s'',\bar s,\bar a)p^{\bar \pi}(s''|\bar s,\bar a)= \\
&=\sum_{s'\in\bar s'} \sum_{s''}p(s'|s'',\bar a)p^{\bar \pi}(s''|\bar s)= \sum_{s'\in\bar s'}\sum_{s''}p(s'|s'',\bar a)p^{\bar \pi}(s''|\bar s)\\
\end{aligned}
\end{equation}

%\FMO{note that both of Mark's equations (3) and (4) were referred to in the Appendix -> need to be fixed}

%%%%%%%%%%%%%%%%%%  off  %%%%%%%%%%%%%%%%%%%%
\off{
Then, with respect to this policy, the probabilities are calculated as follows:
First, $p(\bar s|s,\bar a)=\sum_{s'\in\bar s}p(s'| s,\bar a)$.
Second, observe that
\begin{equation}\label{eq:p9}
 p(s'|\bar s,\bar a)=
\sum_{s''}p(s'|s'',\bar s,\bar a)p_{\bar a}(s''|\bar s,\bar a)=\sum_{s''}p(s'|s'',\bar a)p_{\bar a}(s''|\bar s),
\end{equation}
\begin{equation}\label{eq:p10}
p(\bar s'|\bar s,\bar a)=\sum_{s'\in\bar s'}\sum_{s''}p(s'|s'',\bar a)p_{\bar a}(s''|\bar s).
\end{equation}
} % end off
%%%%%%%%%%%%%%%%%%%%%%%%%%%%%%%%%%%%%%%%%

Where $p^{\bar \pi}(s''|\bar s)$ denotes the stationary probability of being in the detailed state $s'' \in \scrS $, conditioned on the aggregated state $\bar s$. Obviously, these probabilities may depend on the policy $\bar\pi\in\bar\scrU$, hence the superscript $\bar\pi$; yet, for simplicity in the sequel, when clear from the context, we will omit the superscript. Clearly, $\sum_{s''\in\bar s}p^{\bar \pi}(s''|\bar s)=1$.
\MS{Define the cost of the policy restricted MDP as follows:
{\small{
		\(
 J^{\bar\pi}(s_0)=\E\big[\sum_{t=0}^{\infty}{\gamma^t}r(s'_{t+1},\bar a^{\bar\pi}_{t},s_{t})|s_0\big] %=
 %\E\big[\sum_{t=0}^{\infty}{\gamma^t}r(s'_{t+1},\bar a^{\bar\pi}_{t},s_{t})|s_0\big]=
 %\E\big[\sum_{t=0}^{\infty}{\gamma^t}\sum \bar r(s'_{t+1},\bar a^{\bar\pi}_{t},\bar s_{t})|s_0\big]
\)
}}.
%where $\bar r(s'_{t+1},\bar a^{\bar\pi}_{t},\bar s_{t})=\sum_{s\in\bar s}r(s'_{t+1},\bar a^{\bar\pi}_{t},s_{t})p^{\bar\pi}(s|\bar s)$.
The corresponding value function is given by $V_{\bar\scrU}(s_0)=\sup_{\bar\pi\in\bar\scrU}J^{\bar\pi}(s_0)$.
Since policy restricted MDP sees the detailed states we also define $J^{\bar\pi}(\bar s_0)=\sum_{s_0\in\bar s_0}J^{\bar\pi}(s_0)p^{\bar \pi}(s_0|\bar s_0)$ and $V_{\bar\scrU}(\bar s_0)=\sup_{\bar\pi\in\bar\scrU}J^{\bar\pi}(\bar s_0)$}.

Next we formally define the induced MDP:

\off{Consequently, based on the aggregated state space and the corresponding policy restriction space, one can define a probability space which includes only the aggregated states and the corresponding transition probabilities $p(\bar s'|\bar s,\bar a)$, yet, these transition probabilities rely on the underline detailed states. In line with this note that in contrast to other studies which constrain the aggregated states to have identical or similar cost and transition probabilities (e.g., \cite{}), in this paper there are no such constraints and the cost and the transition probabilities between aggregated states can be different. Note that the policy restricted MDP is still based on the detailed state-space and thus cannot be efficiently computed. Hence we present the induced MDP which significantly reduced state space. Specifically, the induced MDP is formed by the atomic states, \textit{induced} by the aforementioned aggregated states, and has similar action rules. We next define the {\it induced} MDP.}

\off{For example, assume $s \in \bar\scrS $, $s',s'' \in \bar\scrS' $ and $s''' \in \bar\scrS'' $ the probability of moving between states $\bar\scrS \rightarrow \bar\scrS' \rightarrow \scrS' $ taken actions $a_t$ and $a_{t+1}$ depends on the detailed states and in particular on the intermediate state hence can be $p(s'|s,a_t) \times p(s'''|s',a_{t+1})$ or it can be $p(s''|s,a_t) \times p(s'''|s'',a_{t+1})$, which can be quite different. Accordingly, next we define the {\it induced} MDP to which the detailed states are transparent.
} %end off

\begin{definition}
\label{def:Induced_MDP}
MDP $\hat{\scrM}=\scrI(\scrM_0,\phi,\hat\scrA)$ is induced by policy restricted $\scrM_1$ on $\scrM_0$, if

\noindent(I) Each state $\hat s\in\hat{S}$ uniquely relates to some $\bar s\in\bar{S}$;
Denote this relation as $\hat s\sim\bar s$. % holds for each $\hat s\in\hat S$, is one-to-one correspondence and defined by  $\hat s=\phi(\phi^{-1}(\bar s))$
\newline(II) For all $\hat s\sim\bar s$, the actions $\hat a(\hat s)$ available in $\hat s$ are equivalent to $\bar a(\bar s)$. Denote the relation of the action space as $\hat A\sim\bar A$, \MS{and relation of the actions as $\hat a\sim\bar a$. }
\newline(III) The transition probabilities are defined on similar probability space and comply with $p(\hat s'|\hat s,\hat a)=p(\bar s'|\bar s,\bar a)$, for all $\hat s',\hat s,\hat a$,
for which $\hat s\sim\bar s$. % in $\scrM_1$.
\end{definition}
\ifArX
Note that an induced MDP sees no detailed states. \MS{That is}, each state of the induced MDP stands for distinct aggregation of detailed states in a policy restricted MDP. Note that if one takes a sequence of detailed states $\{s_0,s_1,s_2,\cdots\}$ and applies $\phi$ to it, the resulting sequence $\{\phi(s_0),\phi(s_1),\phi(s_2),\cdots\}$ is not necessarily Markovian. This is because $\phi$ is non-injective surjective function. That is, it is not a bijection for the reason the injective property does not hold. However, as we show in the sequel, one can construct transition probabilities from $\phi(s_i)$ to $\phi(s_j)$, i.e. the aggregated states, such that the resulting process is Markovian. \MS{As far as the problem of coded retransmission is concerned, the state space is reduced from $\scrS=2^{K(K-1)}$ to $\bar\scrS$, where the size of the latter is determined by the properties of the aforementioned mapping $\phi$.}
\else
Note that an induced MDP sees no detailed states. \MS{That is}, each state of the induced MDP stands for distinct aggregation of detailed states in a policy restricted MDP.
\fi
Denote $\hat\scrU$ defined over $\hat\scrA$.

The discounted infinite horizon cost associated with some policy $\hat\pi\in\hat\scrU$ is given by
$
J^{\hat\pi}(\hat s_0)=\E\big[\sum_{t=0}^{\infty}\gamma^t\hat r(\hat{s'}_{t+1},a^{\hat\pi}_{t},\hat{s}_{t})|\hat s_0\big]$. The corresponding value function is given by $V_{\hat\scrU}(\hat s_0)=\sup_{\hat\pi\in\hat\scrU}J^{\hat\pi}(\hat s_0)$.

We aim to set the appropriate reward function for the induced MDP such that its value function will be comparable to that of the policy restricted one.
The relation between $\scrI(\scrM_0,\phi,\hat A)$ and $\scrP(\scrM_0,\phi,\bar A)$ is given by the following proposition:
\begin{proposition}
For an MDP $\scrM_0(\scrS,\scrA,\scrP,\scrR,\gamma)$, a policy restricted MDP $\scrM_1(\scrS,\bar\scrA,\scrP,\bar \scrR,\gamma)$ such that ${\scrM_1}=\scrP(\scrM_0,\phi,\bar \scrA)$, and an induced MDP $\hat\scrM(\hat \scrS,\hat \scrA,\hat \scrP,\hat \scrR,\gamma)$, where $\hat{\scrM}=\scrI(\scrM_0,\phi,\hat \scrA)$, with given initial states $\hat s_0\sim\bar s_0$, there exists a reward function $\hat \scrR$, such that $V_{\hat\scrU}(\hat s_0)=V_{\bar\scrU}(\bar s_0)$. %where $s_0\in\{s|\phi(s_0)=\bar s_0\}$.
\label{prop1}
\end{proposition} \noindent See Appendix~\ref{app:proof} for the proof.

Intuitively, one sees that the reward of an induced MDP may be interpreted as the suitably weighted sum of the rewards of the corresponding policy restricted MDP, normalized by the sum of the weights.
Note that these weights are found by the transition probabilities to the detailed states which compose the corresponding destination aggregated state, $\bar s'$, for which the relation $\bar s'\sim\hat s'$ holds. The key point is that with the proper reward function, the induced MDP achieves {\it the same value function} as the restricted one. Note that since $\calU_1\subset\calU$, in general, we have $V_{\hat\scrU}(\hat s_0)=V_{\bar\scrU}(s_0)\leq V_{\scrU}(s_0)$.
%\MS{
%\begin{remark}
%The transition probabilities of the policy restricted MDP can be calculated by the following:
%\[
%p(\bar s'|\bar s,\bar a)=\sum_{s'\in\bar s'}\sum_{s\in\bar s}p(s'|\bar a,s))p(s|\bar s)
%\]
%\end{remark}
%}
\off{We emphasize the generality of our approach in this paper. Namely, the only parts that determine the aggregation method %and, as we show in the following sections, the entire solution and algorithms,
are the choice of the aggregation function $\phi$ and of the actions space formed by $\bar a(\phi(s))$, %In particular, the ones
which allow to uniquely define $\scrP(\scrM_0,\phi,\bar \scrA)$ and $\scrI(\scrM_0,\phi,\hat \scrA)$.
\MS{Hence, the method of choice of $\phi$ is left open. The function should be (may be heuristically) constructed according to the specifics of the original problem.}
In the rest of the paper, for the sake of simplicity, we will omit the relation $\sim$, where the process it is related to is clear. Hence, in these cases we will refer to both induced and policy restricted MDPs with slight abuse of notation, by merely using $\bar\scrA$ and $\bar\scrS$ signs.
%Denote by $P$, the transition probabilities, by $R$ the reward given for the number of successfully decoded bits, and the action space $A$.
} % end off 

%% file: State_Aggregation1_O.tex
%%%%%%%%%%%%%%%%%%%%%%%%%%%%%%%%%%%%%%%%%%%%%%%%%%%%%%%%%%%%%%%%%%%%%%%%
\section{State aggregation and reinforcement learning based solution}
\label{sec:st-agg}
%%%%%%%%%%%%%%%%%%%%%%%%%%%%%%%%%%%%%%%%%%%%%%%%%%%%%%%%%%%%%%%%%%%%%%%%
Having laid the ground, in this section we follow the notations and definitions described in Section~\ref{sec:induced} to provide the formal definition of the state aggregation and restricted policy for the \emph{communication problem} considered.  Specifically, we will base both the aggregated states and the action space on the clique size (which will be defined shortly) and on the number of empty lines in the state matrix; the rewards and transition probabilities of the induced MDP will be determined accordingly. %State aggregation may be equivalently applied for both average and discounted costs. In this section, we mainly focus on the discounted case.

%%%%%%%%%%%%%%%%%%%%%%%%%%%%%%%%%%%%%%%%%%%%%%%%%%%%%%%%%%%%%%%%%
\subsection{State aggregation and the restricted action space}
\label{subsec:State_aggregation}
%%%%%%%%%%%%%%%%%%%%%%%%%%%%%%%%%%%%%%%%%%%%%%%%%%%%%%%%%%%%%%%%%
In order to define the state aggregation and the restricted action space, let us first define a {\it clique} structure and associate it with clique transmission. We associate a directed graph $G(V,\Gamma)$, with each state $s \in S$, such that a vertex $v_j \in V$ is assigned to each user $j$ and a set of directed edges are formed between each user and the users it holds a packet to, i.e., $\Gamma(s)=\{e_{ij}=\{v_i,v_j\} | s(i,j)=1\}$. As commonly defined in graph theory, a clique is a subset of vertices such that each vertex is connected to each other vertex in the set, i.e., $ Q \; is \; a \; clique; iff \; \{ \forall v_{i}, v_{j} \in Q: s(i,j)=1,\forall j\neq i;\;\; i,j\in\{1,\cdots,a\}\}$ . The size of a clique is determined by the number of vertices it contains. Note that in the context of our problem any set of users forming such a clique ($\forall v_{i} \in Q$) implies that each user in the set has all the messages intended to all other users in the set. Accordingly, a coded message, composed of all the messages intended to all users in the set, can be sent, such that each user in the set can decode its own. Denote the size of the maximum clique induced by state $s$ by $L(s)$ and by $E(s)$ the number of empty lines in $s$.

We construct the aggregation such that each aggregated state is defined by the tuple $\{L(s),E(s)\}$, i.e., $\phi (s) = \{L(s),E(s)\}$. For clarification let us examine the following example:

\begin{example}
\label{ex:aggregated states}
Consider a communication network consisting of $5$ users. Observe the following states:
{\tiny\[
s_1 =
\begin{pmatrix}
0 & 1 & 0 & 0 & 1 \\
1 & 0 & 1 & 1 & 0 \\
0 & 1 & 0 & 1 & 0 \\
0 & 1 & 1 & 0 & 0\\
0 & 1 & 0 & 0 & 0
\end{pmatrix}
\quad
s_2 =
\begin{pmatrix}
0 & 1 & 1 & 0 & 1 \\
1 & 0 & 1 & 0 & 1 \\
1 & 1 & 0 & 0 & 0 \\
0 & 0 & 0 & 0 & 1\\
1 & 1 & 0 & 0 & 0
\end{pmatrix}
\]}
Note that $s_1$ contains a clique of size $3$ associated with users $2,3,4$ and a clique of size $2$ associated with users $1,2$. The state $s_2$ contains the $2$ cliques of size $3$ associated with users $1,2,3$ and users $1,2,5$. There are no empty lines in either state.
Since the suggested aggregation considers only the maximum clique size and the number of empty lines, both states above pertain to the same aggregated state denoted by $(3,0)$, i.e., $\phi(s_1)= \{L(s_1),E(s_1)\} =\{3,0\}$ and $ \phi(s_2) =\{L(s_2),E(s_2)\} =\{3,0\}$, i.e., $\phi(s_1)= \phi(s_2) =\{3,0\}$.
\end{example}
\ifArX
The additional detailed example can be found in Appendix~\ref{xmpl:1}.
\else
The additional detailed example can be found in~\cite{2015arXiv150202893S}.
\fi
Note that the number of possible states (i.e., number of unique pairs $\{ L(s),E(s)\}$) is dramatically reduced and is upper bounded by $J=(K+1)K$. Further note that while finding a maximum clique is hard in general, graphs resulting from the state matrix in our setting are random and have cliques of logarithmic size~\cite{cohen2013coded}, hence $L(s)$ can be found efficiently.

Having defined the state aggregation, we define the restricted action space. In particular, in accordance with the aggregated states we allow only two actions, sending a coded packet to the maximum clique which we denote by $\bar a=1$, or sending an uncoded packet corresponding to a randomly chosen empty line denoted by $\bar a=2$ ($\bar\scrA\in\{1,2\}$). Note that the restricted action space complies with the constraint that the same policy should be applied to all states in the same aggregated state. It is important to note that once an action is decided (according to the aggregated state), the actual combination depends on the detailed state, (i.e., to which user (users) to send an uncoded (coded) packet. %is determined based on the detailed states)%, i.e., the complexity of the decision process is done on the aggregated states, yet the AP needs to examine the detailed states in order to determine who are the set of users comprising the maximum clique or to which user correlated to an empty line to send the packet.
In Example~\ref{ex:aggregated states}, since there are no empty lines, the only permissible action is to send a coded packet to the maximum clique, that is, sending $p_2\oplus p_3\oplus p_4$ for $s_1$ or one of $p_1\oplus p_2\oplus p_3$, $p_1\oplus p_2\oplus p_5$ for $s_2$. Note that in the case that there are no empty lines and the maximum clique size is one, the AP should send a coded packet to one of the maximum cliques, yet since the size of the maximum clique is equal to $1$, the coded packet comprises a single packet hence it is practically uncoded.

Obviously, the action space defined here is not the only plausible option. For example, one may define sending the empty line which has the greatest potential to increase the maximal clique. Moreover, in some cases sending an uncoded packet to a non-empty line might be a more valuable option. However, our approach is to choose a simple aggregation that even though not optimal, is clearly motivated by the original communication problem, hence is expected to attain good results. In addition, we aspire that the number of operations (e.g., determining the maximal clique or random selection of an empty line) which is required from the AP to perform (on the detailed states) will be minimal. %, which in our case is to determine the maximum clique or to randomly choose an empty line without trying to appraise the potential gain.
The evaluation part (Section~\ref {sec:Impl}) confirms that even though our approach is not optimal it attains very good results.

%%%%%%%%%%%%%%%%%%%%%%%%%%%%%%%%%%%%%%%%%%%%%%%%%%%%%%%%%%%%%%%%%%%%%%%%%%%%%%%%%%%%%%%%%%%%%%%%
\subsection{Finding the policy utilizing reinforcement learning}
\label{Subsec:Reinforcement_Learning}
%%%%%%%%%%%%%%%%%%%%%%%%%%%%%%%%%%%%%%%%%%%%%%%%%%%%%%%%%%%%%%%%%%%%%%%%%%%%%%%%%%%%%%%%%%%%%%%%
In the previous subsection we have defined the state aggregation and the restricted action space. In order to complete the setup in this subsection we obtain the appropriate reward $\hat\scrR$ and the transition probabilities $p(\hat s|\hat s',\hat a)$, for the induced MDP.

There are three major obstacles in computing the transition probabilities and constructing the associated rewards according to Proposition~\ref{prop1}. First, the packet loss probabilities typically are not known to the AP. Second, in order to compute the transition probabilities one needs to go over each detailed state and compute the probability of going to each state for each possible action (it implies order of ($2^{K(K-1)} \times 2^{K(K-1)})$ action). Third, the transition probabilities are policy dependent, i.e., the transition probability of going from aggregated state $\bar s$ to aggregated state $\bar s'$ relies on the steady state probability of being in detailed state $s$ given that the system is in state $\bar s$ (see equation~\eqref{eq:ind1}). These probabilities are policy dependent. %, which is not known in advance.
Recall that our objective is to determine the policy. %is what we want to determine in the first place.
Even though the first obstacle is relatively easy to resolve as the AP can keep a history record and if necessary send dedicated probe packets to estimate the packet loss on each outgoing link, the other difficulties are more challenging as obviously trying to compute the transition probabilities and the reward values is impractical. Accordingly, we utilize reinforcement learning (RL), an effective learning technique which has the capability of finding the reward maximizing policy, %that ought to be applied in order to maximize rewards,
in discrete stochastic environments, without explicit specification of the transition probabilities.
\off{Specifically, in reinforcement learning the transition probabilities are accessed through a simulator (or real-world scenarios) for which typically the reinforcement learner (agent) chooses a random (or in some cases particular) initial state and iterates between the different states taking different actions and updating the transition probabilities and expected rewards accordingly. The policy is attained by utilizing a trial-and-error learning procedure in which the system tries different actions at different states trying to learn the optimal action for each state (e.g., \cite{}). }
Specifically, RL is based on a feedback loop in which the reinforcement agent (learner or AP in our case) selects an action based on its current state, gets feedback in the form of the next state and an associated reward, and updates the estimated records. The selection of the action is based on the current state $s$ and the temporary (current) policy, and balances exploration and exploitation, i.e., on the one hand the agent has to exploit what is already known, but on the other hand it has to explore in order to examine other options for making better action selections in the future. Accordingly, the agent must try a variety of actions and progressively favor those that appear to be best (e.g., \cite{sutton1998introduction}). %For example in RL with an $\epsilon-greedy$ policy approach (\cite{} \FMO{add cite}) at each transaction the agent chooses the action that it believes has the best long-term value with probability $1-\epsilon$, and it chooses an arbitrary action uniformly at random otherwise. The parameter $0<\epsilon<1$ is a tuning parameter which ordinarily is decreased as the learning procedure proceeds, making the agent explore less as time goes by.
One of the difficulties of our learning problem is expressed in highly differentiated \emph{access frequencies} among the various states. Accordingly, since the algorithm is expected to visit each state multiple times, we need to direct it and to force it to visit less visited states. Several RL algorithms that can be utilized to solve our problem exist, e.g., MBIE~\cite{strehl2008analysis}, ${\bf E^3}$~\cite{kearns2002near} and R-Max~\cite{brafman2003r}; each one has its own merits. Nonetheless, since our main concern is in the application itself, rather than trying to adopt one of the known algorithms, we derived a modified simple algorithm which suits best our problem.

The proposed algorithm iterates between two steps; the learning step and the policy improvement step. Specifically, we utilize a random policy (e.g., choose at random if to transmit a randomly chosen empty line, or to transmit to the maximum clique) for the learning. In each step, we apply the {\it temporary policy} which was found in the previous step. We utilize $\epsilon-greedy$ approach with the temporary policy (that is, choose the action according to the temporary policy with probability $1-\epsilon$, and choose a random action otherwise), for $N_k$ consecutive iterations (transmissions), recording the visited aggregated states and the attained rewards (the number of consecutive transmission can vary between steps, hence the subscript $k$). It is important to note that even though the system traverses the detailed states, only the aggregated states, the actions taken and the rewards attained are recorded. That is, the AP does not hold any record of the visited detailed states. Next, we {\it update the temporary policy}, utilizing the newly learned reward functions and transition probabilities obtained during the learning phase, by applying value iteration on the corresponding \emph{Bellman equation}, that is,
\ifdouble
\begin{align}\label{3}
& V(\hat s)=\max\Big\{E_{\hat s'}[r(\hat s',\hat a=1,\hat s)+\gamma V(\hat s')], \nonumber \\
&\qquad \qquad \qquad E_{\hat s'}[r(\hat s',\hat a=2,\hat s)+\gamma V(\hat s')]\Big\}.
\end{align}
\else
\begin{align}\label{3}
& V(\hat s)=\max\Big\{E_{\hat s'}[r(\hat s',\hat a=1,\hat s)+\gamma V(\hat s')], E_{\hat s'}[r(\hat s',\hat a=2,\hat s)+\gamma V(\hat s')]\Big\}.
\end{align}
\fi
This reinforcement learning procedure continues until sufficient convergence in $V(\hat s)$ or until the policy is unchanged. %\MS{As a possible implementation option, each learning phase may be started from a different initial state, in order to heuristically overcome the obstacle of rarely visited states.}
The outcome of the proposed algorithm is the optimal policy for the induced MDP and the nearly-optimal corresponding $V(\hat s)$. %\FMO{V??? we are under policy restricted setup. Maybe nearly-optimal for the policy restricted MDP.}

%%%%%%%%%%%%%%%%%%%% Pseudo Code - Algorithm A %%%%%%%%%%%%%%%%%%%%
\ifdouble
\begin{figure}
\fbox{\begin{minipage}{0.45\textwidth}
{\small
{\it Algorithm A} \\{\it   \hspace*{1ex} Initialization}
\begin{enumerate}
  \item Initialize policy $\pi_{1}^0=\pi_R$.
  \item Run $\scrM_1$ with $\pi_{1}^0$ for $N_0$ transmissions. Each visit to $s'$, find $f_1(s')=\bar s'$. Use $\hat s'\sim\bar s'$ to update by sampling $\hat p(\hat s'|\hat a,\hat s)$ and $r(\hat s',\hat a,\hat s)$.
  \item Calculate $\pi_B^1$ by finding $\hat V_0$ from solution to the Bellman equation (performing value iteration)~\eqref{3}, using sampled probabilities and sampled rewards found in step $2$.
  \end{enumerate}}
  {\small
{\it At step $k>0$}
  \begin{enumerate}
    \item Update $\eps_k$ from predefined decreasing sequence of $\{\eps_k\}$. Set $N_k$.
    \item Set policy $\pi_{1}^k=\begin{cases}
    \pi_R\;\text{with p. $\eps_k$} \\\pi_B^k \;\text{with p. $1-\eps_k$}.
    \end{cases}$
    \item Run $\scrM_1$ with $\pi_{1}^k$ for $N_k$ transmissions. Each visit to $\hat s'$, update $\hat p(\hat s'|\hat a,\hat s)$ and $r(\hat s',\hat a,\hat s)$.
    \item Find $\hat V_k$ by value iteration over $\hat\scrM_2$. Retrieve the optimal policy $\pi_B^{k+1}$.
    \item If $|\hat J_k-\hat J_{k-1}|<\varepsilon$, for some predefined $\varepsilon$, then finish. Otherwise perform step $k+1$.
  \end{enumerate}}
\end{minipage}}
  \end{figure}
  \else
 \begin{figure}
 	\fbox{\begin{minipage}{0.95\textwidth}
 			{\small
 				{\it Algorithm A} \\{\it   \hspace*{1ex} Initialization}
 				\begin{enumerate}
 					\item Initialize policy $\pi_{1}^0=\pi_R$. Set $n(\hat s',\hat a,\hat s)=0$, $R(\hat s',\hat a,\hat s)=0$
 					\item Set $\pi_B^0=\pi_R$.
 					%\item Run $\scrM_1$ with $\pi_{1}^0$ for $N_0$ transmissions. %Each visit to $s'$, find $f_1(s')=\bar s'$. Use $\hat s'\sim\bar s'$ to update by sampling $\hat p(\hat s'|\hat a,\hat s)$ and $r(\hat s',\hat a,\hat s)$.
 					%\item Calculate $\pi_B^1$ by finding $\hat V_0$ from solution to the Bellman equation (performing value iteration), using sampled probabilities and sampled reward found in step $2$.
 				\end{enumerate}}
 				{\small
 					{\it At step $k\geq0$}
 					\begin{enumerate}
 						\item Update $\eps_k$ from predefined decreasing sequence of $\{\eps_k\}$. Set $N_k$.
 						\item Set policy $\pi_{1}^k=\begin{cases}
 						\pi_R\;\text{with probability $\eps_k$} \\\pi_B^k \;\text{with probability $1-\eps_k$}.
 						\end{cases}$
 						\item Run $\scrM_1$ with $\pi_{1}^k$ for $N_k$ transmissions. 	 \begin{enumerate}
 						\item Each visit to $\hat s$ acting $\hat a$ with reward $r'$ and going to $\hat s'$, \\set $n(\hat s',\hat a,\hat s)=n(\hat s',\hat a,\hat s)+1$, $R(\hat s',\hat a,\hat s)=R(\hat s',\hat a,\hat s)+r'$.
 							\end{enumerate}
 						\item  Calculate $\hat p(\hat s'|\hat a,\hat s)$ and $r(\hat s',\hat a,\hat s)$, from $n(\hat s',\hat a,\hat s)$, $R(\hat s',\hat a,\hat s)$ and $N_k$.
 						\item Find $\hat V_k$ by value iteration over $\hat\scrM_2$. Retrieve the optimal policy $\pi_B^{k+1}$.
 						\item If $|\hat J_k-\hat J_{k-1}|<\varepsilon$, for some predefined $\varepsilon$, then finish. Otherwise perform step $k+1$.
 					\end{enumerate}}
 				\end{minipage}
 			
 				}
 			\end{figure}
  \fi
%%%%%%%%%%%%%%%%%%%   end - Algorithm A %%%%%%%%%%%%%%%%%%%%%%%%%%%%

A pseudo code of the algorithm is given in \emph{Algorithm A}. The algorithm starts with picking a random initial policy, denoted by $\pi_R$ (Initialization step in \emph{Algorithm A}). The random policy $\pi_R$ we implemented chooses between the possible actions with equal probability, namely, $\hat a=1$ or $\hat a=2$ with probability $1/2$ each, when the choice is feasible, where $1$ and $2$ stand for transmitting the maximal clique and the random empty line, correspondingly. After the Initialization step, the algorithm runs between two steps; the learning step and the policy improvement step which are repeated iteratively. At each step the algorithm starts with a least visited aggregated state (the detailed state within can be arbitrary), and starts traversing the states for $N_k$ consecutive transmissions, based on the $\epsilon-greedy$ policy (line 2). Obviously, only the restricted actions, i.e., transmitting an empty line or transmitting the maximum clique, are allowed. The parameter $\epsilon$ is updated at the beginning of each step (line 1). After each action the agent records the previous and the next aggregated states, the action taken and the reward attained (line 4). After $N_k$ consecutive transmissions, the policy for the next steps is updated by solving the Bellman equation. The algorithm terminates when the policy or the attained value converges.

Note that the algorithm does not rely on knowing the packet loss probabilities. That is, the algorithm learns transition probabilities of the induced MDP at any fixed channel condition regardless of the exact packet loss values. Obviously, the algorithm relies on that these probabilities are fixed in time.

For the average cost long run case, the algorithm should be altered by correspondingly adjusting the learning step and the update step (see, e.g.,~\cite{mahadevan1996average}).
We discuss the implementation details and results in Section~\ref{sec:Impl}.

%%%%%%%%%%%%%%%%%%%%%%%%%%%%%%%%%%%%%%%%%%%%%%%%%%%%%%%%%%%%%%%%%%%%%%%
\subsection{State aggregation with a TTE constraint}\label{sec:yesTTE}
 %%%%%%%%%%%%%%%%%%%%%%%%%%%%%%%%%%%%%%%%%%%%%%%%%%%%%%%%%%%%%%%%%%%%%%
In this subsection we utilize a similar aggregated MDP formulation to encompass TTE-constraints. 
Since both TTE constrained and unconstrained models are never considered simultaneously, with slight abuse of notation, we will denote the states for the constrained case similarly to the unconstrained one. The connotation will be clear from the context.
Since under a TTE-constraint stored packets are getting obsolete, the suggested state aggregation will incorporate the age of the "oldest" line. In particular, we propose two state aggregations, both of which maintain the number of empty lines and the age of the oldest line, where \textit{Aggregation $I$} also preserves the size of the largest clique encompassing this line, while \textit{Aggregation $II$} keeps the size of the largest clique regardless of whether this clique encompasses the oldest line. %We term the two state aggregations Aggregation I and Aggregation II respectively.
Next, we formally describe the two state aggregation functions which map the detailed state to the corresponding aggregated state; we also design a model-based learning similarly to the case with no TTE constraint.

%%%%%%%%%%%%%%%%%%%%%%%%%%%%%%%
\subsubsection{Aggregation I}
\label{subsec:TTE_agg1}
%%%%%%%%%%%%%%%%%%%%%%%%%%%%%%%
\ifdouble
Define mapping $\phi_I:{\bf M^{TTE}}\to\{{\bf N}\times{\bf N}\times{\bf N}\}$
\[
\phi_I(s)=\{F,C,E\},
\]
\else
Define $\phi_I:{\bf M^{TTE}}\to\{{\bf N}\times{\bf N}\times{\bf N}\}$, $\phi_I(s)=\{F,C,E\}$,
\fi
where $F(s)$ is the lowest strictly positive TTE in $s$, $C(s)$ is the size of the maximal clique, which contains the row with $\tau=F$, and $E(s)$ is the number of empty lines in $s$, where $\tau$ was defined in section~\ref{sec:model}. Note that $C(s)$ is not necessarily equal to $L(s)$, the maximal clique in $s$.
Denote the action space by $\bar A^{I}=\{1,2\}$ where $\bar a\in\bar A^I =1$ stands for sending a coded clique $C(s)$, which contains a line with $\tau=F$, and $\bar a\in\bar A^I =2$ stands for sending an uncoded packet corresponding to a randomly chosen empty line from $E(s)$.

Following the formalization presented in Section~\ref{sec:induced} we define the policy restricted MDP denoted by $\scrM^{I}_1=\scrP(\scrM_0,\phi_I,\bar A^I)$ and the corresponding induced MDP denoted by $\scrM^{I}_2=\scrI(\scrM_0,\phi_I,\hat A^I)$ (see Definition~\ref{def:Policy_Restricted_MDP} and Definition~\ref{def:Induced_MDP}, respectively).

The basic approach for finding an approximately optimal policy under Aggregation $I$, is by harnessing \emph{Algorithm A}.
The corresponding Bellman equation is written similarly to what appears in~\eqref{3}, where the solution is found by substituting the relevant aggregated states.

%%%%%%%%%%%%%%%%%%%%%%%%%%%%%%%
\subsubsection{Aggregation II}
\label{subsec:agg2}
%%%%%%%%%%%%%%%%%%%%%%%%%%%%%%%

Similar to Aggregation I we define a second mapping $\phi_{II}:{\bf M^{TTE}}\to\{{\bf N}\times{\bf N}\times{\bf N}\}$, $\phi_{II}(s)=\{F,L,E\}$, where $E$ denotes the number of empty lines in $s$, $F$ is the lowest strictly positive TTE in $s$, and $L=L(s)$ denotes the size of the maximal clique in $s$. Note that there is no knowledge about the size of the maximal clique containing the line with $\tau=F$, as in Aggregation $I$.
Denote the action space $\bar A^{II}=\{1,2,3\}$, where $\bar a=1$ stands for sending a coded maximal clique $C(s)$, which contains a line with $\tau=F$; $\bar a=2$ stands for sending an uncoded packet corresponding to a randomly chosen empty line, and $\bar a=3$ stands for sending a $L(s)$, maximal coded clique in $s$. Note that the action $\bar a=1$ presumes no prior knowledge about the size of $C(s)$. Thus, the decision in this case is myopic as far as the size of clique being sent is concerned.
The learning in the case of Aggregation II is performed by utilizing algorithm A. We compare by simulations both aggregation types, with an alternative heuristic policy in Section~\ref{sec:Impl}.

%% file: V_Study3.tex
\section{Study of the properties of $V$}\label{sec:vstudy}
In this section, we present an in-depth study of the suggested abstract MDP-based approach by exploring the properties of the value function found through the reinforcement learning procedure.
Our primary objective is to understand the structure of the value function. Namely, we aim to isolate properties of $V(\bar s)$ related to each one of the aggregation parameters. This, in turn, will allow us to incorporate these properties in the main learning algorithm, resulting in improved speed and precision of convergence. Moreover, it will give us better understanding of how each of the parameters (e.g., clique size) affects the results, and how the overall coding process should behave as a function of these parameters. In particular, in some cases, we will observe a {\it threshold type policy} in one of the parameters. That is, a policy in which there is \textit{at most one switching state} from one optimal action to the second. Such a property is desirable as once the switching point is found, we may set the actions to their optimal values \textit{without the need to iterate until the ultimate convergence.} Furthermore, in most cases, such a threshold policy will give a fundamental and rigorous reasoning to very intuitive results, e.g., if sending a coded clique is beneficial for some $L(s)$, it is definitely beneficial for any $l>L(s)$.

%For example,%for fixed $F$ and $C$,
%we are interested in exploring the dependence on $E$, i.e., the number of the empty lines.
For simplicity, we demonstrate the proof of the existence of a threshold-type policy for the 1-dimensional aggregation defined below. %We demon

\input{Aggregation1D5}

The technique demonstrated in the 1-D case can be extrapolated to more complex aggregations. However, the proofs in these cases will involve treatment of significantly more complex Bellman equations.
Alternatively, one may merely assume the existence of a threshold policy, based on observations from simulations.
The main advantage of having the threshold-type policy proof/observation is the possibility to enhance algorithm A, as we explain next.
Assume there exists a threshold policy in $E$, as was presented in Aggregation $I$.
Namely, once for some $E=i$, there is a {\it switch} from optimal action $2$ (transmission of an empty line) to action $1$ (transmission of a clique), then we deduce that $1$ is optimal for all $E<i$ , while $2$ is optimal for all $E\geq i$.
%In some scenarios, the assumptions on threshold policy can be analytically treated and proved.
%To this end, we use the threshold policy observation for the algorithm enhancement.
%Heuristically, the threshold structure may be deduced from policy convergence results of the algorithm for other cases. For example, one can see the policy tendency from cases with lower packet loss, or with smaller state space, e.g., smaller number of users, smaller TTL, etc.
Hence, if existence of a threshold policy in one of the parameters (e.g. $F$,$C$,$E$) is known, at step $4$ of the algorithm, in case the policy in some (possibly rarely visited) state is not yet clear at some point of the algorithm run, correct it according to the already known (or conjectured) threshold rule. This method will accelerate the overall convergence.
%\begin{remark}
%Note that a threshold policy is only clearly seen in the cases where the action space includes two possible actions. For the cases with more than two actions, the thresholds should be more carefully examined.
%\end{remark}
Another useful property of $V$, which gives good understanding of its behavior, is %bounds on the increasing 
its slope. %s of the value function. 
\ifArX
(See Appendix~\ref{app:bounds} for both upper and lower bounds on this slope.)
\else
(See~\cite{2015arXiv150202893S} for both upper and lower bounds on this slope.)
\fi 
Similarly, the bounds can be useful for the manual calibration of the value function in order to speed up the convergence.

%We compare the performance of algorithms defined here in the following section. 

%% file: Aggregation1D5.tex
\subsubsection{One-dimensional aggregation}
As an alternative to the multi-dimensional aggregation patterns, we introduced an even more coarse abstraction. Namely, define $\phi:{\bf M}\to\{{\bf N}\}$, such that $\phi(s)=L(s)$, that is, the size of the largest clique. \MS{Denote a line which is not in the maximal clique  as \textit{e-line}.}
%where $L(s)$ is the size of the maximal clique in matrix representing state $s$ and $N$ is the number of the empty lines.
%When the state $s$ will be clear from the context, we will omit is.
%Observe that the number of unique pairs $\{L,E\}$ is upper bounded by $J=(K+1)K$. %(subtraction of $1$ means that in model without TTL we convene to the rule that cliques of size $1$ are not counted as cliques). %Denote an arbitrary pair $\bar s=(L,E)$ and the set containing all the unique pairs as $\bar S=\{\bar s_i,\;i=1,\cdots,J\}$.
Define a {\it state aggregation} by the set
$\bar s=\{s:L(s)=l\}$, for some given $l$, $l\in\{1,\cdots,K\}$. \MS{The action space consists of two actions, $\bar a=1$ stands for for sending the maximal clique, while $\bar a=2$ stands for sending an e-line. }
While oversimplified, and as such resulting in maybe inferior performance, this aggregation and the induced MDP serve as a good example %for investigation of the corresponding value function.
for which we can investigate the value function and gain important insights.
Proposition~\ref{prop2} below proves the existence of a threshold policy under an average cost. Let $\pi_a$ be a maximizer over all $\pi$ in~\eqref{ind1}. That is:
 $\pi_{a}=\argmax_\pi\lim_{N\to\infty}\frac{1}{N}\E\big[\sum_{t=0}^{N}r(s'_{t+1},a^{\pi}_{t},s_{t})\big]$
\begin{proposition}\label{prop2}
There exists an optimal policy which is threshold policy in the size of the maximal clique. Namely, there exists a constant $k,\;k\in\{2,\ldots,K\}$ such that for $0\leq L(s)<k$ and $s\in\bar s$ it holds %\par \vspace*{-1.5em}
$
\bar a(\bar s)=2
$,
yet for $k\leq L(s)\leq K$ and $s\in\bar s$,  we have
$\bar a(\bar s)=1$.

\end{proposition}
That is, send the maximal clique (a coded packet) if and only if its size is at least $k$. Otherwise, send an e-line (an uncoded packet).
%The proof of Proposition~\ref{prop2} is in Appendix~\ref{App:1-DProof}.
 %As for the value function of the discounted cost, it can be closely estimated, see, e.g., the analytical method in~\cite{hordijk2002blackwell} and references therein.

%We bring first the intuitive explanation to the proposition and provide the formal proof for the average long run cost.

%\textbf{Proof of Proposition~\ref{prop2}.}

We will need the following notation for the proof of Proposition~\ref{prop2}.
We say that a state $s$ is \textit{recurrent under the policy $\mu$} if when starting at state $s$ and acting according to $\mu$, the probability to return to $s$ is $1$. A state which is not recurrent under $\mu$ is \textit{transient under $\mu$}.

Consider a policy $\pi^*$, which is optimal for the average long run cost,
\ifdouble
%\[
%V=\min_\pi J_\pi
%\]
\[
\pi_*=\argmax_\pi J_\pi,
\]
\else
$\pi_*=\argmax_\pi J_\pi$,
\fi
where $J^\pi$ is given in~\eqref{ind1}.
Denote a set of states $S_1\subset S$ such that $s^{(i)}\in S_1$ if $a_{\pi_*}(s^{(i)})=1$.
Denote a state $s^{(m)}$, such that $s^{(m)}\in S_1$ and $L(s^{(m)})<L(s^{(i)})$, $\forall i, s^{(i)}\in S_1$. Namely, $S_1$ is the set of states for which sending a clique is optimal, and $s^{(m)}$ is the state with the minimal maximal clique in $S_1$ -  for which it is optimal to send the maximal clique.
We have the following claim.
\begin{claim}\label{lem6}
Any state $s^{(i)}$ such that $L(s^{(i)})>L(s^{(m)})$ is transient under $\pi^*$.
\end{claim}
\begin{proof}
We use the fact that nodes do not use \textit{coded} packets in order to decode packets \textit{not intended to them}. Namely, nodes store only uncoded packets intended for other users.
%coded packets which are not destined to the receiver are not stored at them.
Hence, clique transmissions cannot increase the clique size, and, moreover, decrease it with some non-zero probability (note that transmission of an e-line can increase the clique size, yet by at most $1$). Consider some $s^{(i)}\in S_1$. By definition $L(s^{(i)})>L(s^{(m)})$.  Since $p(s^{(j)}|s^{(i)},1)>0$, where $j\leq m$, the state $s^{(m)}$ will be reached in finite number of transmissions. Furthermore, the states with clique size more than $m$ will not be attended afterwards. That is, once in $s^{(m)}$, the future state can not be increased.
Consequently, for any $s^{(i)}$ such that $L(s^{(i)})>L(s^m)$, $s^{(i)}$ is transient under $\pi^*$.
\qed
\end{proof}
Note that the claim holds even if $\pi^*$ is not the optimal policy.

\begin{proof}[Proposition~\ref{prop2}]
%The intuition for the difference can be tracked in Figure~\ref{fig5}
Consider a policy $\pi^*$, which is optimal for the average long run cost,
%\ifdouble
%%\[
%%V=\min_\pi J_\pi
%%\]
%\[
%\pi_*=\argmax_\pi J_\pi,
%\]
%\else
%$\pi_*=\argmax_\pi J_\pi$,
%\fi
%where $J_\pi$ is given in~\eqref{ind1}.
%\[
% J_\pi=\lim_{N\to\infty}\frac{1}{N}\sum_{n=0}^Nr_\pi(s(n),a(n))
%\]
a set of states $S_1\subset S$ and $s^{(m)}$ as above.
%Denote a state $s_m$, such that $s_m\in S_1$ and $L(s_m)<L(s_i)$, $\forall i, s_i\in S_1$. Namely, $S_1$ is the set of states for which sending a clique is optimal, and $s_m$ is the state with the minimal clique in $S_1$ - the minimal clique size for which it is optimal to send a clique.
Denote the set $S_r$ such that $s^{(i)}\in S_r$ if $L(s^i)\leq L(s^{(m)})$, and denote $S_t=S\backslash S_r$.
Now see that by the claim above,
$s^{(m)}$ is the only recurrent state in $S_1$.
Define $n_m$, the first time under $\pi^*$ to be in $s^m$. We have
\ifdouble
\begin{align*}
& V=J^{\pi^*}=\lim_{N\to\infty}\frac{1}{N}[\sum_{n=0}^{n_m-1}r_{\pi_*}(s_{n},a_{n})\\
&\;\;+\sum_{n=n_m}^{N}r_{\pi^*}(s_{n},a_{n})].
\end{align*}
\else
\begin{align*}
& V^{\pi^{AC}}=J^{\pi^*}=\lim_{N\to\infty}\frac{1}{N}[\sum_{n=0}^{n_m-1}r_{\pi^*}(s_{n},a_{n})
+\sum_{n=n_m}^{N}r_{\pi^*}(s_{n},a_{n})].
\end{align*}
\fi
Observe that all states encountered at times $n>n_m$ are recurrent. That stems from the fact that after the transmission at time $n_m$, the process stays in $S_r$. Since $n_m$ is finite a.s., the first sum (once normalized by $N$) goes to zero.
Next, define policy $\pi^m$ which acts similarly to $\pi^*$ for all $j$ such that $L(s^{(j)})\leq L(s^{(m)})$ (that is, all recurrent states) yet sets $a(s^{(j)})=2$ otherwise. That is, a threshold policy. Denote by $n_l$ the first time to hit $s^{(m)}$ under $\pi^m$.
Observe that
\ifdouble
\begin{align*}
&\lim_{N\to\infty}\frac{1}{N}\sum_{n=n_m}^{N}r_{\pi^*}(s(n),a(n))\\
&\;\;=\lim_{N\to\infty}\frac{1}{N}\sum_{n=n_l}^{N}r_{\pi_m}(s(n),a(n))=V^{\pi^{AC}}
\end{align*}
\else
\begin{align*}
&\lim_{N\to\infty}\frac{1}{N}\sum_{n=n_m}^{N}r_{\pi^*}(s_{n},a_{n})=\lim_{N\to\infty}\frac{1}{N}\sum_{n=n_l}^{N}r_{\pi^m}(s_{n},a_{n})=V^{\pi^{AC}}
\end{align*}
\fi
Thus $\pi^m$ is also an optimal policy. Note that the relation between $n_l$ and $n_m$ is not essential, since both are finite.

It is left to show that the policy which always sends e-lines, that is, sends no cliques at all is suboptimal. Denote such a policy as $\pi^e$. However, in such a policy the expected reward at each step is given by $1-p$, and any other policy which sends a clique at any step outperforms $\pi^e$ by some $\epsilon>0$.
This accomplishes the proof of the proposition.
\qed
\end{proof}
The proposition above is intuitive, since the clique size can only be increased by $1$. This renders all states with the maximal clique larger than the threshold to be, in the long term, unreachable. 
%In addition, observe that the probability to increase the maximal clique size is rapidly decreasing as the clique size grows.
%Thus, intuitively, the cumulative reward which is associated with sending an e-line will have a marginal component which is decreasing and a constant component related to the immediate reward of decoding the transmitted uncoded packet. This is contrasted to the cumulative reward which is associated with sending a maximal clique. Clearly, in the latter case, the reward is constantly growing in the size of the clique. In other words, 
%Hence, we keep sending e-lines while the probability to increase the clique size is still significant, yet, when this probability decreases and it is more beneficial to collect the reward of a clique, we send the clique. This provides an additional intuition.

\MS{Note that Puterman~\cite{Puterman} gives general guidelines how to demonstrate the monotonicity of the optimal policy, both for the average cost and the discount cost infinite horizon criteria. %However, as 
%However the transition probabilities are difficult to calculate, the properties needed for such a proof a difficult to demonstrate. Hence 
Here, we merely presented the short proof which specifically suits this simple case.}
%The equivalent statement for the discounted case requires knowing the transition probabilities, which are hard to calculate. %and we do not provide it in this paper.

The connection between average and discounted costs, is well-known and is described by the Blackwell optimality condition~\cite{bertsekas2005dynamicI}. In particular, Blackwell optimal policy is optimal for the average cost as well. Yet, as seen from the proof of Proposition~\ref{prop2}, the optimal policy for the average cost, in this case, is not unique. Hence, the opposite is not necessarily true. Nevertheless, we address this in the simulations.

%% file: Simulation2.tex
\section{Simulation results}\label{sec:Impl}
In this section, we evaluate the suggested transmission strategy through extensive MATLAB simulations. Our simulation results provide insight on the impact of each of the mechanisms described throughout the paper.  Specifically, we thoroughly examine the effect of different parameters such as TTE and packet loss probabilities on the value function or on the policy structure. In addition we evaluate our algorithm and compare the different aggregations suggested.

In our simulations we consider a single cell comprising an AP and $K$ receivers. Since our results relate to the traffic from the AP to the users, our simulations only consider the downstream traffic. We assume that all $K$ users have pending traffic waiting to be transmitted. An \emph{i.i.d} Bernoulli channel error is assumed, where each packet transmission is received or dropped by each user with probability $1 - p$ and $p$, respectively, and is independent between different transmission attempts. The AP works according to Algorithm $A$ with corresponding aggregation. In all cases compared, the AP activates the learning routine considering the discounted infinite horizon cost. Thus, it computes the values attained by value functions for all possible initial states. We later use the same policy for calculating the long run average cost. Note that based on the Blackwell optimality argument (e.g.,~\cite{bertsekas2005dynamicI}), if $\gamma\to1$,   under mild conditions the policy which is optimal for the discounted problem is optimal for the average cost problem as well. The number of iterations for each phase (learning and improvement) is set in accordance with the specific configuration.

\subsection{Results without a TTE constraint}
We start by evaluating the policy resulted from our learning algorithm, for the proposed aggregation in the case of no TTE constraint (Section~\ref{sec:st-agg}). We compare our results with the bounds obtained in~\cite{wang2012capacity}. The aggregation for the TTE-unconstrained case constitutes a 2-dimensional state space, namely, the size of the maximal clique $C$ and the number of empty lines $E$ (Section~\ref{sec:st-agg}). The action space comprises two possible actions, transmitting to a user that its packet was not received by any user (empty line in the state matrix) and transmitting to the maximal group of users in which each member of the group has a packet destined to every other user in the group (maximal clique in the state matrix).
The performance results \MS{(i.e the percentage of successfully decoded packets, using the retransmissions)} are seen in Figure~\ref{fig5}(top) along with comparison to the bound from~\cite{wang2012capacity}. The bound is derived for systems with \emph{much stronger coding capabilities}, hence any potential scheme, theoretical or practical as can be, cannot attain better performance. Denote it as the \textit{Wang upper bound}. Note that in order to calculate the bound one needs to solve $120$ inequalities, hence the graph has small discrepancies. For larger systems, such calculations may be too complex.
As for the optimal policy, the simulation results show that is the same regardless of the packet loss probability. In particular, the optimal policy is defined by transmitting a random empty line whenever there are empty lines ($E>0$) and transmitting to the maximal clique otherwise. Accordingly, the obtained policy is a threshold-based policy. The intuition behind this strategy is clear: the reward associated with both possible actions, transmitting a random empty line or transmitting the maximal clique, is time independent, i.e., the expected reward is the same if the transmission occurs now or in one of the following transmission opportunities. Moreover, since any empty line is not included in any clique all the more so in the maximal clique, yet transmitting an empty line can potentially increase the size of a clique without incurring any penalty for delaying the current maximal clique transmission, it is worthy to fill in the state matrix such that no empty lines are left, and only then to transmit the maximal clique. Note that this policy \emph{coincides with the one heuristically suggested} in~\cite{cohen2013coded} denoted as the \textit{semi-greedy algorithm} (SG). Accordingly, the simulation results imply that under the restricted action space described above, the semi-greedy algorithm \cite{cohen2013coded} is optimal, as long as no TTE constraints are applied. Moreover, for the simple case of 2-users system, these results \emph{achieve the sum-capacity} which is found according to~\cite{georgiadis2009broadcast} and~\cite{wang2012capacity}.
Figure~\ref{fig5}(down) shows results (value functions at all states) for differentiated packet loss. One sees that the case with equal packet loss for all users achieves the lowest value function vector. The highest values are obtained for the case where two of the five users have relatively low packet loss ($0.1$), while the other three users have relatively high packet loss (more than $0.4$). %This is explained by that the lossy users tend quickly to form a clique, while the reliable users keep successfully receiving uncoded packets. In overall, the performance is tangibly increased, however fairness concerns arise.}
This is explained by that the lossy users tend quickly to have a pending packet stored at reliable users. Hence, the lines corresponding to these users are most probably not empty while reliable users keep successfully receiving uncoded packets. A clique will be sent when some of the reliable users will not receive their packet forming a large enough clique for transmission. In overall, the performance is tangibly increased, but the throughput improvement comes at expense of hampered fairness. %however fairness concerns might arise.

\ifdouble
\begin{figure}
%\begin{center}
%\begin{align*}
{\includegraphics[angle=0, width=0.40 \textwidth]{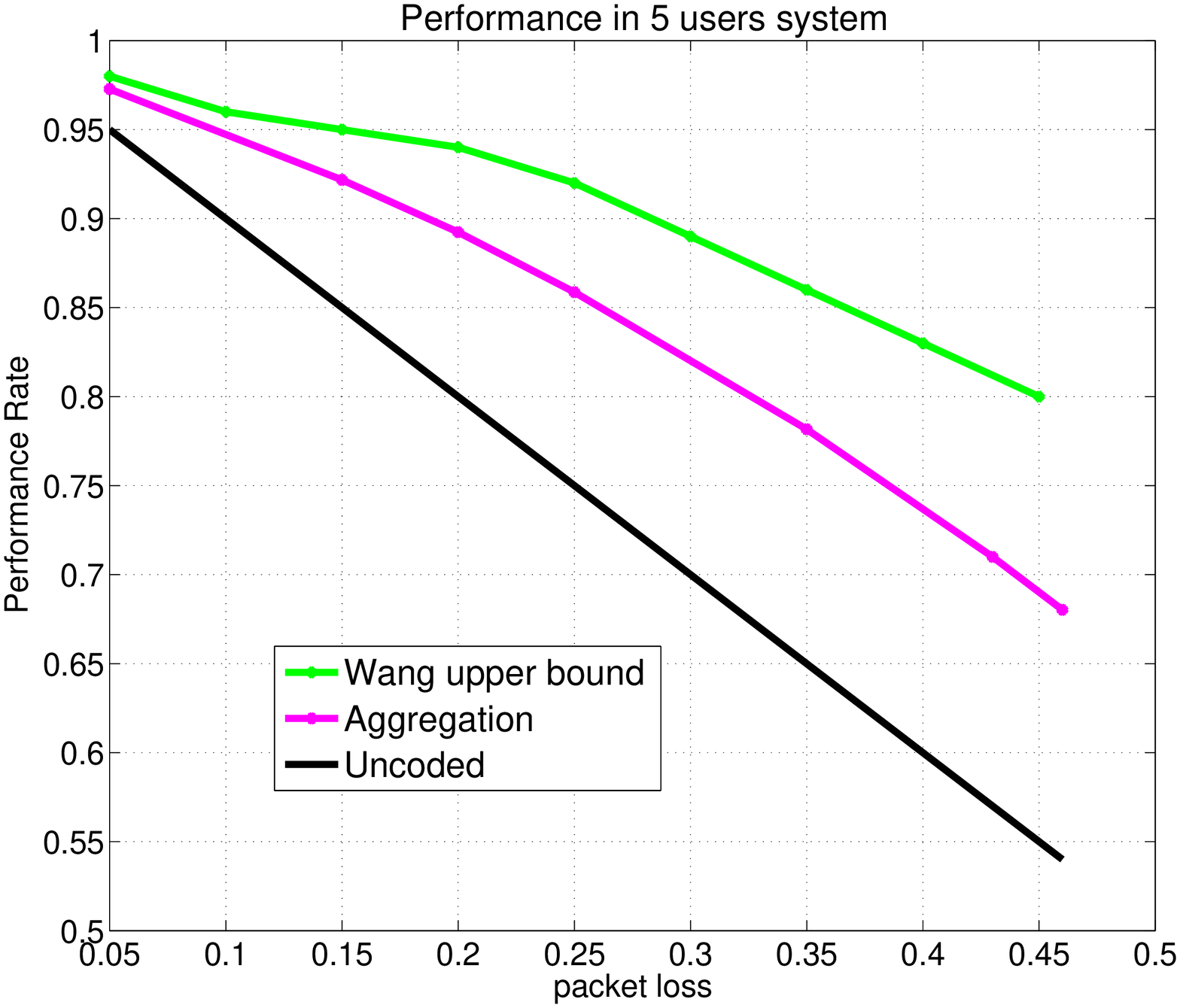}} %
%\end{align*}

%\hfill
\caption{\sl\tiny
System of 5 users. Uncoded system is compared with simple aggregation and with maximal achievable capacity for the general transmission method. } \vspace*{-10pt}}\label{fig5}
%\end{center}
\end{figure}
\else
%\lipsum[1]

\begin{wrapfigure}[19]{R}{0.4\textwidth}
	\vspace*{-40.0pt}%\hspace*{-20.0pt}
  \begin{center}
    \includegraphics[width=0.4\textwidth]{Cap5.eps}
    \includegraphics[width=0.4\textwidth]{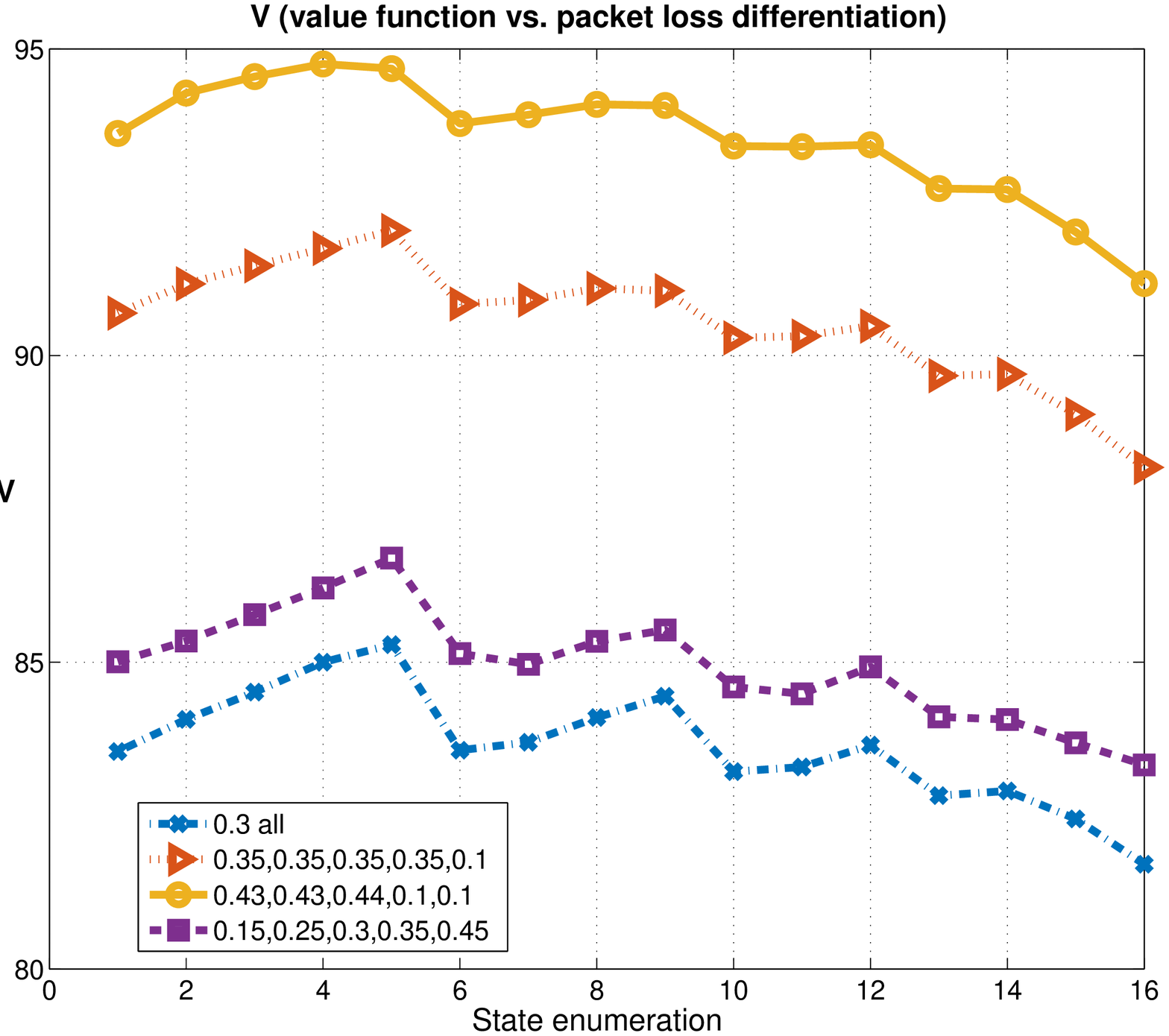}
  \end{center}
  \vspace*{-35.0pt}
  \caption{\sl System of 5 users results with no TTE constraint}\label{fig5} %Uncoded system is compared with simple aggregation and with maximal achievable capacity for the general transmission method. }

\end{wrapfigure}
%\lipsum[2-3]

\fi

\subsection{Results for TTE constrained aggregations}
Next we evaluate the performance of the suggested transmission strategy under TTE constraints.
%specifically we evaluate ***FIXME  X,Y and Z - Algorithm I, aggregation, learning? ****.
%In order to get some insight into the policy restricted MDP we examine the reward function for the induced MDP.
%As before we only considered the downstream traffic, for which we assumed an AP is sending data to fully backlogged $K=5$ to $10$ users.

We simulated \emph{Aggregation I} (Section~\ref{sec:st-agg}), aiming to examine the structure of the value function for all feasible states.
Namely, we try to %observe structural properties of the $V$,
%in order
to understand the effect of different parameters on $V(\hat s)$. Our objective was to identify simple properties such as monotonicity, convexity and threshold-type structure.  Such properties can be potentially utilized for the RL convergence speed-up. This will allow to successfully operate larger systems.
We examined a system with $K=5$ receivers. We set $\gamma=0.99$.
%The Bernoulli channel error is set to $p=0.25$ independent between users and between different transmission attempts.
 %which characterizes the aggregated states by the tuple \emph{ \{lowest TTE ; maximal clique size containing the row with lowest TTE ; number of empty lines \}  }
%denoted by $\bar s\in {F;C;E} $ .
%The discount factor was set to . %which ***FIXME does not penalize severely for delaying a transmission***. We ran the simulations for ***FIXME after …***.
The results are depicted in Figure~\ref{fig2}. The $Y-axis$ depicts the value attained by each state, $ V(F;C; E)$, (denoted by asterisks). Each value corresponds to the given initial state. $X-axis$ relates to an enumeration of the states, $\{1,2,\cdots\}$. %Aggregation I, a graphical representation of the value function. The horizontal axis denotes the state enumeration $\{1,2,\cdots\}$. The vertical axis denotes the discounted infinite horizon values given the corresponding initial state.
Note that the asterisks form groups of monotoneous patterns of values. In particular, the states are assigned numbers which grow first in TTE ($F$), next with maximal clique size ($C$) and finally they grow with the number of empty lines ($E$). For example, state 1 refers to the state in which there are no empty lines, maximal clique size 1 and $TTE = 9$, State 2  relates to the values of the state in which there are no empty lines, the maximal clique size contains the line with the greatest TTE is 8, state 96 which is the last state refers to the state in which there are 5 empty lines (i.e. the empty matrix)

Note that for the widespread (e.g., 802.11) policy that only allows uncoded transmissions the value is fixed $\frac{1-p}{1-\gamma}=\frac{1-0.25}{1-0.99} = 75$, which is below the scale of the graph, i.e., the value for all states is higher than the one for the uncoded ARQ retransmissions.
\ifdouble
\begin{figure}
\begin{center}
%\begin{align*}
\centerline{\includegraphics[angle=0, width=0.43 \textwidth]{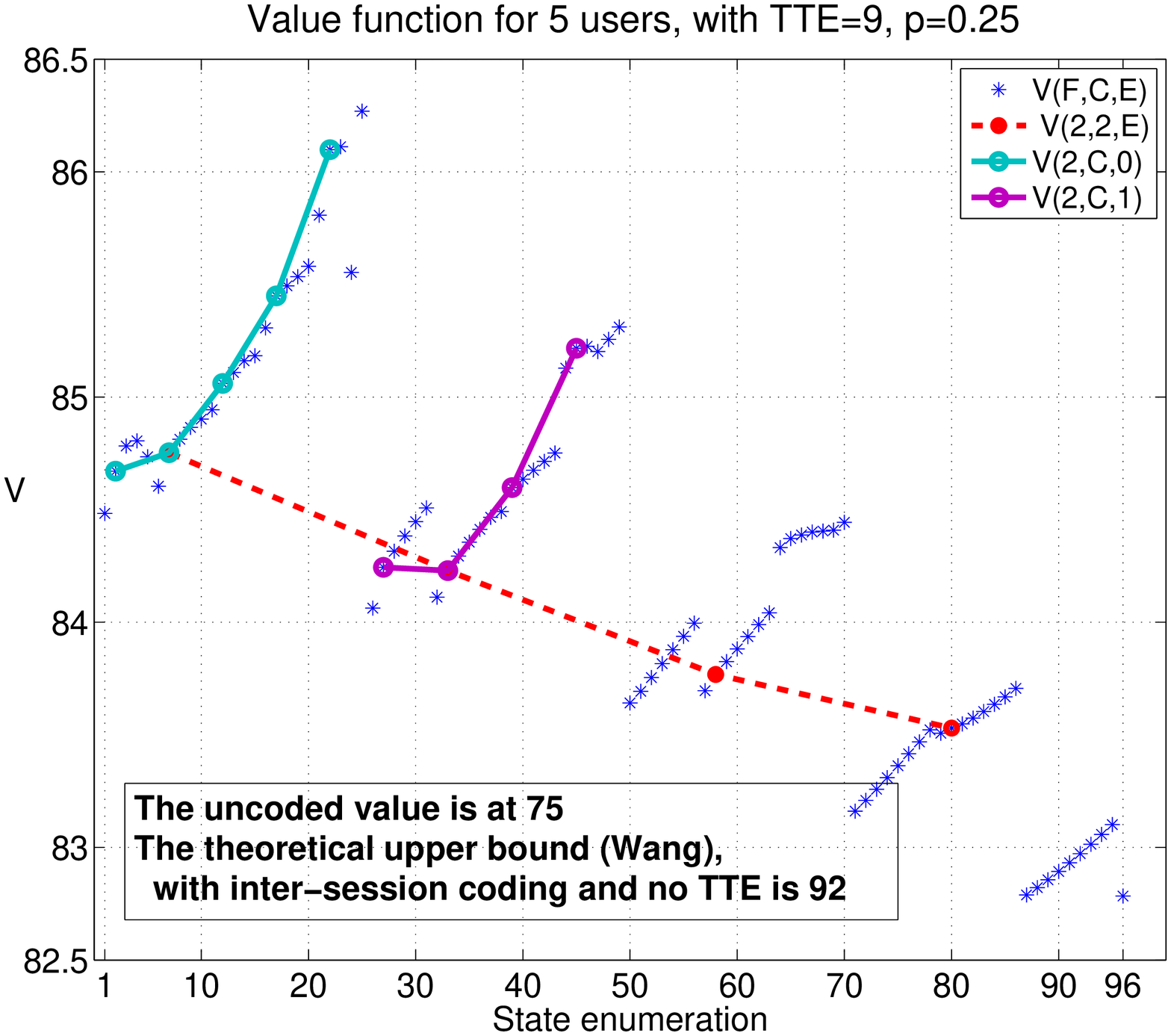}} %
%\end{align*}
\end{center}
\vspace{-20pt}
%\hfill
\caption{\sl\small
Aggregation I, a graphical representation of the value function. The horizontal axis denotes the state enumeration $\{1,2,\cdots\}$. The vertical axis denotes the discounted infinite horizon values given the corresponding initial state. Observe that the dots form groups of ascending patterns of values. Each group represents the number of empty lines. The group with $E=0$, that is $V(F,C,0)$, is near $10$, $V(F,C,1)$ is near $30$, $V(F,C,2)$ is near $60$, $V(F,C,3)$ is near $80$, $V(F,C,4)$ is near $90$ and the lowest isolated state stands for the empty matrix.}\label{fig2}
\vspace*{-10pt}
%\end{center}
\end{figure}
\else

\begin{wrapfigure}[19]{H}{0.43\textwidth}
		\vspace*{-40pt}
\begin{center}
%\begin{align*}
\includegraphics[width=0.43 \textwidth]{ALL_isit1.eps} 
\includegraphics[width=0.43\textwidth]{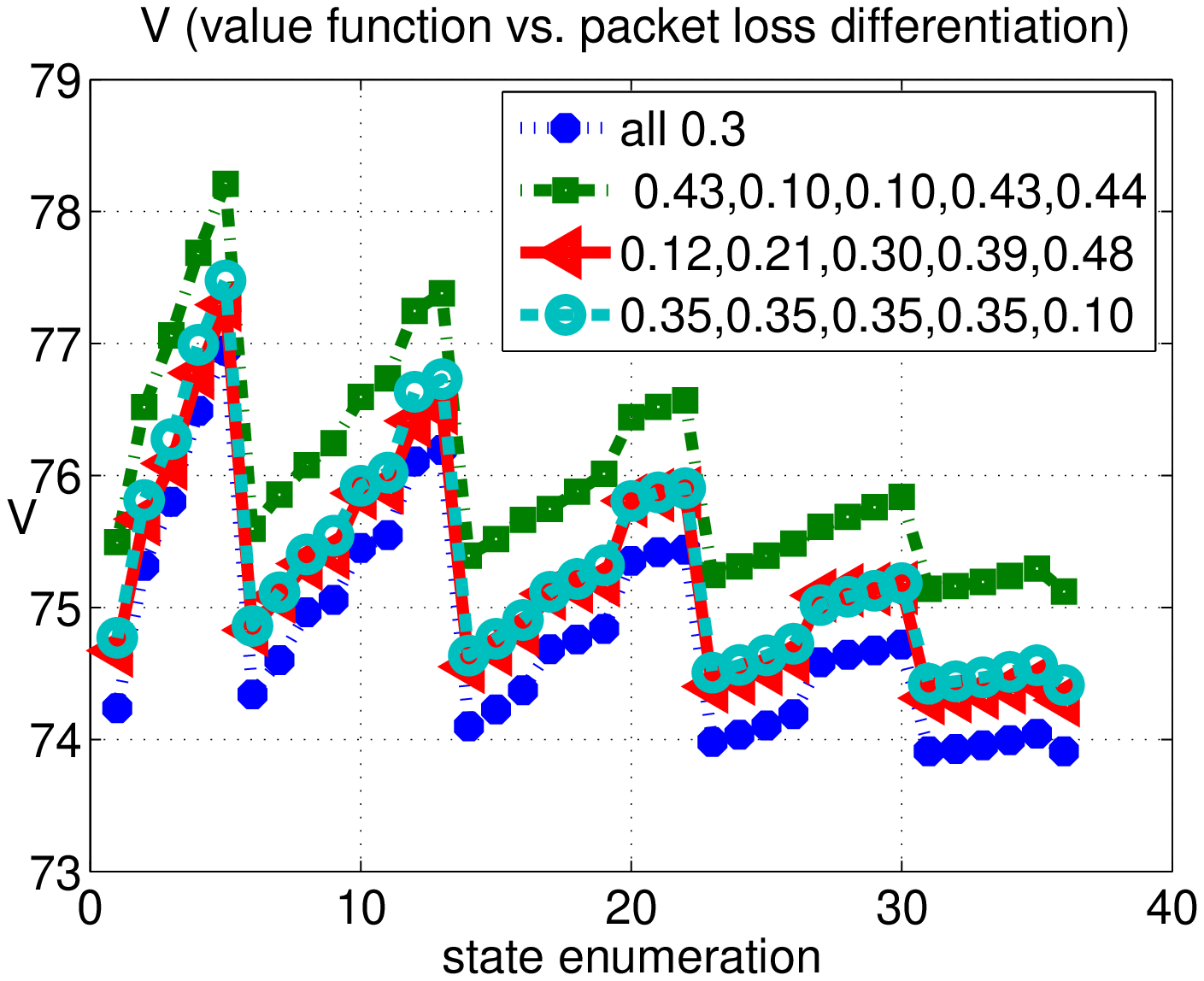}
%\end{align*}
\vspace{-50pt}
%\hfill
\end{center}
\caption{\sl
Aggregation I. Each group of asterics represents the number of empty lines. The group with $E=0$, that is $V(F,C,0)$, is near $10$, $V(F,C,1)$ is near $30$, $V(F,C,2)$ is near $60$, $V(F,C,3)$ is near $80$, $V(F,C,4)$ is near $90$ and the lowest isolated state stands for the empty matrix(top). Effect of differentiated packet loss (down)\vspace*{-1pt}}\label{fig2}

\vspace{-50pt}
%\end{center}
\end{wrapfigure}

\fi

We emphasized the structure of the value function when only a single parameter varies while holding the other two are fixed.
Specifically, in order to understand the effect of empty line on the obtained policy, we emphasize by the dotted (red) line the states in which the TTE and the size of its corresponding clique are constant, specifically $F=2$, $C=2$, and the number of empty lines varies ($0 \leq E \leq 3$). %; obviously, this setup defines only ***FIXME 4 *** feasible states, yet we have connected them to emphasize the trend.  %As can be seen in the figure, the more empty lines state $s$ has, the lower its value function.
This can be intuitively explained by the property that lines which are non-empty contain some information that potentially can be exploited in future transmissions, while the empty lines contain no information whatsoever. In addition, in order to demonstrate the value function dependence on the clique size, we emphasize the states in which TTE is fixed and equals 2 ($F=2$), number of empty lines is fixed (we show two different values), and the clique size varies. Observe $V(2;C;0)$ and $V(2;C;1)$ which are represented by the solid cyan and the solid magenta lines, for $E=0$ and $E=1$, respectively. As expected, both lines have an increasing pattern with $C$, i.e., the greater the maximal clique which corresponds to the line with lowest TTE, the greater the value function. %(again note that the continuous line is only for emphasizing there are only a limited number of feasible states for each setup).
By observation, one can also assume that the value function has a convex increasing form in $C$ (cyan and magenta lines) and convex decreasing in $E$ (the red line).

%In addition the
%In fact, as far as the proof of the former property is concerned, the $max$ operator indeed conserves the convexity. However, the proof of the boundary conditions is not trivial in this case.

%The observation

The effect of the differentiated packet loss is demonstrated in Figure~\ref{fig2}(down). We compared four different packet loss distributions, with average value equal to $0.3$. Similarly to the case with no TTE constraint, the best throughput is achieved where packet loss was with highest variance. However the difference was significantly less visible, which is clearly understood from the TTE constraint, since with TTE will limit the number of packet sent by the AP before sending a clique incorporating the lossy users pending packets. %Note that the throughput improvement comes at expense of hampered fairness. 
Note that for the same reasoning, also the fairness issue is less acute. For example, in the case where most reliable user had packet loss equal to $0.12$ while the most lossy one had packet loss equal to $0.48$, the ratio of the number of sent packets by the AP was $7:4$ in favor of the reliable user.

%which is clearly understood from the TTE constraint, since with TTE will limit the number of packet sent by the AP before sending a clique incorporating the lossy users pending packets. Note that for the same reasoning, also the fairness issue is less acute. For example, in the case where most reliable user had packet loss equal to 0:12 while the most lossy one had packet loss equal to 0:48, the ratio of the number of sent packets by the AP was 7 : 4 in favor of the reliable user.

We explore next the dependence of the policy found for Aggregation I on various parameters, at equal packet loss which ranged from $5$$\%$ to $35$$\%$. The results are shown in Figure~\ref{fig1}.
For reference convenience, the first column denotes the state enumeration. Recall, that $1$ stands for sending the maximal clique containing the oldest line, while $2$ stands for transmitting a random empty line. %The results and differences between the policies found for various packet loss values are brought in

These results clearly demonstrate that the algorithm %overcomes the packet loss uncertainty 
converges to the optimal policy in accordance with the channel condition.
As for the threshold-type policy, the proof of this property is hard to accomplish, as it relies on the transition probabilities, which are hard to attain.
However, the threshold-type property, can be observed by simulations, as it is seen from the table (see states (20-22), (27-29).) Note that the property can highly accelerate the RL procedure. %, in the $F$, the TTE, of the oldest line
%See that the Bellman equation in~\eqref{3} has the form of $V(F,C,E)=\max\{S(F,C,E),G(F,C,E)\}$, where $S$ stands for the component related to the oldest maximal clique transmission  and $G$ relates to component associated with an empty line transmission. Note that both components constitute linear combination of various states of the value function.
%Fix some $F=f, E=e$. The selection of states in $S(f,C,e)$ and $G(f,C,e)$ varies with $C$ in accordance with the values chosen for $f$ and $e$.
%The threshold policy in $C$ will hold in the case the sequence $Z(f,C,e)=S(f,C,e)-G(f,C,e)$ is changes the sign from negative to positive at most once. (The stronger condition is that $Z$ is increasing in $C$, but it is not necessary.)
As explained in Section~\ref{sec:st-agg} the transition probabilities are approximated by RL. Hence, simulation-based exploration is imminent in order to identify structural properties. %\MS{Once the transition probabilities are approximately known from the numerical experiment one can try to establish structerul properties (e.g. superadditivity) which are needed in order to show the monotonity of the policy, see~\cite{Puterman} for the related results.}
 %Indeed, we could see the threshold policy is clearly seen from the table.
Alternatively, one can attempt to prove the threshold property for the average long run case, as we proved for the 1-D case in Section~\ref{sec:vstudy}.
Note that as long as all three dimensions of $V(\hat s)$  are viewed, the thresholds are expected to form three-dimensional surfaces. %For this property to hold, the standard convexity, even if proven, is not sufficient, as one should prove other properties, such as multi-modularity~\cite{koole2007monotonicity}
%Additionally, one can explore lower and upper bounds on the difference $V(\cdot,C,\cdot)-V(\cdot,C-1,\cdot)$. %in a similar way to the analytical results for the 1-D aggregation case.

%All these results and observations can be effectively exploited for the reinforcement learning speedup, as we discuss in the final section.

We conclude the observations above by proposing an effective speedup for Algorithm $A$. %The proposed change is symbolically termed here by Algorithm $B$.
The proposed enhancement stems from simulation results and by the previously discussed properties of value function in section~\ref{sec:vstudy}.
First, in order to successfully operate a larger system, one can solve a (trial) system with small number of users with the same aggregation and the same channel conditions. Next, the resulting optimal policy can be extrapolated in order to get the policy for the desired system, for example, threshold and monotonicity patterns, as we examined above.
In particular, define an approximating policy $\pi_X^0$ using an assessment based on the policy found from a smaller system and the observed properties. Heuristically, this policy should allow a randomization around \textit{conjectured} threshold states. Next, an adjustment of $\hat V_i$ and that of $\pi_X^{k+1}$ is heuristically performed. Again, this improvement can be done using the estimated properties of the value function, or can be combined within the regular run of the reinforcement learning as it appears in Algorithm $A$. \MS{See also monotone policy iteration algorithm in~\cite{Puterman}}.

In order to evaluate the effect of TTE on the policy, we compare both Aggregation I and Aggregation II with the greedy and semi-greedy algorithms proposed in [16]. Specifically, the greedy algorithm aims at maximizing the instantaneous reward received for each transmission opportunity. Hence, the policy according to the greedy algorithm is to transmit the maximal clique for each transmission opportunity. Whenever there is no clique (i.e., $C \leq 1$) transmit a random empty line. The semigreedy (SG) policy is defined in the subsection above. Figure~\ref{fig4} (left and middle) compares the value function of the discounted infinite horizon cost with a zero matrix as the initial state for the various policies. %and Figure 3 (right) compares the value functions for the average cost long run of the different policies, where we used the policy found in the discount case.

Figure~\ref{fig4} (left) clearly depicts that as expected  under the TTE constraints the semi-greedy algorithm performs almost as poorly as the uncoded policy. This is explained by that it does not take into account lines which can be discarded, hence misses clique transmission opportunities just for trying to fill the matrix with non-empty lines.
Moreover, in system where the number of users is greater than TTE, the AP will never be able to fill the state matrix with non-empty lines and the aforementioned semi-greedy algorithm coincides with the uncoded algorithm which sends only uncoded packets. Hence, we devised an alternative heuristic algorithm, termed modified semi-greedy (MSG). MSG differs from SG in that whenever there is a line in which the TTE is going to expire on the next slot (i.e., TTE = 1) the AP transmits the maximal clique containing the oldest line. The results of the MSG heuristic are also depicted in Figure~\ref{fig4}. Note that MSG is indifferent to the channel conditions and acts identically for any packet loss (Figure~\ref{fig4} left).
Further note that even though both policies rely on the same parameters to make a decision, i.e.,  both perform based on the triplet $\{$oldest line, maximal clique size, number of empty lines$\}$, Aggregation II outperforms the MSG algorithm at all packet loss values. This can be explained by that MSG, while being effective as a simple heuristic algorithm, neglects the channel condition, i.e., MSG provides only a single retransmission opportunity for a packet before it gets obsolete, regardless the loss probability. This is opposed to Aggregation II which effectively \emph{adjusts the policy to the channel packet loss} with no prior knowledge on the packet loss ($p$), \emph{based on the on-line learning}. %performed, utilizing the actual channel output and decoding results.
Indeed, the advantage of Aggregation II becomes more prominent at higher packet loss values, as can be seen in Figure~\ref{fig4}.

\ifdouble
\begin{figure}[h!]
\begin{center}
%\begin{align*}
\centerline{\includegraphics[width=25em]{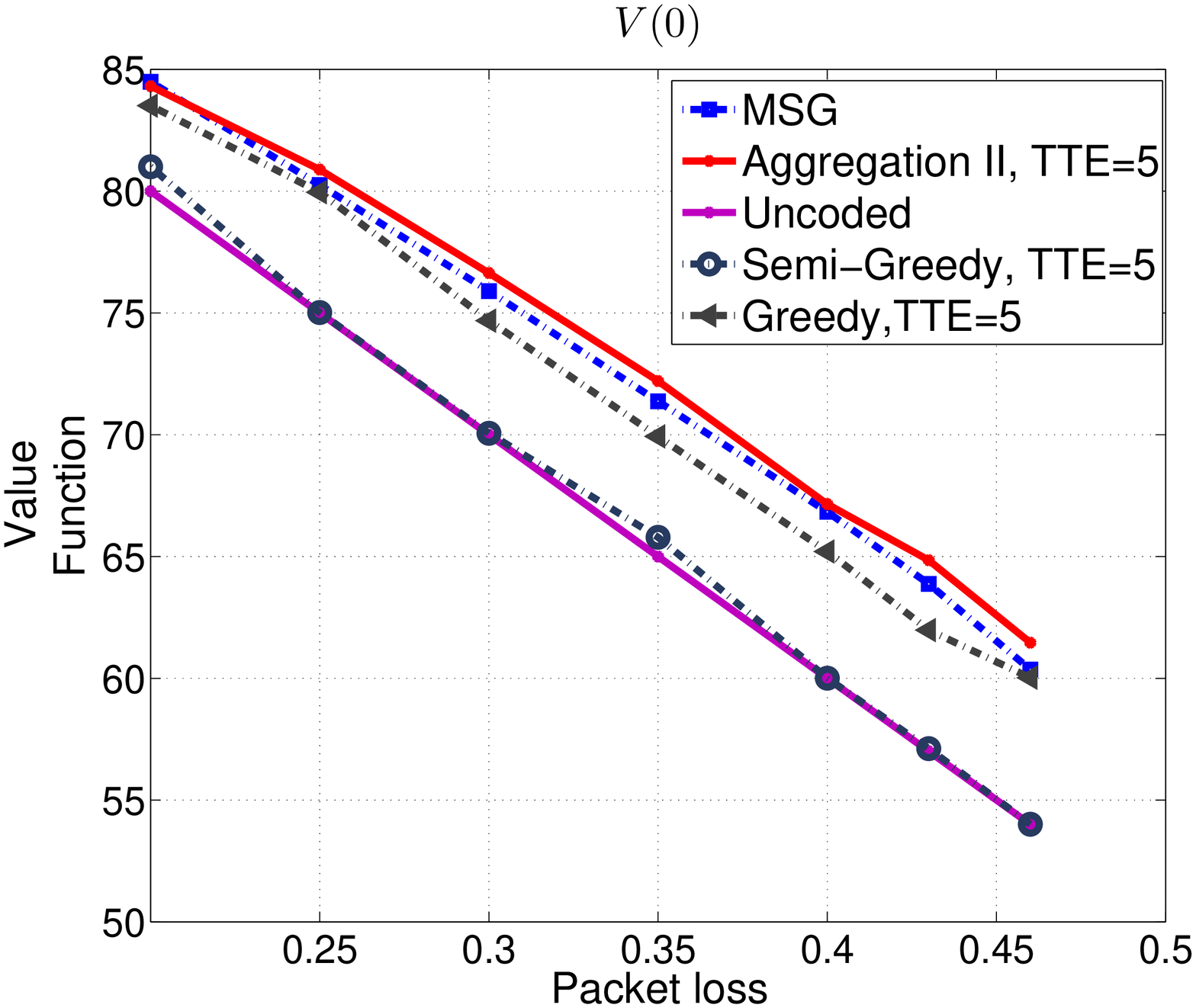}}
%\end{align*}
\label{fig4}
\caption{\sl\small
Discounted value function comparison.}
\end{center}
\end{figure}

\begin{figure}[h!]
\begin{center}
%\begin{align*}
\centerline{\includegraphics[width=25em]{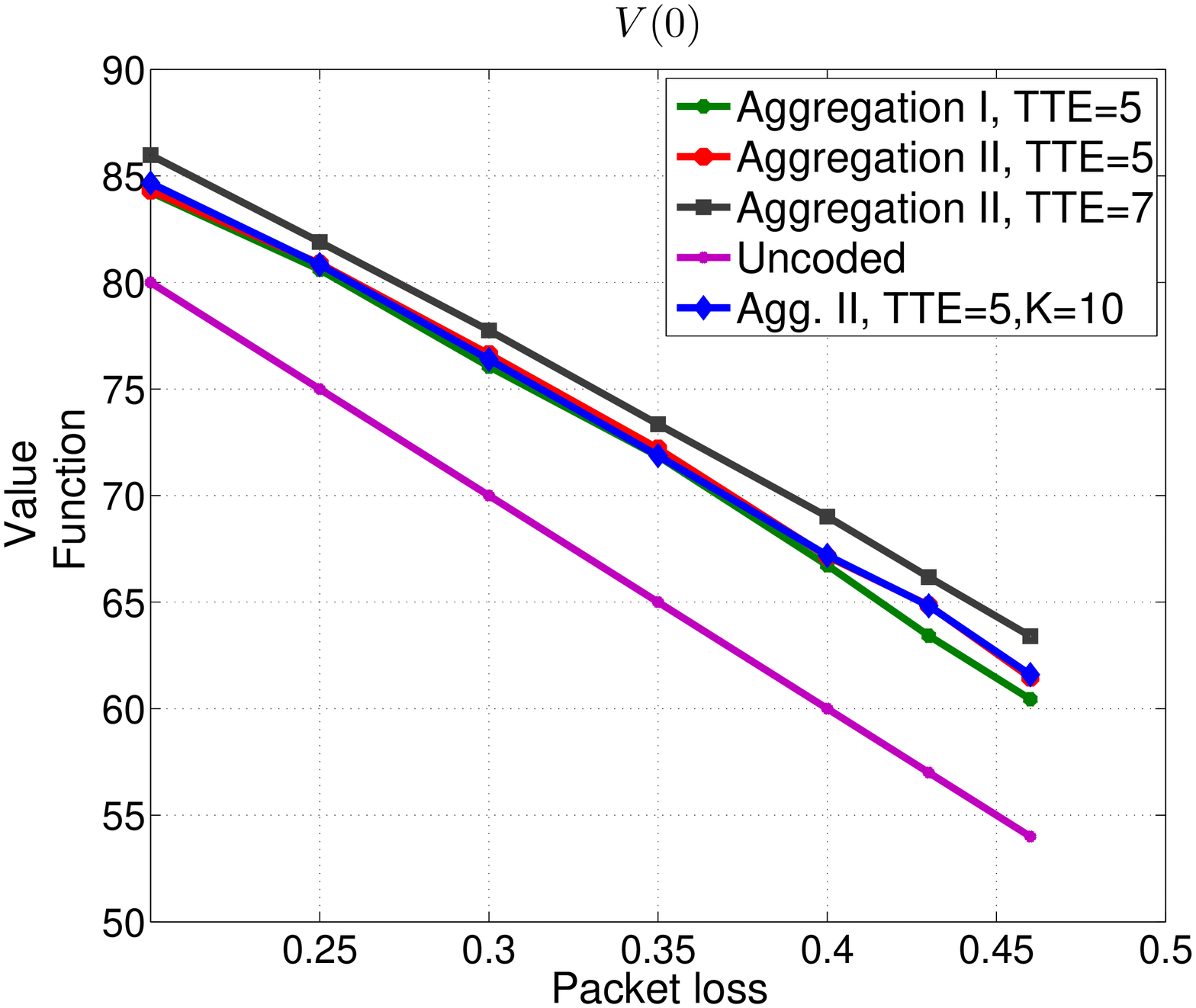}}
%\end{align*}
\label{fig4}
\caption{\sl\small
Discounted value function comparison.}
\end{center}
\end{figure}

\begin{figure}[h!]
\begin{center}
%\begin{align*}
\centerline{\includegraphics[width=25em]{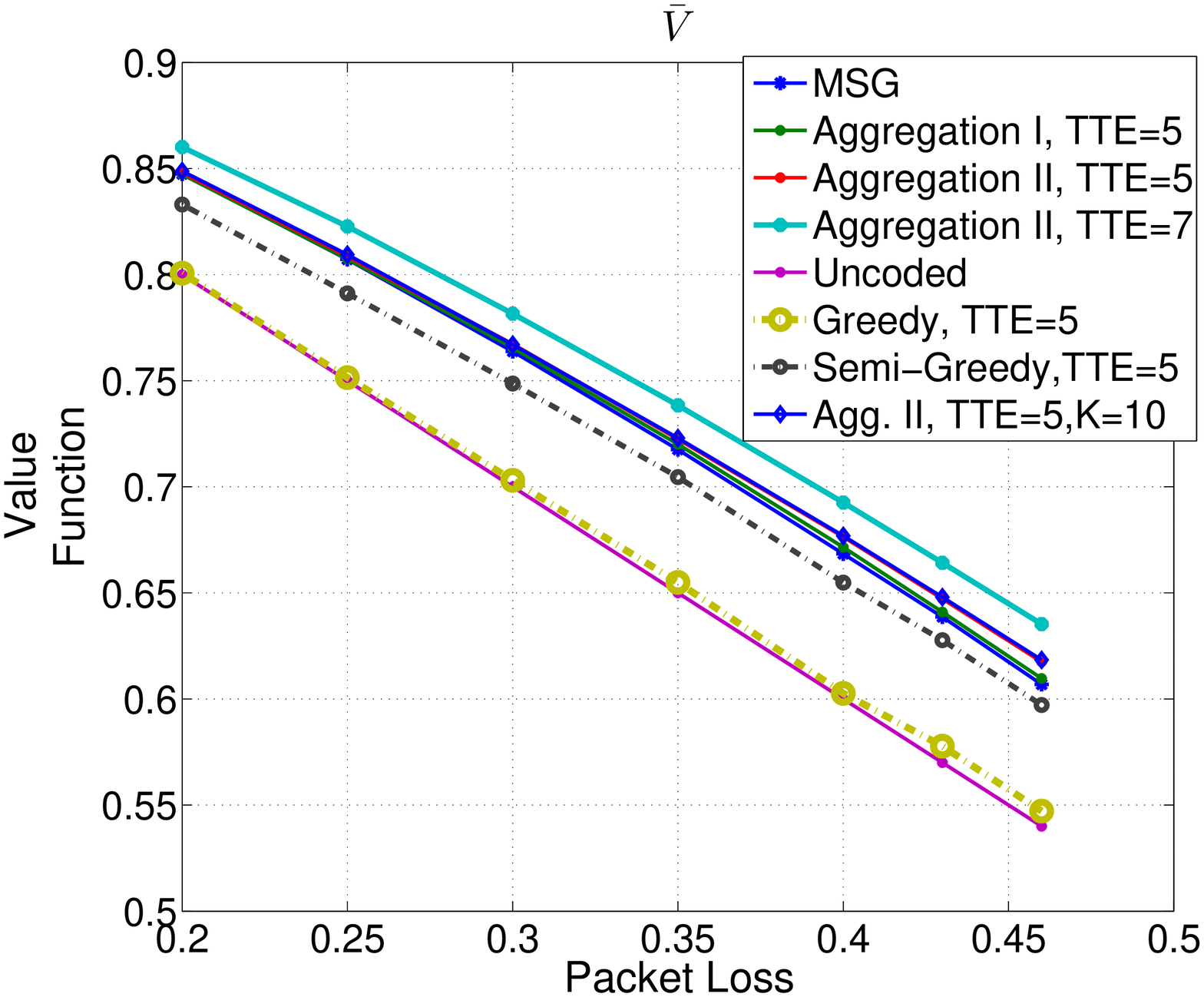}}
%\end{align*}
\caption{\sl\small
Average cost value function comparison. }\label{fig3}s
\end{center}
\end{figure}
\else
\begin{figure}[h]
	\begin{center}
		%\begin{align*}
		{\includegraphics[width=13em]{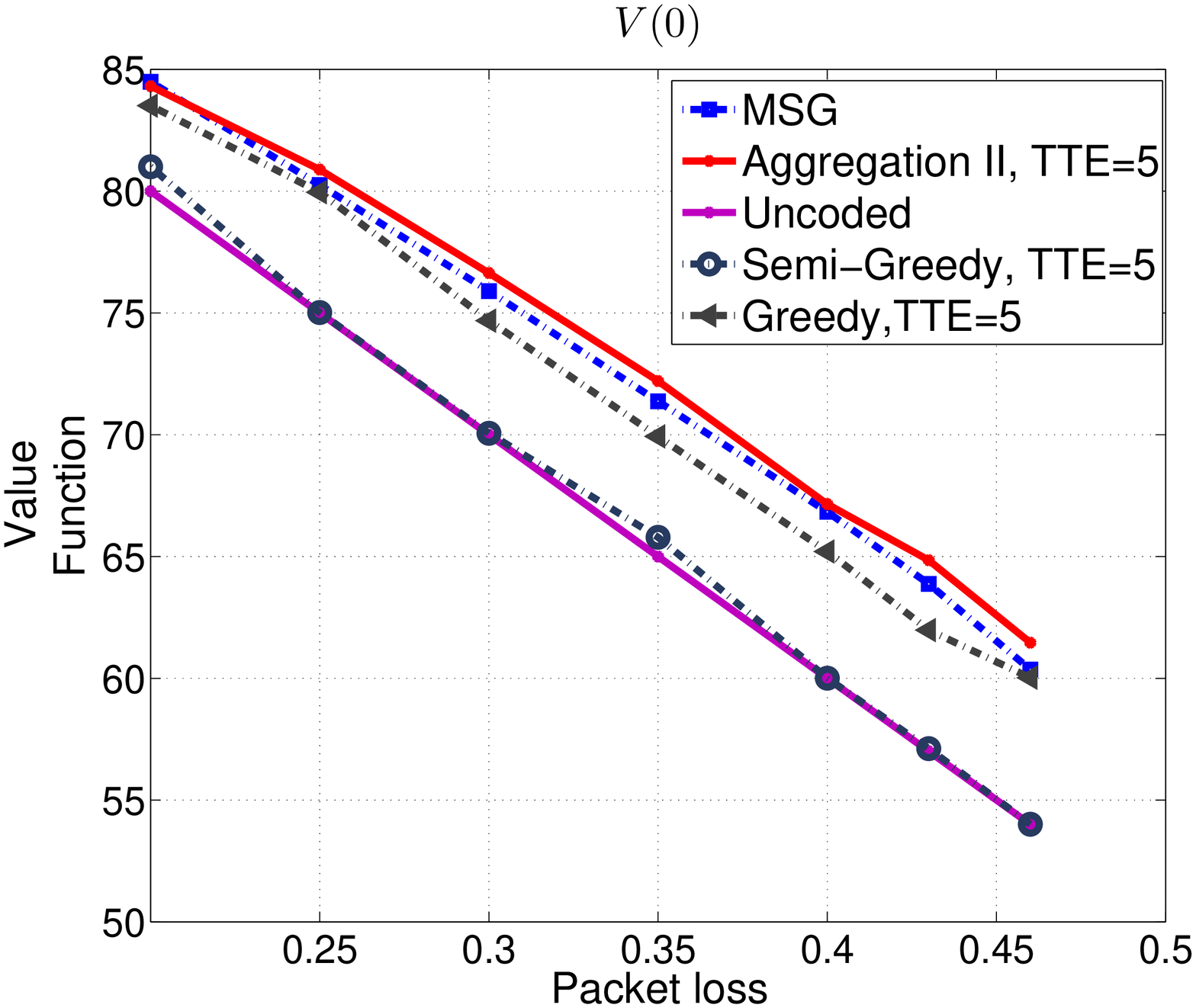}}
		%\end{align*}
	%	\caption{\sl\small
	%		Discounted value function comparison.}%
	%\end{center}
%\end{figure}
%\begin{figure}[h!]
%	\begin{center}
		%\begin{align*}
		{\includegraphics[width=13em]{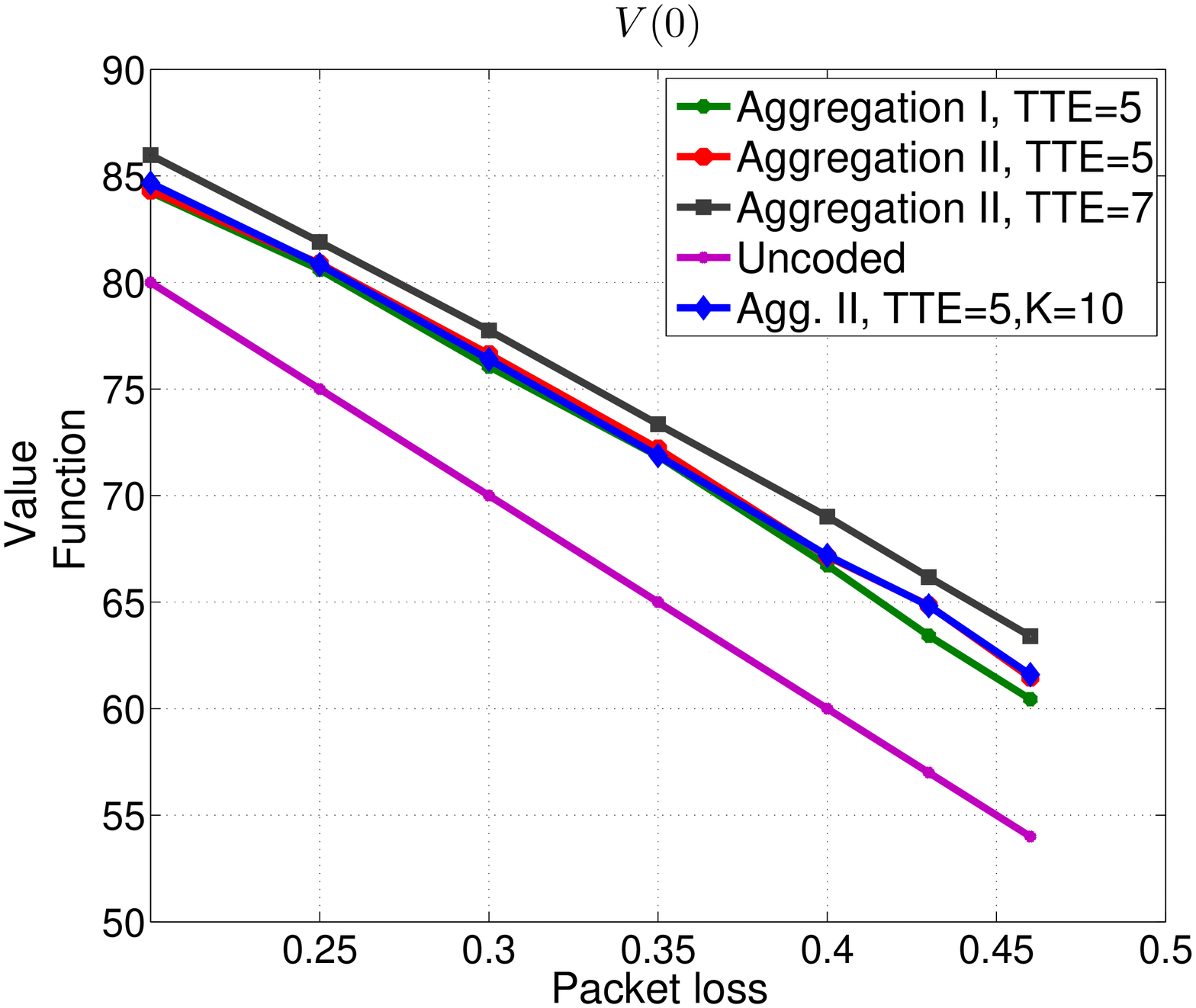}}
		{\includegraphics[width=12.5em]{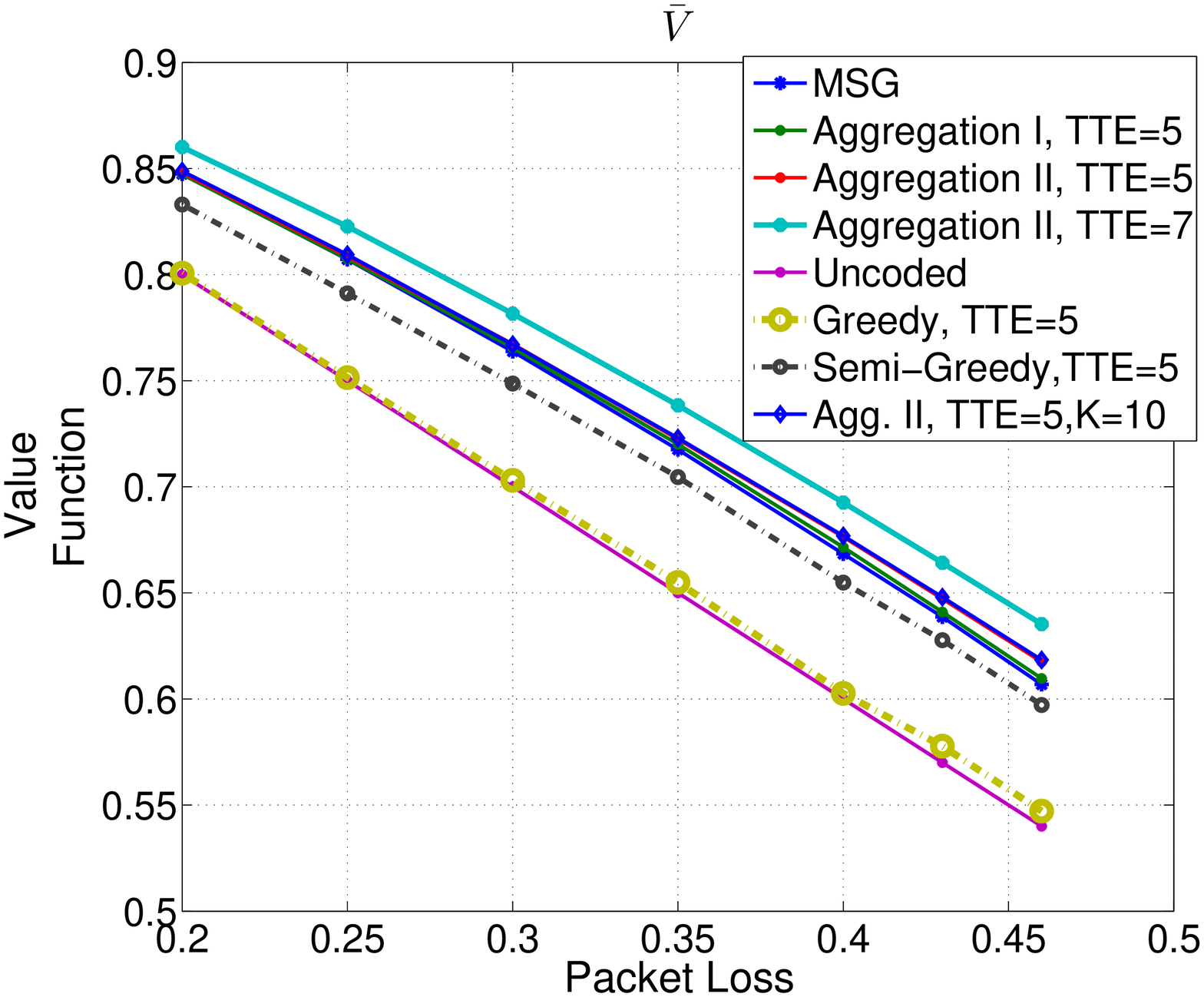}}
		%\end{align*}
		\caption{\sl\scriptsize
			Value function comparison. The left and the middle figures show the discounted case. The right figure shows the average cost long run.}\label{fig4}
	\end{center}
\end{figure}
\fi

Next, observe that when the number of users is greater than TTE, the effect of the surplus of the number of users is negligible. This stems from the fact that at most $E=TTE$ lines can have non-zero entries at all times.
Indeed, we see that $K=10$ leads to almost no improvement in performance compared to the $TTE=5$ case (the corresponding lines in the middle graph are almost coincide). %***FIXME why is there a difference?***.
Hence, we conjecture that for the case where $K>TTE$, further state-space minimization could be done.
However, once one increases the $TTE$ parameter the performance improvement is tangible. These results are seen on the middle graph as well.
Finally we compare the average cost long run simulation results (Figure~\ref{fig4}, right).
Relying on Blackwell optimality, we used the same policies we found for the discounted case. One sees the same performance gradation as for the discounted cost. %This is indeed observed on the graph.
%*** FIXME more substantial outcome is expected****

\begin{figure}[h!]
\begin{center}
%\begin{align*}
\ifArX
\includegraphics[width=\textwidth,natwidth=610,natheight=642]{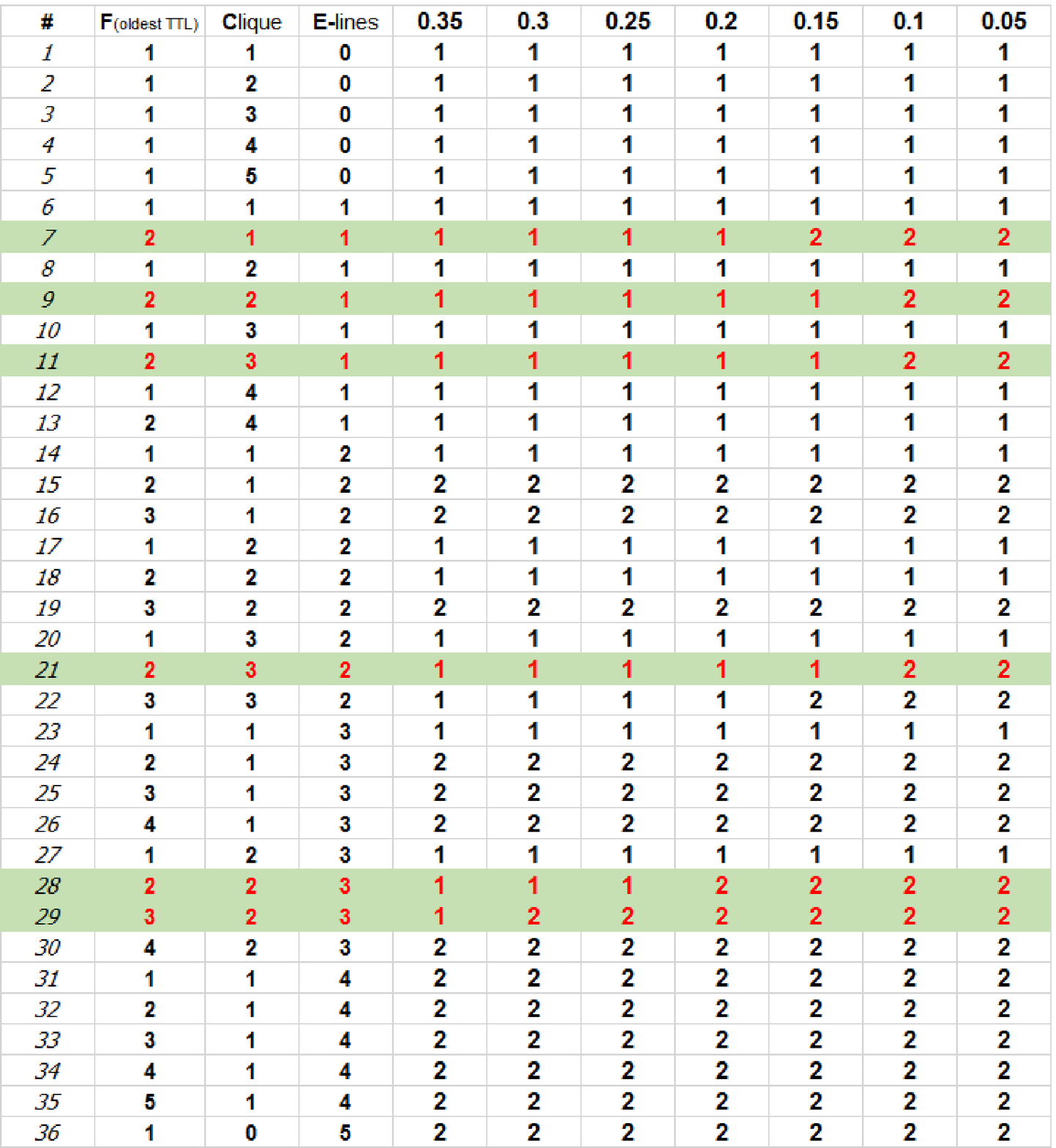}
\else
\includegraphics[width=0.28\textwidth]{T1able5TTL_A5_D.eps}
\fi
%\end{align*}
\caption{\sl\scriptsize
Approximately optimal policy, for a system with $K=5$ users and $TTE=5$.
$1$ stands for sending the clique containing the oldest line, while $2$ stands for sending a random empty line.
Observe the dependence of the policy on the packet loss, e.g. in states $7$,$9$,$11$,$21$,$28,29$ (These states are marked in red).
The impact of the parameter $F$ can be seen from states $\{F,3,2\}$,(states $20,21,22$), for example. Note that the clique is always sent in the cases where $F=1$, i.e., the oldest line in this clique is about to expire. In the cases where $F>1$, the policy depends on the packet loss, and generally tends to change to $2$ once $p$ is greater and/or $F$ is higher.\vspace*{-15pt}}\label{fig1} %\vspace{-10pts}
\end{center}
\end{figure}

%\subsubsection*{Reinforcement learning speedup}

%% file: Appndx4.tex
%\appendix
%\newpage
\appendix

\subsection{Proof of Proposition~\ref{prop1}}\label{app:proof}

\begin{proof}
We prove by constructing a reward function $\hat\scrR=\{\hat r(\hat s',\hat a,\hat s)\}$.
Let the rewards associated with \MS{policy restriction and aggregated originating state} be $\bar r(s',\bar a,\bar s)$.
%Denote by $p(s|\bar s)$ the probability that process present in the aggregated state $\bar s$ is present in the detailed state $s$, when $s\in\bar s$. Note that $\sum_{s\in\bar s}p(s|\bar s)=1$ for all $\bar s$.
%The expected value of the reward associated with the aggregated state of the origin $\bar s$ is
\MS{Observe that $\sum_{s''\in\bar s}P\Big(\bar r(s',\bar a,\bar s)=r(s',\bar a,s'')\Big)=1$. }
%By definition of $r_1$
\MS{Hence},
{\small{
\begin{align}\label{eq:Er1}
&\E \bar r(s',\bar a,\bar s)=\sum r(s',\bar a,s'')P\Big(\bar r(s',\bar a,\bar s)=r(s',\bar a,s'')\Big)=\sum_{s''\in\bar s}[r(s',\bar a,s'')]p^{\bar\pi}(s''|\bar s),
\end{align}}}
Partitioning all states in $\scrS$ to the aggregated states, we have:
{\small{
\begin{align}\label{eq:partitioning1}
& \bar r(\bar s,\bar a)=\sum_{\bar s'}\bar r(s',\bar a,\bar s)p(s'|\bar s,\bar a)=\sum_{\bar s'}\sum_{s'\in\bar s'}\bar r(s',\bar a,\bar s)p(s'|\bar s,\bar a).
\end{align}}}
{\small{
\begin{align}\label{eq:partitioning}
& \bar r(\bar s,\bar a)=\sum_{\bar s'}\bar r(s',\bar a,\bar s)p(s'|\bar s,\bar a)=\sum_{\bar s'}\Big(\sum_{s''\in\bar s}[r(s',\bar a,s'')]p_{\bar a}(s''|\bar s)\Big)p(s'|\bar s,\bar a)=\sum_{\bar s'}\sum_{s'\in\bar s'}\bar r(s',\bar a,\bar s)p(s'|\bar s,\bar a).
\end{align}}}
Similarly to $\bar r(\bar s,\bar a)$ in $\scrM_1$, define $\hat r(\hat s,\hat a)$ in $\hat\scrM$:
%Now see that the reward function of $\hat{\scrM}$ is given by
{\small{
\begin{equation}\label{eq:R}
\hat r(\hat s,\hat a)=\sum_{\hat s'}\hat r(\hat s',\hat a,\hat s)p(\hat s'|\hat s,\hat a)
\end{equation}}}
Thus, we wish to find $\hat r(\hat s',\hat a,\hat s)$ such that
{\small{
\begin{equation}\label{2}
\bar r(\bar s',\bar a,\bar s)=\hat r(\hat s',\hat a,\hat s).
\end{equation}}}
Since both the summation in\MS{~\eqref{eq:R}} and the outer summation in\MS{~\eqref{eq:partitioning}} are over all aggregated states,~\eqref{2} will be achieved by taking:
\[
\hat r(\hat s',\hat a,\hat s)p(\hat s'|\hat s,\hat a)=\sum_{s'\in\bar s'}\bar r(s',\bar a,\bar s)p(s'|\bar s,\bar a).
\]
%Substitute~\eqref{eq:Er1} in the above and take the expected value:
%Since the objective is that the value functions of $\scrM_1$ and $\hat\scrM$ will be equivalent, we use the partitioning defined in~\eqref{eq:r1} and compare it to the reward of the $\hat\scrM$ defined in~\eqref{eq:R} to derive the following reward function:
That is,
{\small{
\begin{equation}\label{eq:Rfinal}
\hat r(\hat s',\hat a,\hat s)=\frac{\sum_{s'\in\bar s'}\Big(\bar  r(s',\bar a,\bar s)\Big)p(s'|\bar s,\bar a)}{p(\bar s'|\bar s,\bar a)}
\end{equation}}}
with the mapping $\hat s\sim\bar s$ and $\hat a\sim\bar a$. \MS{Note that one should use~\eqref{eq:Er1} in~\eqref{eq:Rfinal}. }
Hence, we have the desired result:
{\small{
\[
V_{\hat\scrU}(\hat s_0)=\sum_{n=0}^\infty\gamma^n\hat r_n(\hat s_{n+1},\hat a,\hat s_n)=\sum_{n=0}^\infty\gamma^n\bar r_{n}(\bar s_{n+1},\bar a,\bar s_n)=V_{\bar \scrU}(\bar s_0)
\]}}
%Now, since each summand in~\eqref{eq:R} is equal to each summand in~\eqref{eq:partitioning}, i.e. each inner sum in the partitioning,~\eqref{2} holds.
%where the values in numerator are substituted by~\eqref{eq:Er1} and~\eqref{eq:p9}, while the value in denominator is substituted by~\eqref{eq:p10}.
%Using this reward~\eqref{eq:Rfinal} for $J^{\hat\pi}(\hat s_0)$ satisfies the statement of the theorem.
\MS{\qed}
\end{proof}
%\MS{
%\begin{remark}
%Note that the reward in initial state at step $n=0$ is left unspecified. We assume that the reward in the initial state $s_0\in\bar s_0$ is calculated according to the distribution of $p_{\bar a}(s_0|\bar s)$. We omit the calculation details.
%\end{remark}}
%\subsection{State aggregation example}\label{app:exmpl2}
\input{example2}

\ifArX
\subsection{Proof of Bounds}\label{app:bounds}

\input{bounds}

%\textbf{Proof of Proposition~\ref{prop2}.}

\subsubsection{Proof of Lemma~\ref{l2}}\label{secA3}

\input{lemmaTh2}

\fi
%\end{appendices} 

%% file: example2.tex
\begin{example}\label{xmpl:1}
The following demonstrates state aggregation (as it was defined by Aggregation I in Section~\ref{sec:st-agg}) and results of Proposition~\ref{prop1}.
Consider the case of $4$ users. Each line holds the packets of user $i$. We exemplify the detailed states where $L=3$, $E=1$. These states are aggregated into the state denoted by $\bar s_{3,1}$. %There are $4$ possible cliques.
Possible cliques are demonstrated in the detailed states denoted $s_1,s_2,s_3,s_4$ below. Observe that these states contain only {\it minimal} number of $1$-s.
\ifdouble
{\small
\[
s_1 =
\begin{pmatrix}
0 & 1 & 1 & 0  \\
1 & 0 & 1 & 0  \\
1 & 1 & 0 & 0   \\
0 & 0 & 0 & 0
\end{pmatrix}
\quad
s_2 =
\begin{pmatrix}
0 & 0 & 0 & 0  \\
0 & 0 & 1 & 1  \\
0 & 1 & 0 & 1   \\
0 & 1 & 1 & 0
\end{pmatrix}
\]
}
{\small
\[
s_3 =
\begin{pmatrix}
0 & 0 & 1 & 1  \\
0 & 0 & 0 & 0  \\
1 & 0 & 0 & 1   \\
1 & 0 & 1 & 0
\end{pmatrix}
\quad
s_4 =
\begin{pmatrix}
0 & 1 & 0 & 1  \\
1 & 0 & 0 & 1  \\
0 & 0 & 0 & 0   \\
1 & 1 & 0 & 0
\end{pmatrix}
\]
}
\else
\par \vspace*{-1.5em}
{\tiny\[
\text{    }
\]}
{{\tiny
	\[
	s_1 =
	\begin{pmatrix}
	0 & 1 & 1 & 0  \\
	1 & 0 & 1 & 0  \\
	1 & 1 & 0 & 0   \\
	0 & 0 & 0 & 0
	\end{pmatrix}
	\quad
	s_2 =
	\begin{pmatrix}
	0 & 0 & 0 & 0  \\
	0 & 0 & 1 & 1  \\
	0 & 1 & 0 & 1   \\
	0 & 1 & 1 & 0
	\end{pmatrix}
	\quad
	s_3 =
	\begin{pmatrix}
	0 & 0 & 1 & 1  \\
	0 & 0 & 0 & 0  \\
	1 & 0 & 0 & 1   \\
	1 & 0 & 1 & 0
	\end{pmatrix}
	\quad
	s_4 =
	\begin{pmatrix}
	0 & 1 & 0 & 1  \\
	1 & 0 & 0 & 1  \\
	0 & 0 & 0 & 0   \\
	1 & 1 & 0 & 0
	\end{pmatrix}
	\]}}
\fi
See that in $s_1$, there are $8$ additional options for the last column. In particular, observe the following four states with the same empty line and the same clique as in $s_1$.
\ifdouble
{\tiny\[s_5 =
\begin{pmatrix}
0 & 1 & 1 & 0  \\
1 & 0 & 1 & 1  \\
1 & 1 & 0 & 0   \\
0 & 0 & 0 & 0
\end{pmatrix}
\quad
s_6=
\begin{pmatrix}
0 & 1 & 1 & 1  \\
1 & 0 & 1 & 0  \\
1 & 1 & 0 & 0  \\
0 & 0 & 0 & 0
\end{pmatrix}
\]}
{\tiny
\[
s_7=
\begin{pmatrix}
0 & 1 & 1 & 0  \\
1 & 0 & 1 & 0  \\
1 & 1 & 0 & 1   \\
0 & 0 & 0 & 0
\end{pmatrix}
\quad
s_8=
\begin{pmatrix}
0 & 1 & 1 & 1  \\
1 & 0 & 1 & 0  \\
1 & 1 & 0 & 1   \\
0 & 0 & 0 & 0
\end{pmatrix}
\]}
\else
\par \vspace*{-1.5em}
{\tiny\[
\text{    }
\]}
{{\tiny\[s_5 =
	\begin{pmatrix}
	0 & 1 & 1 & 0  \\
	1 & 0 & 1 & 1  \\
	1 & 1 & 0 & 0   \\
	0 & 0 & 0 & 0
	\end{pmatrix}
	\quad
	s_6=
	\begin{pmatrix}
	0 & 1 & 1 & 1  \\
	1 & 0 & 1 & 0  \\
	1 & 1 & 0 & 0  \\
	0 & 0 & 0 & 0
	\end{pmatrix}
	\quad
	s_7=
	\begin{pmatrix}
	0 & 1 & 1 & 0  \\
	1 & 0 & 1 & 0  \\
	1 & 1 & 0 & 1   \\
	0 & 0 & 0 & 0
	\end{pmatrix}
	\quad
	s_8=
	\begin{pmatrix}
	0 & 1 & 1 & 1  \\
	1 & 0 & 1 & 0  \\
	1 & 1 & 0 & 1   \\
	0 & 0 & 0 & 0
	\end{pmatrix}
	\]}}
\fi
The same holds for $s_2,s_3$ and $s_4$. Concluding, the state $\bar s_{3,1}$ aggregates $32$ detailed states.

There are two possible actions, denote them $\bar a=1$ and $\bar a=2$, which stand respectively for transmitting the clique and transmitting (the only) empty line. %Note that transmission of a clique stands for transmission of encoded linear combination containing all vertexes of the clique. That is, a receiver which is a member of a clique will be instantly able to decode the packet designated for itself, provided the packet was successfully accepted.
Note that the encoded message for $s_1$ contains the bits $1,2,3$, for $s_2$ it contains packets $2,3,4$, for $s_3$ it contains packets $1,3,4$ and for $s_4$ it contains packets $1,2,4$.
%Assume that we are given a stationary policy $\pi_1$. Consequently,
The probability $p(s_i|\bar s_{3,1})$ stand for the probability to be in a specific detailed state which belongs to the aggregated state  $\bar s_{3,1}$, (we omit the superscript of the policy in this example).  %, being in $\bar s_{3,1}$. Note that $\sum_{s_i\in\bar s_{3,1}}p(s=s_i|\bar s=\bar s_{3,1})=1$.
The rest of the example concentrates on the state $s_5\in\bar s_{3,1}$ and action $\bar a=1$, i.e., transmission of the clique. Assume the action results in the detailed state $s_a$.
\par \vspace*{-1.5em}
{\tiny\[
\text{    }
\]}
{\tiny
\[s_a =
\begin{pmatrix}
0 & 0 & 0 & 0  \\
1 & 0 & 1 & 1  \\
0 & 0 & 0 & 0   \\
0 & 0 & 0 & 0
\end{pmatrix}
\quad
s_9 =
\begin{pmatrix}
0 & 0 & 0 & 0  \\
1 & 0 & 1 & 1  \\
0 & 1 & 0 & 1   \\
0 & 1 & 1 & 0
\end{pmatrix}
\]}
Clearly, $s_a\in\bar s_{1,3}$. Further, assume equal packet loss probability denoted by $q$. %For simplicity, we assume the same $p$ for all users.
The aforementioned transition occurs with probability $p(s_a|\bar a=1,s_5)=q(1-q)^2$. That is, two of the users in the clique ($1$ and $3$) successfully decoded the encoded bit, while user $2$ failed to do so. See that the same transition can happen from state $s_9$. %(compare to $s_5$)
%{\small
%\[s_9 =
%\begin{pmatrix}
%0 & 0 & 0 & 0  \\
%1 & 0 & 1 & 1  \\
%0 & 1 & 0 & 1   \\
%0 & 1 & 1 & 0
%\end{pmatrix}
%\]}
That is, the clique containing encoding of $2,3,4$ was transmitted, and user $2$ failed to decode. This transition occurs with probability $p(s_a|\bar a=1,s_9)=q(1-q)^2$ as well.
We sum up over all such detailed states (according to Appendix~\ref{app:proof}):
{\small
\[
p(s_a|\bar a=1,\bar s=\bar s_{3,1})=\sum_{s_i\in\bar s_{3,1}}p(s_a|\bar a=1,s_i)p(s_i|\bar s_{3,1}),
\]}
This summation counts over all $32$ detailed states in $\bar s_{3,1}$.
Clearly, some of the probabilities, e.g., $p(s_a|\bar a=1,s_2)$ are zero, hence do not contribute to the summation. % since the corresponding transitions do not exist.
For calculation convenience, we assume convention that in these cases $r(s_a,\bar a=1,s_i)=0$.
We calculate the average reward associated with the transition from $s_{3,1}$ to $s_a$, according to~\eqref{eq:Er1}:
{\small \[\E\bar r (s_a,\bar a=1,\bar s_{3,1})=\sum_{s_i\in\bar s_{3,1}}r(s_a,\bar a=1,s_i)p(s_i|\bar s_{3,1})
\]}
Note that transition to state $s_a$, acting $\bar a=1$ from $\bar s_{3,1}$, is only possible when $2$ of $3$ encoded packets were successfully decoded. Thus, the reward for these cases is equal to $2$, while for the other cases it is zero.
%We next observe additional states, other than $s_a$ to which the transitions from $\bar s_{3,1}$ can happen.
%Denote $\bar\scrS'\in\bar\scrS$ as follows.
Let  the subset $\bar\scrS'\in\bar\scrS$ to contain the possible next (aggregated) states, assuming the clique size in the previous state was $2$. Namely, $\bar\scrS'=\{\bar s_{3,1},\bar s_{2,2},\bar s_{1,3},\bar s_{0,4}\}$, where the components  refer to the events of successfully decoding of $0,1,2$ and $3$ packets correspondingly.
In order to calculate $\bar r (\bar s_{3,1},\bar a)$, we %use~\eqref{eq:r1} to
first summarize over all possible outcomes
{\small{$\bar r (\bar s_{3,1},\bar a=1)=\sum_{s_i}\bar r (s_i,\bar a=1,\bar s_{3,1})p(s_i|\bar a=1,\bar s_{3,1})
$.}} Substituting the expected values and the probabilities we found above, and arranging according to the aggregated states, we have:
\ifdouble
{\small
\begin{align*}
& \bar r (\bar s_{3,1},\bar a=1)= \\
& \sum_{s_i\in\bar s_{3,1}}\E\bar r (s_i,\bar a=1,\bar s_{3,1})p(s_i|\bar a=1,\bar s_{3,1})\\
&+\sum_{s_i\in\bar s_{2,2}}\E\bar r (s_i,\bar a=1,\bar s_{3,1})p(s_i|\bar a=1,\bar s_{3,1}) +\\
& \sum_{s_i\in\bar s_{0,3}}\E\bar r (s_i,\bar a=1,\bar s_{3,1})p(s_i|\bar a=1,\bar s_{3,1})+\\
&\sum_{s_i\in\bar s_{0,4}}\E\bar r (s_i,\bar a=1,\bar s_{3,1})p(s_i|\bar a=1,\bar s_{3,1})=\\
&\sum_{\bar s\in\bar S'}\sum_{s_i\in\bar s}\E\bar r (s_i,\bar a=1,\bar s_{3,1})p(s_i|\bar a=1,\bar s_{3,1})
\end{align*}}
\else
{\small
	\begin{align*}
	& \bar r (\bar s_{3,1},\bar a=1)=\\
	& \sum_{s_i\in\bar s_{3,1}}\E\bar r (s_i,1,\bar s_{3,1})p(s_i|1,\bar s_{3,1})
	+\sum_{s_i\in\bar s_{2,2}}\E\bar r (s_i,1,\bar s_{3,1})p(s_i|1,\bar s_{3,1}) +\\
	 &\sum_{s_i\in\bar s_{1,3}}\E\bar r (s_i,1,\bar s_{3,1})p(s_i|1,\bar s_{3,1})+
	\sum_{s_i\in\bar s_{0,4}}\E\bar r (s_i,1,\bar s_{3,1})p(s_i|1,\bar s_{3,1})= \\
  &\sum_{\bar s\in\bar\scrS'}\sum_{s_i\in\bar s}\E\bar r (s_i,1,\bar s_{3,1})p(s_i|1,\bar s_{3,1})=\sum_{\bar s\in\bar\scrS'}\sum_{s_i\in\bar s}\Big(\sum_{s_j\in\bar s_{3,1}}\bar r (s_i,1,s_{j})p(s_i|\bar s_{3,1})\Big)p(s_i|1,\bar s_{3,1})
	\end{align*}}
\fi
We now turn to the induced MDP $\bar\scrM$. %Recall that the reward function is defined by~\eqref{eq:R}.
%{\small\[
%R(\hat s,\hat a)=\sum_{\hat s'}R(\hat s',\hat a,\hat s)\hat p(\hat s'|\hat s,\hat a)
%\]}
%We will assume that the states and the actions of $\bar\scrM$ are equivalent to those of $\scrM_1$.
%Continuing with the example,
Denote $\hat s=\hat s_{3,1}$ and $\hat a=1$. We find the reward associated with transition to $\hat s_{1,3}$, $\hat r(\hat s_{0,3},\hat a=1,\hat s_{3,1})$.
Equate component-wise $\hat r(\hat s,\hat a)$ and $\bar r (\bar s_{3,1},\bar a=1)$ as follows:
\ifdouble
{\small\begin{align*}
& r(\hat s_{0,3},\hat a=1,\hat s_{3,1})\hat p(\hat s_{0,3}|\hat s_{3,1},\hat a=1)=\\
&\sum_{s_i\in\bar s_{0,3}}\E\bar r (s_i,\bar a=1,\bar s_{3,1})p(s_i|\bar a=1,\bar s_{3,1})
\end{align*}}
\else
{\small\begin{align*}
	& r(\hat s_{1,3},\hat a=1,\hat s_{3,1})p(\hat s_{1,3}|\hat s_{3,1},\hat a=1)=
	\sum_{s_i\in\bar s_{1,3}}\E\bar r (s_i,\bar a=1,\bar s_{3,1})p(s_i|\bar a=1,\bar s_{3,1})
	\end{align*}}
\fi
It is left to calculate the probability $p(\hat s_{1,3}|\hat s_{3,1},\hat a=1)$. %We equate this transition probability to $p$ and use~\eqref{eq:p10}:
\ifdouble
{\small\begin{align*}
& p(\hat s_{1,3}|\hat s_{3,1},\hat a=1)=p(\bar s_{1,3}|\bar s_{3,1},\bar a=1)=\\
&\sum_{s'\in\bar s_{1,3}}\sum_{s\in\bar s_{3,1}}p(s'|\bar a=1,s)p(s|\bar s_{3,1})
\end{align*}}
\else
{\small\begin{align*}
	& p(\hat s_{1,3}|\hat s_{3,1},\hat a=1)=p(\bar s_{1,3}|\bar s_{3,1},\bar a=1)=
	\sum_{s'\in\bar s_{1,3}}\sum_{s\in\bar s_{3,1}}p(s'|\bar a=1,s))p(s|\bar s_{3,1})
	\end{align*}}
\fi
Finally, the solutions for all possible $\hat r(\hat s',\hat a=1,\hat s_{3,1})$ are found from
\ifdouble
{\small\[
r_{\hat s_{0,3},\hat a=1,\hat s_{3,1}}=\frac{\sum_{s_i\in\bar s_{0,3}}\E\bar r (s_i,\bar a=1,\bar s_{3,1})p(s_i|1,\bar s_{3,1})}{\sum_{s'\in\bar s_{0,3}}\sum_{s\in\bar s_{3,1}}p(s'|\bar a=1,s)p(s|\bar s_{3,1})}
\]}
{\small\[
r_{\hat s_{0,4},\hat a=1,\hat s_{3,1}}=\frac{\sum_{s_i\in\bar s_{0,4}}\E\bar r (s_i,\bar a=1,\bar s_{3,1})p(s_i|1,\bar s_{3,1})}{\sum_{s'\in\bar s_{0,4}}\sum_{s\in\bar s_{3,1}}p(s'|\bar a=1,s)p(s|\bar s_{3,1})}
\]}
{\small\[
r_{\hat s_{2,2},\hat a=1,\hat s_{3,1}}=\frac{\sum_{s_i\in\bar s_{2,2}}\E\bar r (s_i,\bar a=1,\bar s_{3,1})p(s_i|1,\bar s_{3,1})}{\sum_{s'\in\bar s_{2,2}}\sum_{s\in\bar s_{3,1}}p(s'|\bar a=1,s)p(s|\bar s_{3,1})}
\]}
{\small\[
r_{\hat s_{3,1},\hat a=1,\hat s_{3,1}}=\frac{\sum_{s_i\in\bar s_{3,1}}\E\bar r (s_i,\bar a=1,\bar s_{3,1})p(s_i|1,\bar s_{3,1})}{\sum_{s'\in\bar s_{3,1}}\sum_{s\in\bar s_{3,1}}p(s'|\bar a=1,s)p(s|\bar s_{3,1})}
\]}
\else
{\small\[
	r_{\hat s_{1,3},\hat a=1,\hat s_{3,1}}=\frac{\sum_{s_i\in\bar s_{1,3}}\E\bar r (s_i,\bar a=1,\bar s_{3,1})p(s_i|1,\bar s_{3,1})}{\sum_{s'\in\bar s_{1,3}}\sum_{s\in\bar s_{3,1}}p(s'|\bar a=1,s)p(s|\bar s_{3,1})}
\quad
	r_{\hat s_{0,4},\hat a=1,\hat s_{3,1}}=\frac{\sum_{s_i\in\bar s_{0,4}}\E\bar r (s_i,\bar a=1,\bar s_{3,1})p(s_i|1,\bar s_{3,1})}{\sum_{s'\in\bar s_{0,4}}\sum_{s\in\bar s_{3,1}}p(s'|\bar a=1,s)p(s|\bar s_{3,1})}
	\]}
{\small\[
	r_{\hat s_{2,2},\hat a=1,\hat s_{3,1}}=\frac{\sum_{s_i\in\bar s_{2,2}}\E\bar r (s_i,\bar a=1,\bar s_{3,1})p(s_i|1,\bar s_{3,1})}{\sum_{s'\in\bar s_{2,2}}\sum_{s\in\bar s_{3,1}}p(s'|\bar a=1,s)p(s|\bar s_{3,1})}
	\quad
	r_{\hat s_{3,1},\hat a=1,\hat s_{3,1}}=\frac{\sum_{s_i\in\bar s_{3,1}}\E\bar r (s_i,\bar a=1,\bar s_{3,1})p(s_i|1,\bar s_{3,1})}{\sum_{s'\in\bar s_{3,1}}\sum_{s\in\bar s_{3,1}}p(s'|\bar a=1,s)p(s|\bar s_{3,1})}
	\]}
\fi
Note that $p(s|\bar s_{3,1})$ are policy dependent and in order to be found, the Markov chain associated with the MDP should be entirely solved. As it is explained throughout the paper, we circumvent this difficulty by reinforcement learning.
This finishes the example.
\end{example}

%% file: bounds.tex
We prove low and upper bounds on the slope of $V(s)$, discounted infinite horizon cost.
Denoting $p^e_{k}$, the probability to increase $L(s)$ from $k$ to $k+1$ when transmitting an empty line,
see that $p^e_{k}<p$, that is, incrementing the clique is conditioned on the transmission being unsuccessful.
Denote by $p^c_{k,i}$, $0\leq i\leq k$, the transition probability from state $k$ from to state $i$, when acting by the transmission of the clique (i.e. $a=1$).
Note that $p^k_{i}$ is formally given by $p^c_{k,i}=p(\bar s'=i|\bar s=k,a=1)$
Define operator $T$, corresponding to the Bellman equation, acting on $V$ 
\ifdouble 
\par \vspace*{-1.5em} {\small
\begin{align}
& TV(k)=\max\{[p^e_{k}\gamma V(k+1)+(1-p^e_k)\gamma V(k)+(1-p)],\nonumber\\
&[\sum_{i=0}^kp^c_{k,i}\gamma V(i)+(1-p)k]\},\label{4}
\end{align}
}
\else
\par \vspace*{-1.5em} {\small
	\begin{align}
	& TV(k)=\max\{[p^e_{k}\gamma V(k+1)+(1-p^e_k)\gamma V(k)+(1-p)],
	[\sum_{i=0}^kp^c_{k,i}\gamma V(i)+(1-p)k]\},\label{4}
	\end{align}
}
\fi
with boundary conditions
\ifdouble
\par \vspace*{-1.5em}{\small{
\begin{align*}
 TV(0) &=\{[p^e_{0}\gamma V(1)+(1-p^e_0)\gamma V(0)+(1-p)]\}, \\
 TV(K) &=\sum_{i=0}^Kp^c_{K,i}\gamma V(i)+(1-p)K.
\end{align*} }}
\else
\par \vspace*{-1.5em}{\small{
		\begin{align*}
		TV(0) =\{[p^e_{0}\gamma V(1)+(1-p^e_0)\gamma V(0)+(1-p)]\}, \quad
		TV(K) =\sum_{i=0}^Kp^c_{K,i}\gamma V(i)+(1-p)K.
		\end{align*} }}
\fi
The immediate rewards are explained as follows. The reward for transmission of an empty line is given by the probability of a successful transmission, that is $1-p$. In the case a clique of size $k$ is transmitted, we have $k$ potential i.i.d rewards, which gives $(1-p)k$. 
To simplify the notation, denote %$S(k)=\gamma\sum_0^kp_kV(k)$,
$\tilde S(k)=\gamma\sum_{i=0}^kp^c_{k,i}V(k-i)+(1-p)k$ and  $\tilde E(k)=p^e_{k}\gamma V(k+1)+(1-p^e_k)\gamma V(k)+(1-p)$.

Let $\calS$ be the set of functions from $\{0,1,\ldots,K\}$ to $\R$ that are nondecreasing, and have slope bounded from above by $d_k$, that is\par \vspace*{-1.5em}
{\small
\begin{equation}\label{21}
V(k+1)-V(k)\le d,\qquad k\in\{0,1,\ldots,K-1\},
\end{equation} }
and bounded from below as follows:\par \vspace*{-1.5em}
\ifdouble
{\small
\begin{align}\label{19}
& V(k)-V(k-i)\geq i-c, \text{  where  } \\
& i\in\{1,\ldots,K-1\},\;k\in\{i,i+1,\ldots,K\}.\no
\end{align}}
\else
{\small
	\begin{align}\label{19}
	& V(k)-V(k-i)\geq i-c, \text{   where  } 
	 i\in\{1,\ldots,K-1\},\;k\in\{i,i+1,\ldots,K\}.
	\end{align}}
\fi
%We show the existence of properties~\eqref{21} and~\eqref{19} and employ them throughout the proof of Proposition~\ref{prop2}.
Lemma~\ref{l2} below asserts that $T$ preserves $\calS$, and acts on it as a strict contraction. The combination of these two assertion implies that $V(s)$ is in $\calS$ (see the discussion below), that is, it possesses the corresponding properties~\eqref{19} and~\eqref{21}. %This, in its turn, will be exploited in Proposition~\ref{prop2} to prove the threshold policy.
\begin{lemma}\label{l2}
	There exist constants $c$ and $d$, such that one has $T\calS\subset\calS$. Moreover, there exists a constant $\alpha\in(0,1)$ such that
	\[
	\|TU-TW\|\le\alpha\|U-W\| \,\text{ for every }\, U,W\in\calS.
	\]
\end{lemma}
{\it Discussion.}
The main difficulty of the proof below stems from the ambiguity regarding the transition probabilities. That is, the precise calculation of these probabilities is computationally infeasible, especially for large number of users, $K$. We solved this by reinforcement learning on the practical side. On the analytical side, we make several assumptions and estimations, which we justify throughout the proof.
To this end, the proof is primarily built on the assumption that $V\in\calS$ and possesses all the corresponding properties. We exploit this assumption in order to prove that operator $T$, acting on $\calS$, {\it preserves} these properties, that is $TV\in\calS$.
Now note that the map defined by operator $T$ in~\eqref{4}, acting on a complete metric space $S$, with $T : \R^{|S|}\to\R^{|S|}$ of value functions, is a strict contraction, %This result on MDP with discounted factor on infinite horizon is well-known in literature, for example see~\cite{Puterman}, and we skip its technical demonstration in our specific case. For the contraction mapping principle see for example
~\cite[Theorem V.18]{RS}. Therefore, $T$ has a unique fixed
point which solves $TU = U$. On the other hand, $V$ is the unique solution
to the same (Bellman) equation in the space of {\it all} functions. %from $\calM \to\R$.
As a result, $V = U$. Whence, in case we start the converging procedure with initial function which preserves~\eqref{21} and~\eqref{19} , by iteratively activating the operator $T$, we end up with solution which preserves the aforementioned property.%\par\nolinebreak

%\begin{figure}
%	%\begin{center}
%	%\begin{align*}
%	\centerline{\includegraphics[angle=0, width=0.48 \textwidth]{zz1.eps}} %
%	%\end{align*}
%	
%	%\hfill
%	\caption{\sl\small
%		Discounted accumulative reward as a function of the maximal clique size. Comparison of the two policies. Observe a threshold at $3$.
%		\label{fig5}}
%	%\end{center}
%\end{figure}

%% file: lemmaTh2.tex
%\begin{proof}

%Hence, to show the threshold structure, one has to prove that operator $T$ preserves the property~\eqref{6}. This will show that $V\in\R^{S}$ and the threshold policy in $E$ will follow.
Denote by $p^c_{k,i,j}$ the probability $p^c_{k,i}$, \textit{conditioned} that the largest fully disjoint clique with the clique of size $k$, prior the transmission, was of size $j$. Note that $j\leq k$. Denote the probability of having such a disjoint clique as $p_{k,j}$
(by total probability $p^k_{i}=\sum_jp^k_{i,j}p_{k,j}$. )

By Equation~\eqref{A1} and Lemma~\ref{lem5} (see the end of this section) it holds either $p_{k,i}^c=p^c_{k,i,0}+a_1= p^{i}(1-p)^{k-i}\binom{k}{i}+a_1$, for some nonnegative $a_1$, or $p_{k,i}^c=0$.
%Since $\sum_{i=0}^kp^k_{i}=1$, the actual probabilities to pass to the states higher than $j-1$ are increased. Hence, $a_1$ compensates for this increase.
(Note, that $a_1=0$ in the case there were no other cliques of size $k-i$ prior to the encoded transmission.)

%See that $p_e<p$.

See that by multiple application of~\eqref{21} and~\eqref{19}
\ifdouble
{\small{
\begin{align}
& \tilde S_k\leq\gamma V_0+\gamma\sum_{i=0}^ki*p_{k,i}^cd+(1-p)k\no\\
&\;\;\;=\gamma V_0+\gamma\sum_{i=0}^ki*p_{k,i,0}^{c}d+(1-p)k+a_2(k) \no\\
& \;\;\;=\gamma V_0+\gamma pkd+(1-p)k+a_2(k)\label{23}
\end{align}
}}
\else
{\small{
		\begin{align}
		& \tilde S_k\leq\gamma V_0+\gamma\sum_{i=0}^ki*p_{k,i}^cd+(1-p)k \no\\
		&=\gamma V_0+\gamma\sum_{i=0}^ki*p_{k,i,0}^{c}d+(1-p)k+a_2(k) 
	=\gamma V_0+\gamma pkd+(1-p)k+a_2(k)\label{23}
		\end{align}
	}}
\fi
and
\ifdouble
{\small{
\begin{align}
& \tilde S_k\geq\gamma V_k-\gamma\sum_{i=0}^k(k-i)*p_{k,i}^cd+(1-p)k \no\\
&\;\;\;=\gamma V_k-\gamma\sum_{i=0}^k(k-i)*p_{k,i,0}^{c}d+(1-p)k-b_2(k) \no\\
& \;\;\;=\gamma V_k-\gamma(1-p)kd+(1-p)k-b_2(k)\label{24}
\end{align}
}}
\else
{\small{
		\begin{align}
		&\tilde S_k\geq\gamma V_k-\gamma\sum_{i=0}^k(k-i)*p_{k,i}^cd+(1-p)k \no\\
		&=\gamma V_k-\gamma\sum_{i=0}^k(k-i)*p_{k,i,0}^{c}d+(1-p)k-b_2(k)
	=\gamma V_k-\gamma(1-p)kd+(1-p)k-b_2(k)\label{24}
		\end{align}
\fi
where $a_2(k)$ and $b_2(k)$ stand for summations of all compensation constants $a_1(k,i)$, in both cases above.

We use the contraction property in the remaining part of the proof. Since, by assumption, $V$ satisfies~\eqref{21} and~\eqref{19}, we only have to show that
{\small{
\begin{align}
\label{18}
&\max\{\tilde S(k+1),\tilde E(k+1)\}-\max\{\tilde S(k),\tilde E(k)\}\leq d \\
&\max\{\tilde S(k-i),\tilde E(k-i)\}-\max\{\tilde S(k),\tilde E(k)\}\leq-i+c
\end{align}}}

We analyze all the possible options within the curly brackets, as follows.
%Starting with the first bound, substitute~\eqref{23}.
%\[
%\Delta_i^{k}=p^{k+1}_{i}-p^{k}_{i}
%\]
%Note that $\sum_{i=0}^k\Delta_i^{k}=p^{k+1}$ and observe that

%It suffices to prove
{\it 1. }
{\small{
\begin{align*}
& TV(k+1)-TV(k)=\tilde S(k+1)-\tilde S(k)\\
& TV(k-i)-TV(k)=\tilde S(k-i)-\tilde S(k)
\end{align*}}}

Applying Lemma~\ref{lem4} it immediately follows that $TV(k+1)-TV(k)<d$ and $TV(k-i)-TV(k)\geq-i+c$ in this case.

{\it 2. }
{\small{
\begin{align*}
& TV(k+1)-TV(k)=\tilde E(k+1)-\tilde E(k)\\
& TV(k-i)-TV(k)=\tilde E(k-i)-\tilde E(k)
\end{align*}}}
In order to prove the second case we should comply with the expressions for $d$ and $c$ found in the first case.
Note that $ p^e_{k+1}< p^e_{k}$. That is, the probability to increase the size of the maximal clique then acting by sending an empty line decreases with the state size. Hence,
{\small{
\begin{align*}
&\tilde E(k+1)-\tilde E(k)=
 p^e_{k+1}\gamma V(k+2)+(1-p^e_{k+1})\gamma V(k+1)-
 p^e_{k}\gamma V(k+1)-(1-p^e_{k})\gamma V(k)\\&=
 p^e_{k+1}\gamma V(k+2)+(1-p^e_{k+1}-p^e_{k})\gamma V(k+1)
-(1-p^e_{k})\gamma V(k)\\&\leq
 \gamma dp^e_{k+1}+(1-p^e_{k})\gamma V(k+1)-(1-p^e_{k})\gamma V(k)\leq
 \gamma dp^e_{k+1}+d(1-p^e_{k})\gamma<d\gamma<d
\end{align*}}}
and
{\small{
\begin{align*}
&\tilde E(k-i)-\tilde E(k)=
 p^e_{k-i}\gamma V(k-i+1)+(1-p^e_{k-i})\gamma V(k-i)-
 p^e_{k}\gamma V(k+1)-(1-p^e_{k})\gamma V(k)\\ & \leq
 [p^e_{k-i}\gamma V(k-i+1)-p^e_{k-i})\gamma V(k-i)]+\gamma V(k-i)+
 [(1-p^e_{k})\gamma V(k+1)-(1-p^e_{k})\gamma V(k)]-\gamma V(k+1)\\& \leq
 \gamma dp^e_{k-i}+\gamma V(k-i)-(1-p^e_{k})d\gamma -\gamma V(k+1)\leq
 \gamma dp^e_{k-i}+(1-p^e_{k})d\gamma-\gamma i-\gamma+\gamma c<-i+c
\end{align*}}}
See that for $\gamma$ close enough to $1$ the last assertion is true.

{\it 3.  }
{\small{
\begin{align*}
& TV(k+1)-TV(k)=\tilde S(k+1)-\tilde E(k) \\
& TV(k-i)-TV(k)=\tilde S(k-i)-\tilde E(k)
\end{align*}}}

Using the proof of case {\it 1}:
{\small{
\begin{align*}
&\tilde S(k+1)-\tilde E(k)\leq\tilde S(k+1)-\tilde S(k)\leq d \\
&\tilde S(k-i)-\tilde E(k)\leq\tilde S(k-i)-\tilde S(k)\leq-i+c
\end{align*}}}

{\it 4. }
{\small{
\begin{align*}
& TV(k+1)-TV(k)=\tilde E(k+1)-\tilde S(k) \\
& TV(k-i)-TV(k)=\tilde E(k-i)-\tilde S(k)
\end{align*}}}

Using the proof of case {\it 2}:
{\small{
\begin{align*}
&\tilde E(k+1)-\tilde S(k)\leq\tilde E(k+1)-\tilde E(k)\leq d \\
&\tilde E(k-i)-\tilde S(k)\leq\tilde E(k-i)-\tilde E(k)\leq-i+c
\end{align*}}}

There are additional combinations, such as $\tilde E(k+1)-\tilde S(k)$ and
$\tilde S(k-i)-\tilde S(k)$, however their proof is straightforward using same considerations as above.
It is trivially seen that all the cases hold for the boundary conditions as well.

To see that $V(k)$ is non-decreasing in $k$ we use the following argumentation. Denote the aggregated state of having a maximal clique of size $k$ as $s_k$, $\bar s_k\in\bar S$.
Define function $g_k:s_k\to s_{k-1},\;\;k>1$, such that for each $s_k$, $g_k$ acts by deleting a random line from the maximal clique of size $k$, i.e. updating all entries of the chosen line to $0$. We aim to compare $V(s(k))=V(k)$ and $V(g_k(s(k)))$. By simple coupling argumentation one defines two processes and sees that $V(s(k))\geq V(g_k(s(k)))$. We skip the trivial details.
Finally the contraction property of operator $T$ follows from the well known results on MDP.
See~\cite{Puterman}, for example.
This accomplishes the proof of the lemma.
%where $\Delta_i^{k}=
\qed
%\end{proof}

%We have the following auxiliary lemma
\begin{lemma}\label{lem5}
For $j>2$, that is disjoint clique exists,
{\small{
\begin{align*}
& p^c_{k,i,j}=0, \;\;\;\; j>i \\
& p^c_{k,i,0}\leq p^k_{i,j}, \; \; j\leq i \\
%& p^k_{i,0}\leq p^k_{i,i}, \; \; \\
\end{align*}}}
\end{lemma}
\begin{proof}
Trivially, in case the disjoint clique is larger than $j$, the probability to have clique smaller than $j$ is zero. Therefore, the first assertion trivially holds, $p^c_{k,i,j}=0 \;\; j>i$.

Next, see that for all $i$,
{\small{
\begin{equation}\label{A1}
p^c_{k,i,0}= p^{i}(1-p)^{k-i}\binom{k}{i}.
\end{equation}}}
The sum of all transition probabilities from state $k$ acting $\bar a=1$, for all $j$ is $1$:
{\small{\[
\sum_{i=0}^kp^c_{k,i,j}=1
\]}}
Hence, the second assertion holds.
%\qed
\end{proof}
\begin{lemma}\label{lem4}
One has constants $d$ and $c$ such that
{\small{
\begin{align*}
&\tilde S(k+1)-\tilde S(k)<d \\
&\tilde S(k-i)-\tilde S(k)>-i+c
\end{align*}}}
For all $k$ and $i<k$.
\end{lemma}
\begin{proof}
We prove by finding such constants.
Substitute~\eqref{21} and~\eqref{19}, using inequalities~\eqref{23} and~\eqref{24}, and perform algebraic simplifications. Write
{\small{
\begin{align*}
&\tilde S(k)-\tilde S(k-1)=
\gamma\sum_0^{k}p^{c}_{k,i}V_i+(1-p)(k)-\gamma\sum_0^{k-1}p^c_{k-1,i}V_i+(1-p)(k-1)\\&\leq
\gamma V_0+\gamma kdp+(1-p)(k)+a_2(k)
-\gamma V_{k-1}+(1-\gamma d)(1-p)(k-1)+b_2(k-1)\\&\leq
 d_k\gamma k+d\gamma p-d\gamma-p+1+a_2(k)+b_2(k-1)+ \left( 1-k \right) \gamma+c\gamma
\leq d
\end{align*} }}
and
{\small{
\begin{align*}
&\tilde S(k-i)-\tilde S(k)\leq
\gamma V_0+\gamma pd(k-i)+(1-p)(k-i)+a_2(k-i)
-\gamma V_k+(1-\gamma d)(1-p)k+b_2(k)\\&\leq
-d\gamma ip+d+k\gamma\,k+pi+\gamma c-\gamma i-k+a_2(k-i)+b_2(k)
\leq-i+c
\end{align*}}}
%where we substituted~\eqref{14} for $V_0-V_k$.

%Note that $d$ is such that it contains the difference between the $p^{k}_{i,j}$ and $p^{k}_{i,0}$. Likewise, $a_3$, $b_3$ act similarly to $a_2(k,i)$ and $b_2(k,i)$, and in addition they consider the (even if small) variations of $d_k$ as a function of $k$.

Next, for simplicity, assume equalities for both inequalities above and write
\[
\begin{cases}
 d\gamma\,k+d\gamma\,p-d\gamma-p+1+ \left( 1-k \right) \gamma+c\gamma=d \\
-d\gamma\,ip+d+k\gamma\,k+pi+\gamma c-\gamma i-k+a_2(k-i)-b_2(k)\\
\;\;\;\;\;=-i+c
\end{cases}
\]
Solving for $d$ and $c$ we have the following expressions
{\small{
\begin{align}
& c=A( \gamma\,k+\gamma\,p-\gamma-1) b_2(k)-A(\gamma\,k+\gamma\,p-\gamma-1)a_2(k-i)\no\\
& +A\big(p( {\gamma}^{2}ik-{\gamma}^{2}i+{\gamma}^{2}k-\gamma\,ik-\gamma
\,k+i)
\big)\label{25c} \\
& d=A\gamma a_2(k-i)-A\gamma b_2(k)+A(\gamma\,ip-{\gamma}^{2}-\gamma\,k+\gamma\,p-p+1) \label{25d}
\end{align}}
}
Where $1/A={\gamma}^{2}ip+{\gamma}^{2}p-{\gamma}^{2}-\gamma\,k-\gamma\,p+1$.
Observe that $1/A\backsimeq ip-k$ as $\gamma\to1$. The rightmost part of $d$ in~\eqref{25d} is essentially independent of $i$ and $k$, and is less than $1$ for all $k,i$. Consequently, the assumption $d$ is independent of $k$ is plausible.
One the other hand, $c$ has very low positive values, comparatively to that of $i$.
Hence, the constants $d$ and $c$ above satisfy the lemma.
%\qed
\end{proof} 

%% file: ncmdp27Arx.bbl
% Generated by IEEEtran.bst, version: 1.13 (2008/09/30)
\begin{thebibliography}{10}
\providecommand{\url}[1]{#1}
\csname url@samestyle\endcsname
\providecommand{\newblock}{\relax}
\providecommand{\bibinfo}[2]{#2}
\providecommand{\BIBentrySTDinterwordspacing}{\spaceskip=0pt\relax}
\providecommand{\BIBentryALTinterwordstretchfactor}{4}
\providecommand{\BIBentryALTinterwordspacing}{\spaceskip=\fontdimen2\font plus
\BIBentryALTinterwordstretchfactor\fontdimen3\font minus
  \fontdimen4\font\relax}
\providecommand{\BIBforeignlanguage}[2]{{%
\expandafter\ifx\csname l@#1\endcsname\relax
\typeout{** WARNING: IEEEtran.bst: No hyphenation pattern has been}%
\typeout{** loaded for the language `#1'. Using the pattern for}%
\typeout{** the default language instead.}%
\else
\language=\csname l@#1\endcsname
\fi
#2}}
\providecommand{\BIBdecl}{\relax}
\BIBdecl

\bibitem{ahlswede2000network}
R.~Ahlswede, N.~Cai, S.-Y. Li, and R.~W. Yeung, ``Network information flow,''
  \emph{IEEE Transactions on Information Theory}, vol.~46, no.~4, pp.
  1204--1216, 2000.

\bibitem{sorour2012densifying}
S.~Sorour and S.~Valaee, ``On densifying coding opportunities in instantly
  decodable network coding graphs,'' in \emph{IEEE International Symposium on
  Information Theory Proceedings. (ISIT)}, 2012, pp. 2456--2460.

\bibitem{dougherty2006nonreversibility}
R.~Dougherty and K.~Zeger, ``Nonreversibility and equivalent constructions of
  multiple-unicast networks,'' \emph{IEEE Transactions on Information Theory},
  vol.~52, no.~11, pp. 5067--5077, 2006.

\bibitem{wang2012capacity}
C.-C. Wang, ``On the capacity of 1-to-broadcast packet erasure channels with
  channel output feedback,'' \emph{IEEE Transactions on Information Theory},
  vol.~58, no.~2, pp. 931--956, 2012.

\bibitem{sorour2011adaptive}
S.~Sorour and S.~Valaee, ``An adaptive network coded retransmission scheme for
  single-hop wireless multicast broadcast services,'' \emph{IEEE/ACM
  Transactions on Networking (TON)}, vol.~19, no.~3, pp. 869--878, 2011.

\bibitem{nguyen2009wireless}
D.~Nguyen, T.~Tran, T.~Nguyen, and B.~Bose, ``Wireless broadcast using network
  coding,'' \emph{IEEE Transactions on Vehicular Technology}, vol.~58, no.~2,
  pp. 914--925, 2009.

\bibitem{keller2008online}
L.~Keller, E.~Drinea, and C.~Fragouli, ``Online broadcasting with network
  coding,'' in \emph{Workshop on Network Coding, Theory and Applications.
  (NetCod.)}, Jan 2008, pp. 1--6.

\bibitem{xiao2008wireless}
X.~Xiao, Y.~Lu-Ming, W.~Wei-Ping, and Z.~Shuai, ``A wireless broadcasting
  retransmission approach based on network coding,'' in \emph{IEEE
  International Conference on Circuits and Systems for Communications.
  (ICCSC).}, 2008, pp. 782--786.

\bibitem{costa2008informed}
R.~A. Costa, D.~Munaretto, J.~Widmer, and J.~Barros, ``Informed network coding
  for minimum decoding delay,'' in \emph{IEEE International Conference on
  Mobile Ad Hoc and Sensor Systems. (MASS).}, 2008, pp. 80--91.

\bibitem{eryilmaz2006delay}
A.~Eryilmaz, A.~Ozdaglar, and M.~Medard, ``On delay performance gains from
  network coding,'' in \emph{Annual Conference on Information Sciences and
  Systems}, 2006, pp. 864--870.

\bibitem{lin2010slideor}
Y.~Lin, B.~Liang, and B.~Li, ``Slideor: Online opportunistic network coding in
  wireless mesh networks,'' in \emph{IEEE Conference on Computer Communications
  (INFOCOM)}, 2010, pp. 1--5.

\bibitem{parastoo2010optimal}
P.~Sadeghi, R.~Shams, and D.~Traskov, ``An optimal adaptive network coding
  scheme for minimizing decoding delay in broadcast erasure channels,''
  \emph{EURASIP Journal on Advances in Signal Processing}, pp. 4:1--4:14, 2010.

\bibitem{rayanchu2008loss}
S.~Rayanchu, S.~Sen, J.~Wu, S.~Banerjee, and S.~Sengupta, ``Loss-aware network
  coding for unicast wireless sessions: design, implementation, and performance
  evaluation,'' in \emph{ACM SIGMETRICS Performance Evaluation Review},
  vol.~36, no.~1, 2008, pp. 85--96.

\bibitem{katti2006xors}
S.~Katti, H.~Rahul, W.~Hu, D.~Katabi, M.~M{\'e}dard, and J.~Crowcroft, ``Xors
  in the air: practical wireless network coding,'' in \emph{ACM SIGCOMM
  Computer Communication Review}, vol.~36, no.~4, 2006, pp. 243--254.

\bibitem{nguyen2007multimedia}
D.~Nguyen, T.~Nguyen, and X.~Yang, ``Multimedia wireless transmission with
  network coding,'' in \emph{IEEE Packet Video Workshop}, 2007, pp. 326--335.

\bibitem{nguyen2009network}
D.~Nguyen and T.~Nguyen, ``Network coding-based wireless media transmission
  using {POMDP},'' in \emph{IEEE Packet Video Workshop.}, 2009, pp. 1--9.

\bibitem{georgiadis2009broadcast}
L.~Georgiadis and L.~Tassiulas, ``Broadcast erasure channel with
  feedback-capacity and algorithms,'' in \emph{Workshop on Network Coding,
  Theory, and Applications, (NetCod.)}.\hskip 1em plus 0.5em minus 0.4em\relax
  IEEE, 2009, pp. 54--61.

\bibitem{cohen2013coded}
A.~Cohen, E.~Biton, J.~Kampeas, and O.~Gurewitz, ``Coded unicast downstream
  traffic in a wireless network: analysis and wifi implementation,''
  \emph{EURASIP Journal on Advances in Signal Processing}, no.~1, pp. 1--20,
  2013.

\bibitem{dougherty2005insufficiency}
R.~Dougherty, C.~Freiling, and K.~Zeger, ``Insufficiency of linear coding in
  network information flow,'' \emph{Information Theory, IEEE Transactions on},
  vol.~51, no.~8, pp. 2745--2759, 2005.

\bibitem{birk2006coding}
Y.~Birk and T.~Kol, ``Coding on demand by an informed source (iscod) for
  efficient broadcast of different supplemental data to caching clients,''
  \emph{IEEE/ACM Transactions on Networking (TON)}, vol.~14, no.~SI, pp.
  2825--2830, 2006.

\bibitem{el2010index}
S.~El~Rouayheb, A.~Sprintson, and C.~Georghiades, ``On the index coding problem
  and its relation to network coding and matroid theory,'' \emph{IEEE
  Transactions on Information Theory}, vol.~56, no.~7, pp. 3187--3195, 2010.

\bibitem{daidata}
M.~Dai, K.~Shum, and C.~W. Sung, ``Data dissemination with side information and
  feedback,'' \emph{IEEE Transactions on Wireless Communications}, vol.~13,
  no.~9, pp. 4708--4720, 2014.

\bibitem{lucani2009network}
D.~E. Lucani, F.~H. Fitzek, M.~M{\'e}dard, and M.~Stojanovic, ``Network coding
  for data dissemination: it is not what you know, but what your neighbors
  don't know,'' in \emph{International Symposium on Modeling and Optimization
  in Mobile, Ad Hoc, and Wireless Networks. (WiOPT)}, 2009, pp. 1--8.

\bibitem{ghaderi2008reliability}
M.~Ghaderi, D.~Towsley, and J.~Kurose, ``Reliability gain of network coding in
  lossy wireless networks,'' in \emph{IEEE Conference on Computer
  Communications. (INFOCOM)}, 2008.

\bibitem{sundararajan2008arq}
J.~K. Sundararajan, D.~Shah, and M.~M{\'e}dard, ``Arq for network coding,'' in
  \emph{IEEE International Symposium on Information Theory. (ISIT)}, 2008, pp.
  1651--1655.

\bibitem{sundararajan2009feedback}
J.~K. Sundararajan, P.~Sadeghi, and M.~M{\'e}dard, ``A feedback-based adaptive
  broadcast coding scheme for reducing in-order delivery delay,'' in
  \emph{Workshop on Network Coding, Theory, and Applications. (NetCod)}, 2009,
  pp. 1--6.

\bibitem{chaudhry2011complementary}
M.~A.~R. Chaudhry, Z.~Asad, A.~Sprintson, and M.~Langberg, ``On the
  complementary index coding problem,'' in \emph{Information Theory Proceedings
  (ISIT), 2011 IEEE International Symposium on}.\hskip 1em plus 0.5em minus
  0.4em\relax IEEE, 2011, pp. 244--248.

\bibitem{li2006towards}
L.~Li, T.~J. Walsh, and M.~L. Littman, ``Towards a unified theory of state
  abstraction for {MDP}s.'' in \emph{International Symposium on Artificial
  Intelligence and Mathematics}, 2006.

\bibitem{jong2005state}
N.~K. Jong and P.~Stone, ``State abstraction discovery from irrelevant state
  variables,'' in \emph{International Joint Conference on Artificial
  Intelligence, (IJCAI)}, 2005, pp. 752--757.

\bibitem{ortner2013adaptive}
R.~Ortner, ``Adaptive aggregation for reinforcement learning in average reward
  markov decision processes,'' \emph{Annals of Operations Research}, vol. 208,
  no.~1, pp. 321--336, 2013.

\bibitem{ponsen2010abstraction}
M.~Ponsen, M.~E. Taylor, and K.~Tuyls, ``Abstraction and generalization in
  reinforcement learning: A summary and framework,'' in \emph{Adaptive and
  Learning Agents}.\hskip 1em plus 0.5em minus 0.4em\relax Springer, 2010, pp.
  1--32.

\bibitem{ferns2011bisimulation}
N.~Ferns, P.~Panangaden, and D.~Precup, ``Bisimulation metrics for continuous
  markov decision processes,'' \emph{SIAM Journal on Computing}, vol.~40,
  no.~6, pp. 1662--1714, 2011.

\bibitem{kearns2002sparse}
M.~Kearns, Y.~Mansour, and A.~Y. Ng, ``A sparse sampling algorithm for
  near-optimal planning in large markov decision processes,'' \emph{Machine
  Learning}, vol.~49, no. 2-3, pp. 193--208, 2002.

\bibitem{bertsekas2005dynamicI}
D.~P. Bertsekas, \emph{Dynamic programming and optimal control}.\hskip 1em plus
  0.5em minus 0.4em\relax Athena Scientific Belmont, MA, 1995.

\bibitem{sutton1998introduction}
R.~S. Sutton and A.~G. Barto, \emph{Introduction to reinforcement
  learning}.\hskip 1em plus 0.5em minus 0.4em\relax MIT Press, 1998.

\bibitem{strehl2008analysis}
A.~L. Strehl and M.~L. Littman, ``An analysis of model-based interval
  estimation for markov decision processes,'' \emph{Journal of Computer and
  System Sciences}, vol.~74, no.~8, pp. 1309--1331, 2008.

\bibitem{kearns2002near}
M.~Kearns and S.~Singh, ``Near-optimal reinforcement learning in polynomial
  time,'' \emph{Machine Learning}, vol.~49, no. 2-3, pp. 209--232, 2002.

\bibitem{brafman2003r}
R.~I. Brafman and M.~Tennenholtz, ``R-max-a general polynomial time algorithm
  for near-optimal reinforcement learning,'' \emph{The Journal of Machine
  Learning Research}, vol.~3, pp. 213--231, 2003.

\bibitem{mahadevan1996average}
S.~Mahadevan, ``Average reward reinforcement learning: Foundations, algorithms,
  and empirical results,'' \emph{Machine learning}, vol.~22, no. 1-3, pp.
  159--195, 1996.

\bibitem{Puterman}
M.~L. Puterman, \emph{Markov Decision Processes: Discrete Stochastic Dynamic
  Programming}, 1st~ed.\hskip 1em plus 0.5em minus 0.4em\relax New York, NY,
  USA: John Wiley \& Sons, Inc., 1994.

\bibitem{RS}
M.~C. Reed and B.~Simon, \emph{Methods of Modern Mathematical Physics:
  Functional analysis. 1}.\hskip 1em plus 0.5em minus 0.4em\relax Access Online
  via Elsevier, 1980, vol.~1.

\end{thebibliography}
